\documentclass{lmcs}
\pdfoutput=1

\usepackage{lastpage}
\lmcsdoi{18}{4}{12}
\lmcsheading{}{\pageref{LastPage}}{}{}%
{Jan.~13,~2022}{Dec.~19,~2022}{}

\usepackage[utf8]{inputenc}
\usepackage{color}

\usepackage{xspace} 

\definecolor{BlueViolet}{rgb}{0, 0, 0.55}
\definecolor{RubineRed}{rgb}{0.88, 0.07, 0.37}
\definecolor{ForestGreen}{rgb}{0.13, 0.55, 0.13}
\definecolor{Blue}{rgb}{0.0, 0.0, 1.0}
\definecolor{NavyBlue}{rgb}{0.0, 0.0, 0.5}
\definecolor{Black}{rgb}{0.02, 0.02, 0.02}
\definecolor{MidnightBlue}{rgb}{0.0, 0.2, 0.4}
\definecolor{Gray}{rgb}{0.41, 0.41, 0.41}
\definecolor{TealBlue}{rgb}{0.212,0.459,0.533}
\definecolor{Plum}{rgb}{0.6,0.25,0.6}

\usepackage{bm} 

\usepackage{amssymb}
\usepackage{graphicx}

\usepackage{mathtools}
\usepackage{url}
\usepackage{cite}

\usepackage{etoolbox}

\usepackage[inline,shortlabels]{enumitem} 
\newlist{inlinelist}{enumerate*}{1}
\setlist*[inlinelist,1]{%
  label=(\roman*),
}
\usepackage{xifthen}  
\usepackage{ifthen}

\usepackage{nicefrac}

\usepackage{hyperref}
\usepackage{cleveref}

\hypersetup{
  breaklinks   = true,
  colorlinks   = true, 
  urlcolor     = blue, 
  linkcolor    = blue, 
  citecolor    = red   
}
\hypersetup{final} 

\newcommand{\ifempty}[3]{%
  \ifthenelse{\isempty{#1}}{#2}{#3}%
}

\newcommand{\ifdots}[3]{%
  \ifthenelse{\equal{#1}{...}}{#2}{#3}%
}

\newcommand{\hidden}[1]{}

\newcommand{\mypar}[1]{\paragraph*{#1}}

\newcommand{\keyterm}[1]{\textbf{\emph{#1}}}%

\newcommand{\Real}[1]{\mathrm{Real}}

\newcommand{\Eg}{E.g.\@\xspace}
\newcommand{\eg}{e.g.\@\xspace}
\newcommand{\ie}{i.e.\@\xspace}
\newcommand{\wrt}{w.r.t.\@\xspace}

\newcommand{\emptyseq}{\varepsilon}

\renewcommand{\epsilon}{\varepsilon}

\newenvironment{proofof}[2]{%
  \subsection*{Proof of {#1}~\ref{#2}}
  \label{#2-proof}
  \def\myqed{\qed}
  \renewcommand{\qedhere}{\hfill\qed\global\def\myqed{}}
  }%
  {\myqed}

\newcommand{\BTC}{\textup{%
  \leavevmode
  \vtop{\offinterlineskip 
    \setbox0=\hbox{B}%
    \setbox2=\hbox to\wd0{\hfil\hskip-.03em
    \vrule height .3ex width .15ex\hskip .08em
    \vrule height .3ex width .15ex\hfil}
    \vbox{\copy2\box0}\box2}}\xspace}

\def\pmvColor{\color{ForestGreen}}

\newcommand{\pmvFmt}[1]{{\pmvColor{\sf{#1}}}} 

\newcommand{\PmvU}{\pmvFmt{\mathbb{A}}\xspace} 

\newcommand{\pmv}[2][]{\pmvFmt{#2}_{\pmvColor{#1}}\xspace}

\newcommand{\pmvA}[1][]{\pmv[{#1}]{A}}
\newcommand{\pmvAi}[1][]{\pmv[{#1}]{A'}}
\newcommand{\pmvB}[1][]{\pmv[{#1}]{B}}

\newcommand{\pmvM}[1][]{\pmv[{#1}]{M}} 

\newcommand{\supply}[2][]{\mathit{S}_{#1}{#2}} 

\newcommand{\txType}[1]{\mathit{type}({#1})}

\newcommand{\txTok}[1]{\mathit{tok}({#1})}

\newcommand{\txWal}[1]{\mathit{wal}({#1})}

\newcommand{\gain}[3][]{\mathit{G}_{#2}(\ifempty{#1}{}{{#1},}{#3})}

\def\txColor{\color{MidnightBlue}}

\newcommand{\txFmt}[1]{{\txColor{\sf #1}}}

\newcommand{\tx}[2][]{\txFmt{#2}_{\txColor{#1}}}
\newcommand{\txT}[1][]{\tx[#1]{T}}

\newcommand{\txTi}[1][]{\txFmt{T'_{\txColor{{#1}}}}}

\DeclareMathAlphabet{\mathbfsf}{\encodingdefault}{\sfdefault}{bx}{n}

\newcommand{\bcB}[1][]{{\txColor{\lambda_{#1}}}}
\newcommand{\bcBi}[1][]{{\txColor{\lambda'_{#1}}}}

\newcommand{\irule}[2]{\dfrac{#1}{#2}}

\newcommand{\nrule}[1]{{\scriptsize \textsc{#1}}}

\newcommand{\dom}[1]{\operatorname{dom} {#1}}

\newcommand{\RNN}{\mathbb{R}_{\geq 0}}
\newcommand{\RP}{\mathbb{R}_{> 0}}

\newcommand{\bind}[2]{\nicefrac{#2}{#1}}
\newcommand{\setenum}[1]{\{#1\}}

\newcommand{\qedex}{\ensuremath{\diamond}}

\crefname{appendix}{appendix}{appendices}
\Crefname{appendix}{Appendix}{Appendices}
\crefname{notation}{notation}{notations}
\Crefname{notation}{Notation}{Notations}

\definecolor{LightGrey}{rgb}{0.95,0.95,0.95}
\definecolor{keyword}{HTML}{7F0055}

\def\tokColor{\color{magenta}}
\newcommand{\tokFmt}[1]{{\tokColor{#1}}}

\newcommand{\tokM}[2]{\setenum{{#1},{#2}}}

\newcommand{\tokT}[1][]{\tokFmt{\tau_{#1}}}
\newcommand{\tokTi}[1][]{\tokFmt{\tau'_{#1}}}

\newcommand{\TokU}[1][]{\tokFmt{\mathbb{T}_{#1}}} 

\newcommand{\ltsLabel}[1][]{{\txT_{#1}}}
\newcommand{\ltsLabeli}[1][]{{\txTi_{#1}}}

\newcommand{\ammDepositOp}{{\txColor{\sf dep}}}
\newcommand{\ammSwapOp}{{\txColor{\sf swap}}}
\newcommand{\ammRedeemOp}{{\txColor{\sf rdm}}}
\newcommand{\ammSwapParOp}[1]{{\txColor{\sf swap}^{#1}}}

\newcommand{\actAmmDeposit}[5]{\ifempty{#1}{}{{#1}:}\ammDepositOp({#2}:{#3},{#4:}{#5})}

\newcommand{\actAmmRedeem}[2]{\ifempty{#1}{}{{#1}:}\ammRedeemOp({#2})}

\newcommand{\actAmmSwapExact}[5]{\ifempty{#1}{}{{#1}:}\ammSwapParOp{#2}({#3},{#4},{#5})}

\newcommand{\actAmmSwapGuarded}[6]{\ifempty{#1}{}{{#1}:}\ammSwapParOp{#2}({#3}:{#4},{#5}:{#6})}
\newcommand{\actAmmDepositGuarded}[5]{\ifempty{#1}{}{{#1}:}\ammDepositOp({#2}:{#3},{#4}:{#5})}

\newcommand{\tokBal}[1][]{\sigma_{#1}} 
\newcommand{\tokBali}[1][]{\sigma'_{#1}} 
\newcommand{\wal}[2]{{#1}[{#2}]}
\newcommand{\walA}[2][]{\wal{\pmvA[#1]}{#2}}
\newcommand{\walB}[2][]{\wal{\pmvB[#1]}{#2}}

\newcommand{\val}[2][]{#2_{#1}} 
\newcommand{\valV}[1][]{\val[#1]{v}}
\newcommand{\valVi}[1][]{\val[#1]{v'}}

\newcommand{\valSXa}{\alpha}

\newcommand{\valSXb}{\beta}

\newcommand{\valSXc}{\gamma}

\newcommand{\X}[2][]{\ifempty{#1}{X}{X_{#1}}\ifempty{#2}{}{({#2})}}

\newcommand{\exchO}[2][]{\ifempty{#1}{P}{P_{#1}}\ifempty{#2}{}{{#2}}}

\newcommand{\SX}[2][]{\mathit{SX}_{#1}\ifempty{#2}{}{({#2})}}

\newcommand{\SL}[2][]{\ifempty{#1}{\Delta X}{\Delta X_{#1}}\ifempty{#2}{}{({#2})}}

\newcommand{\RX}[4]{\mathit{RX}^{#1}_{#2}\ifempty{#3}{}{(#3,#4)}}

\newcommand{\confG}[1][]{\Gamma_{#1}}
\newcommand{\confGi}[1][]{\Gamma'_{#1}}
\newcommand{\confGii}[1][]{\Gamma''_{#1}}

\newcommand{\confD}[1][]{\Delta_{#1}}
\newcommand{\confDi}[1][]{\Delta'_{#1}}
\newcommand{\confDii}[1][]{\Delta''_{#1}}

\newcommand{\amm}[2]{\setenum{{#1},{#2}}} 
\newcommand{\bigamm}[2]{\Big\{ {#1},\, {#2} \Big\}} 
\newcommand{\ammR}[1][]{r_{#1}} 
\newcommand{\ammRi}[1][]{r'_{#1}}

\newcommand{\fee}{\phi}

\newcommand{\githubrepo}{https://github.com/blockchain-unica/defi-workbench}

\newtoggle{arxiv}
\toggletrue{arxiv}

\makeatletter
\renewcommand\paragraph{\@startsection{paragraph}{4}{\z@}%
	{2.25ex \@plus 1ex \@minus .2ex}%
	{-0.75em}%
	{\normalfont\normalsize\bfseries}}
\makeatother

\begin{document}

\title{A theory of Automated Market Makers in DeFi}

\author[M.~Bartoletti]{Massimo Bartoletti\lmcsorcid{0000-0003-3796-9774}}[a]
\author[J. Chiang]{James Hsin-yu Chiang\lmcsorcid{0000-0002-5126-9494}}[b]
\author[A. {Lluch-Lafuente}]{Alberto {Lluch-Lafuente}\lmcsorcid{0000-0001-7405-0818}}[b]

\address{University of Cagliari, Cagliari, Italy}
\email{bart@unica.it}

\address{Technical University of Denmark, DTU Compute, Copenhagen, Denmark}
\email{jchi@dtu.dk, albl@dtu.dk}

\begin{abstract}
  Automated market makers (AMMs) are one of the most prominent decentralized 
  finance (DeFi) applications. 
  AMMs allow users to trade different types of crypto-tokens, 
  without the need to find a counter-party. 
  There are several implementations and models for AMMs, 
  featuring a variety of sophisticated economic mechanisms. 
  We present a theory of AMMs.
  The core of our theory is an abstract operational model 
  of the interactions between users and AMMs,
  which can be concretised by instantiating the economic mechanisms.
  We exploit our theory to formally prove a set of fundamental properties
  of AMMs, characterizing both structural and economic aspects. 
  We do this by abstracting from the actual economic mechanisms used in
  implementations, and identifying sufficient conditions 
  which ensure the relevant properties.
  Notably, we devise a general solution to the 
  \emph{arbitrage problem}, 
  the main game-theoretic foundation behind the economic mechanisms of AMMs. 
\end{abstract}

\maketitle

\section{Introduction}
\label{sec:intro}

\newcommand{\uniswapsupply}{\$3B\xspace}
\newcommand{\curvesupply}{\$4B\xspace}
\newcommand{\uniswapvolumedaily}{\$1B\xspace}
\newcommand{\curvevolumedaily}{\$250M\xspace}
\newcommand{\statsdate}{December 2022\xspace}
\newcommand{\defiincidentstot}{\$2.4B\xspace}

Decentralized finance (DeFi) is a software infrastructure,
based on blockchains and smart contracts,
which allows users to create and trade crypto-tokens without the 
intermediation of central authorities, unlike traditional finance
\cite{werner2021sok,Qin21CefiDefi}.
\emph{Automated Market Makers (AMMs)}
are one of the main DeFi archetypes:
roughly, AMMs are decentralized markets of crypto-tokens,
providing users with three core operations: 
depositing crypto-tokens to obtain shares in an AMM;
the dual operation of redeeming shares in the AMM for the underlying tokens; 
and 
swapping tokens of a given type for tokens of another type.
The amount of tokens received by a user upon a swap is
algorithmically determined by the AMM:
roughly, this is the amount of tokens sent from the user to the AMM,
times the \emph{swap rate}, which is computed by the AMM based on
its internal state and the input amount.

Despite the apparent simplicity of these operations,
AMMs manifest an emerging behaviour, 
where users are incentivized to swap tokens 
to keep their swap rates
aligned with the \emph{exchange rate},
\ie the ratio between the prices of the exchanged tokens
given by external price oracles.
Namely, if an AMM offers a better swap rate than the oracles' exchange rate,
rational users will perform swaps to narrow the gap.
Formally, the optimal strategy can be seen as the solution of
a game, called the \emph{arbitrage game}.
Executing the optimal strategy
closes the gap between AMM's and oracles' exchange rates,
and in this sense AMMs offer users exhange rates
that align towards the external, global exchange rates.%
  
As of \statsdate, the two AMM platforms leading by user activity, 
Uniswap and Curve Finance, alone
hold \uniswapsupply and \curvesupply worth of tokens,
and process \uniswapvolumedaily and \curvevolumedaily worth of transactions daily
\cite{uniswapstats,curvestats}.
Although this massive adoption could suggest that AMMs are a
consolidated, well-understood technology, in practice 
their economic mechanisms are inherently hard to design and implement.
For instance, interactions with AMMs are sensitive to
\emph{transaction ordering attacks},
where actors with the power to influence the
order of transactions in the blockchain can profit from 
an opportunistic behaviour, causing detriment to other users.
The relevance of attacks to AMMs is witnessed by the proliferation 
of scientific literature on the topic
\cite{BCL22fc,Zhou21high,Daian19flash,Eskandari19sok,Qin21quantifying}.
Still, attacks to DeFi applications are not purely theoretical: indeed,
there is a growing history of DeFi incidents,
which have caused losses exceeding \defiincidentstot~\cite{defihacks} so far.
These issues witness a need for foundational work to devise
formal theories of AMMs which allow the understanding of their 
structural properties and of their economic incentive mechanisms.

Current descriptions of AMMs are either 
economic models~\cite{Angeris21analysis,Angeris20aft,Evans21fc,angeris2020does}, 
which focus on the incentive mechanism alone,
or concrete AMM implementations.
While economic models are useful to understand 
the macroscopic financial aspects of AMMs,
they do not precisely describe the interactions between AMMs and their users.
Still, a precise formalisation of these interactions is fundamental to understand the structural and economic properties of AMMs, and to determine possible deviations from safe behaviour. 
Implementations, instead, reflect the exact behaviour of AMMs, but at a level of detail that hampers high-level understanding and reasoning. 
Moreover, the rich variety of implementations, proposals and models for AMMs, each featuring different 
economic mechanisms, makes it difficult to compare AMM designs or to provide a clear contour for the space of possible ``well behaving'' designs.

\subsection{Contributions}

In this paper we exploit techniques from concurrency theory to provide a
formal backbone for AMMs and to study their fundamental properties. 
More specifically, our
main contributions can be summarised as follows:
\begin{enumerate}

\item We introduce a formal model of AMMs (\autoref{sec:amm}), 
which distils the common features of leading AMM implementations 
like Uniswap \cite{uniswapimpl}, Curve \cite{curveimpl}, 
and Balancer \cite{balancerpaper}. 
The core of our model is a transition system that describes
the evolution of AMM states resulting from the interaction between users and AMMs.
A peculiar aspect of our model is that it abstracts from 
the \emph{swap rate function}, a key economic mechanism of AMMs,
which is used to determine the exchange rates between tokens. 

\item Building upon our model, in \autoref{sec:econom-defs}
  we define basic economic notions
  like token prices, exchange rates, slippage, net worth, and gain.
  We compute the gain 
  resulting from swap actions 
  (Lemma~\ref{lem:swap:gain}), and
  we establish a key relation 
  between the gain of swap actions, the swap rate and the exchange rate:
  a swap action has a strictly positive gain \emph{if and only if} 
  the swap rate is strictly greater than the exchange rate between 
  the swapped tokens
  (Lemma~\ref{lem:swap:gain-SX-X}).
  Both lemmata are instrumental to prove many subsequent results.

\item In \autoref{sec:struct-properties} we establish
  a set of structural properties of AMMs.
  In particular, we establish preservation results for the 
  supply of tokens (Lemma~\ref{lem:supply:const}) and
  for the global net worth (Lemma~\ref{lem:amm:W-preservation}).
  We show that
  assets cannot be frozen within AMMs, \ie users can always extract
  any amount of the token reserves deposited in AMMs
  (Lemma~\ref{lem:amm:liquidity}).
  In Lemma~\ref{lem:tx-concurrent} we investigate when transactions
  can be reordered without affecting the resulting state.
  In Theorems~\ref{thm:additivity} and~\ref{thm:reversibility}
  we study compositionality of deposit and redeem transactions:
  in particular, we establish that two deposit actions on the same AMM 
  can be merged in a single action (and similarly for two redeems),
  and that the effect of deposits and redeem actions can be reverted
  by suitable transactions.
  Remarkably, all the structural properties in \autoref{sec:struct-properties}
  do not depend on the choice of the swap rate function.

\item In \autoref{sec:swap-rate} we devise sufficient conditions 
  on swap rate functions that induce good behavioural properties of AMMs.
  These conditions allow us to extend to swap actions
  the additivity and reversibility properties enjoyed by deposit and redeem
  actions (Theorems~\ref{thm:sr-additivity} and~\ref{thm:sr-reversibility}),
  as well as to compute the gain of composed and reversed swaps
  (Lemmata~\ref{lem:swap-gain:additivity} and \ref{lem:swap-gain:reversibility}).
  We study the effect of deposits and redeems on the swap rate
  and on the internal exchange rate (Lemma~\ref{lem:sr-reduced-slippage}).
  We then study the properties of three notable swap rate functions:
  the constant sum, the constant product, and the constant mean.

\item In \autoref{sec:arbitrage} we investigate the incentive mechanism of AMMs.
  We start by considering the \emph{arbitrage problem},
  which requires to find the action which maximizes the gain of a user.
  Performing such optimal action has the side effect of aligning the internal exchange rate of the AMM to 
  the external exchange rate given by token price oracles.
  This gives one of the landmark economic properties of AMMs: assuming rational users, AMMs can be seen as price oracles themselves~\cite{Angeris20aft}.
  Notably, while solutions to the arbitrage problem are already
  known for specific swap rate functions, in Theorem~\ref{thm:arbitrage}
  we generalize the result to any swap rate function respecting the
  conditions given in \autoref{sec:swap-rate}.
  We then show that depositing tokens into AMMs incentivizes subsequent swaps
  (Theorem~\ref{thm:swap-after-dep}), 
  while redeeming tokens disincentivizes them
  (Theorem~\ref{thm:swap-after-rdm}).
  Finally, in Theorems~\ref{thm:dep-arbitrage} and~\ref{thm:rdm-arbitrage}
  we relate the solution of the arbitrage problem in the states 
  before and after a deposit or redeem action, and we compare their gains.

\item In \autoref{sec:mev} we discuss Maximal Extractable Value (MEV),
  a class of attacks where miners exploit their
  power of dropping and reordering user transactions (and inserting their own)
  to increase their gain to the detriment of users.
  These attacks are one of the most carefully studied AMM phenomena, 
  occuring widely in practice and frequently making up the bulk of interactions with AMMs~\cite{Qin21quantifying}.
  The fact that our AMM model can
  accurately express these attacks
  supports the coherence of our modelling choices with respect to behaviour exhibited by actual AMM implementations. 
  
\item In \autoref{sec:variants} we discuss some extensions to our basic 
  AMM model to make it closer to the implementation of Uniswap~\cite{uniswapimpl},
  and their impact on the results in the paper.

\item As a byproduct, we provide an open-source Ocaml implementation of our
  executable semantics as a companion of this paper.\footnote{\url{\githubrepo}}

\item We provide full proofs of all our statements in the Appendices.
  
\end{enumerate}

  \subsection{Related Work}
\label{sec:related}

The work~\cite{Angeris21analysis} proposed one of the first analyses 
of the incentive mechanism of Uniswap.
This analysis was then generalised in~\cite{Angeris20aft} to 
\emph{constant function AMMs (CFMMs)}, 
where, for a pair of token types, the reserves $\ammR[0],\ammR[1]$ before a swap
and the reserves $\ammRi[0],\ammRi[1]$ after the swap
must preserve the invariant $f(\ammR[0],\ammR[1]) = f(\ammRi[0],\ammRi[1])$, 
for a given trading function $f$.
Constant product AMMs, like Uniswap, are an instance of CFMMs, 
where $f(x,y) = x y$.
Both works study the arbitrage problem, for constant product AMMs
and CFMMs, respectively.
The two works show that the solution can be efficiently computed, 
and suggest that constant product AMMs accurately report exchange rates. 
Our work and~\cite{Angeris20aft} share a common goal, 
\ie a theory of AMMs generalizing that of constant product AMMs.
However, the two approaches are quite different.
The work \cite{Angeris20aft} considers a class of AMMs,
\ie CFMMs with a convex trading set, 
and studies the properties enjoyed by AMMs under these assumptions.
Instead, in this paper we devise a minimal set of properties
of the swap rate function which induce good behavioural properties of AMMs.
Notably, we find conditions on the swap rate function 
which ensure that a given swap action maximizes the gain of the player
(Theorem~\ref{thm:arbitrage}). 
Another difference is that the AMM model in~\cite{Angeris20aft} 
describes the evolution of a single AMM, abstracting away the other
components of the state (\ie the users and the other AMMs);
instead, we model AMMs as \emph{reactive systems},
borrowing techniques from concurrency theory.
While the approach followed by~\cite{Angeris20aft} is still adequate 
to study problems that concern AMMs in isolation (\eg, arbitrage),
viewing AMMs as reactive systems allows us to study what happens when
many agents (users and AMMs) can interact. 
\Eg, we are able to reason about Maximal Extractable Value
(\autoref{sec:mev}).

The work~\cite{Danos21defi} generalises the arbitrage problem 
to the setting where a swap between two token types $\tokT[0]$ and $\tokT[n]$ 
can be obtained through a sequence of $n$ intermediate
swaps between $\tokT[i]$ and $\tokT[i+1]$, for $0 \leq i < n$.
In practice, this represents the situation where users can interact
with different AMM platforms, each one providing its own set of token pairs.
To model this scenario, \cite{Danos21defi} introduces \emph{exchange networks},
\ie multi-graphs where nodes are tokens, 
and edges are AMMs which allow users to swap the two endpoint tokens.
To encompass different AMM platforms, each edge has its own price function, 
which determines how many output tokens
are paid for a given amount of input tokens.
The authors show that, under some conditions on the price functions
(\ie, monotonicity, continuity, boundedness and concavity),
the arbitrage problem always admits a non-trivial solution.
In the special case of constant product AMMs, a closed formula
for the solution is provided.
Besides arbitrage, \cite{Danos21defi}
also considers the \emph{optimal routing problem}, 
\ie finding a strategy to maximize the amount of tokens $\tokT[1]$
received for at most a given amount of tokens $\tokT[0]$.
Under the same assumptions on the price function used for the arbitrage problem,
the optimal routing problem admits a solution.
There are several differences between our approach and that 
of \cite{Danos21defi}, besides the fact that we assume 
the same swap rate function for all AMMs, and a graph instead of a multi-graph
(\ie, we admit at most one AMM for each token pair).
A technical difference is that we assume that the amount $y$ of output 
tokens received for an amount $x$ of input tokens is given by
$y = \SX{x,\ammR[0],\ammR[1]} \cdot x$, whereas \cite{Danos21defi}
defines this amount as $y = f_{\ammR[0],\ammR[1]}(x)$.
This results in different structural properties for $\SX{x,\ammR[0],\ammR[1]}$
and $f(x)$ in order to achieve the desired behavioural properties of AMMs.
Having the AMM reserves $\ammR[0]$, $\ammR[1]$ as parameters of
our swap rate functions $\SX{}$ has a benefit, in that we can express
conditions which relate states before and after a transaction:
this is what happens, \eg, in the additivity, reversibility and homogeneity
properties (Definitions~\ref{def:sr-additivity},~\ref{def:sr-reversibility}
and~\ref{def:sr-homogeneity}).
As a consequence of this choice, compared to~\cite{Danos21defi} 
our theory encompasses also deposit and redeem actions, 
providing results that clarify how these actions interfere with swaps
(\eg, Theorems~\ref{thm:swap-after-dep}, \ref{thm:swap-after-rdm}, \ref{thm:dep-arbitrage}, and \ref{thm:rdm-arbitrage}).

A few alternatives to constant product AMMs have been studied. 
Balancer~\cite{balancerpaper} generalizes the constant product function
used by Uniswap to a constant (weighted geometric) mean
$f(\ammR[1],\cdots,\ammR[n]) = \prod_{i=1}^n \ammR[i]^{w_i}$, 
where the weight $w_i$ reflects the relevance of a token $\tokT[i]$
in a tuple of tokens $(\tokT[1],\cdots,\tokT[n])$.
This still fits within the CFMM setting of~\cite{Angeris20aft},
thus inheriting its results about solvability 
of the arbitrage problem~\cite{Evans21fc}.
Curve~\cite{curvepaper} features 
a hybrid of a constant sum and constant product function,
optimized for large swap volumes between \emph{stable coins}, 
where the swap rate can support large amounts with small sensitivity. 
To efficiently compute swap rates, implementations
perform numerical approximations~\cite{curvegetd}. 
Should these approximations fail to converge, these implementations
still guarantee that the AMM remains liquid. 
The work~\cite{krishnamachari2021dynamic} proposes a 
constant product invariant that is adjusted dynamically 
based on the oracle price feed, 
thus reducing the need for arbitrage transactions,
but at the cost of lower fee accrual.
AMMs with \emph{virtual} balances have been proposed~\cite{virtualAMMbalances} 
and implemented~\cite{moonipaper,mooniswapimpl}.
In these AMMs, the swap rate depends on past actions, 
besides the current funds balances in the AMM.
This, similarly to~\cite{krishnamachari2021dynamic}, aims to minimize 
the need for arbitrage transactions to ensure 
the local AMM swap rate tends towards the exchange rates. 
Establishing whether these sophisticated swap rate functions enjoy
the properties in~\autoref{sec:swap-rate} is an interesting open problem.

AMMs are well-known to suffer from transaction-ordering attacks,
through which an adversary with the power of influencing 
the order of transactions (\eg, a miner) 
can extract value from user transactions.
For instance, if the transaction pool contains
a swap transaction sent by user $\pmvA$, then a miner $\pmvM$ can extract value
from $\pmvA$'s swap through a transaction ``sandwich'' constructed as follows.
First, $\pmvM$ front-runs $\pmvA$'s swap with its own swap, 
crafted so that $\pmvA$'s swap decreases $\pmvA$'s net worth as much as possible.
Then, $\pmvM$ closes the sandwich by appending another swap transaction
which maximizes $\pmvM$'s gain, and finalises the whole sandwich
on the blockchain.
In this way, $\pmvA$ will always have a negative gain,
which is counterbalanced by a positive gain of $\pmvM$.
This and other kinds of attacks have fostered 
the research on adversarial and defensive  strategies, 
and on empirical analyses of the impact of attacks
\cite{BCL22fc,Chitra22fc,Zhou21high,Qin21quantifying,Daian19flash,Eskandari19sok}.
For instance, the work~\cite{BCL22fc} devises an optimal strategy through 
which an adversary can extract the maximal value from users' transactions
(not only swaps, but also deposits and redeems),
in the setting of Uniswap-like AMMs.
The swap-rate-agnostic approach pursued by this paper could be exploited to
generalise the attack of~\cite{BCL22fc} to AMMs beyond Uniswap.

A high-level survey on various AMM protocols is in~\cite{Xu21sok}.

\mypar{Comparison with previous work}

A preliminary version of this work was presented at 
COORDINATION 2021~\cite{BCL21coordination}.
The current version substantially extends it, 
streamlining the theory and providing additional results.
A crucial difference between the two papers is that, while in
\cite{BCL21coordination} the semantics of swap actions was parameterized
by an invariant between the old and the new token reserves, 
here we make the semantics parametric \wrt the swap rate function $\SX{}$.
This leads to a substantial simplification of the conditions 
that are put to obtain nice behavioural properties of swaps,
and consequently of the corresponding proofs.
Among the new results \wrt \cite{BCL21coordination}, we mention
in particular the additivity and reversibility properties
(Theorems~\ref{thm:additivity}, \ref{thm:reversibility}, \ref{thm:sr-additivity}, and \ref{thm:sr-reversibility}),
and the results that relate the gain of swaps before and after 
deposit/redeem actions (Theorems \ref{thm:swap-after-dep}, \ref{thm:swap-after-rdm}, \ref{thm:dep-arbitrage}, and \ref{thm:rdm-arbitrage}).
Besides these extensions, the current paper includes a discussion
of the constant sum and of the constant mean swap rate functions,
a new section on MEV attacks (see \autoref{sec:mev}),
and it provides detailed proofs for all its statements.

\section{A formal model of Automated Market Makers} 
\label{sec:amm}

We introduce a formal, operational model of AMMs,
which focusses on the common operations implemented by AMM platforms.
In order to simplify the resulting theory, our model abstracts from
a few features that are often found in AMM implementations, 
like \eg fees, price updates, and guarded transactions.
We discuss in~\S\ref{sec:variants} how to extend our model 
to make it closer to the Uniswap protocol~\cite{uniswapimpl}.

We introduce here some general notation.
We denote by $f x$ the application of a function $f$ to a value $x$
(we use parentheses, \eg $f(x)$, to resolve ambiguities).
We denote with $\dom{f}$ the domain of $f$.
We use the standard notation $f \setenum{\bind{x}{v}}$
to update a partial map $f$ at point $x$:
namely, $f \setenum{\bind{x}{v}}(x) = v$, while
$f \setenum{\bind{x}{v}}(y) = f y$ for $y \neq x$.

\subsection{AMM basics}

\mypar{Tokens}

We assume a set  $\TokU[0]$ of \keyterm{atomic token types},
which represent native cryptocurrencies and application-specific tokens.
For instance, $\TokU[0]$ may include ETH, the native cryptocurrency
of Ethereum, and WBTC, \ie
Bitcoins wrapped with the ERC20 interface for Ethereum tokens.
A \keyterm{minted token type} is an unordered pair of distinct 
atomic token types:
if $\tokT[0]$ and $\tokT[1]$ are atomic token types 
and $\tokT[0] \neq \tokT[1]$,
then the minted token type $\tokM{\tokT[0]}{\tokT[1]}$ 
represents shares in an AMM holding reserves of $\tokT[0]$ and $\tokT[1]$.
We denote by $\TokU[1]$ the set of minted token types.
In our model, tokens are \emph{fungible},
\ie individual units of the same type are interchangeable.
This means that amounts of tokens of the same type
can be split into smaller parts,
and two amounts of tokens of the same type can be joined.
We use $\valV, \valVi, \ammR, \ammRi, x, x'$ to range over
nonnegative real numbers ($\RNN$).
We write $\TokU$ for the universe of all token types, 
\ie $\TokU = \TokU[0] \cup \TokU[1]$,
and we use $\tokT,\tokTi,\ldots$ to range over $\TokU$.
We write \mbox{$\ammR:\tokT$} to denote $\ammR$ units
of a token of type $\tokT$, either atomic or minted.

\mypar{Wallets and AMMs}

We assume a set of \keyterm{users} $\PmvU$,
ranged over by $\pmvA, \pmvAi, \ldots$
We model the \keyterm{wallet} of a user $\pmvA$ as a term
$\wal{\pmvA}{\tokBal}$, where the finite partial map
$\tokBal \in \TokU \rightharpoonup \RNN$
represents $\pmvA$'s token balance.
We model an \keyterm{AMM} holding reserves of \mbox{$\ammR[0]:\tokT[0]$} and
\mbox{$\ammR[1]:\tokT[1]$} (with $\tokT[0] \neq \tokT[1]$) 
as an unordered pair \mbox{$\amm{\ammR[0]:\tokT[0]}{\ammR[1]:\tokT[1]}$}.
Since the order of the token reserves in an AMM is immaterial, 
the terms \mbox{$\amm{\ammR[0]:\tokT[0]}{\ammR[1]:\tokT[1]}$}
and \mbox{$\amm{\ammR[1]:\tokT[1]}{\ammR[0]:\tokT[0]}$}
denote exactly the same AMM.

\mypar{States}

We model the interaction between users and AMMs
as a labelled transition system (LTS).
Its labels represent blockchain \keyterm{transactions}, while
the \keyterm{states} $\confG, \confGi, \confD, \ldots$ are 
finite non-empty compositions of wallets and AMMs.
Formally, states are terms of the form:
\[
\wal{\pmvA[1]}{\tokBal[1]} \mid \cdots \mid \wal{\pmvA[n]}{\tokBal[n]}
\mid
\amm{\ammR[1]:\tokT[1]}{\ammRi[1]:\tokTi[1]}
\mid \cdots \mid
\amm{\ammR[k]:\tokT[k]}{\ammRi[k]:\tokTi[k]}
\]
and subject to the following conditions. For all $i \neq j$:
\begin{enumerate}
\item $\pmvA[i] \neq \pmvA[j]$ 
  (each user has a single wallet);
\item
$\setenum{\tokT[i],\tokTi[i]} \neq \setenum{\tokT[j],\tokTi[j]}$ 
(distinct AMMs cannot hold exactly the same token types).
\end{enumerate}

Note that these conditions allow AMMs 
to have a common token type $\tokT$,
\eg as in \mbox{$\amm{\ammR[1]:\tokT[1]}{\ammR:\tokT}$},
\mbox{$\amm{\ammRi:\tokT}{\ammR[2]:\tokT[2]}$}, 
thus enabling indirect trades between token pairs
not directly provided by any AMM.
A state is \emph{initial}
when it has no AMMs, and its wallets hold only atomic tokens.
We stipulate that the ordering of terms in a state is immaterial.
Hence, we consider two states $\confG$ and $\confGi$ to be equivalent 
when they contain the same terms (regardless of their order). 
For a term $Q$ and a state $\confG$, we write $Q \in \confG$
when $\confG = Q \mid \confGi$, for some $\confGi$.

\mypar{Transactions}

State transitions are triggered by transactions
 $\ltsLabel, \ltsLabeli, \ldots$, which can have the following forms
(where $\tokT[0]$ and $\tokT[1]$ are atomic tokens):
\begin{itemize}

\item$\actAmmDeposit{\pmvA}{\valV[0]}{\tokT[0]}{\valV[1]}{\tokT[1]}$.
  $\pmvA$ deposits $\valV[0]:\tokT[0]$ and $\valV[1]:\tokT[1]$
  to an AMM $\amm{\ammR[0]:\tokT[0]}{\ammR[1]:\tokT[1]}$,
  receiving in return some freshly-minted units of the token $\tokM{\tokT[0]}{\tokT[1]}$;

\item $\actAmmSwapExact{\pmvA}{}{\valV}{\tokT[0]}{\tokT[1]}$.
  $\pmvA$ tranfers $\valV:\tokT[0]$ to an AMM
  $\amm{\ammR[0]:\tokT[0]}{\ammR[1]:\tokT[1]}$,
  receiving in return some units of $\tokT[1]$, which are removed from the AMM;

\item $\actAmmRedeem{\pmvA}{\valV:\tokM{\tokT[0]}{\tokT[1]}}$.
  $\pmvA$ redeems $\valV$ units of the minted token $\tokM{\tokT[0]}{\tokT[1]}$: 
  this means that some units of $\tokT[0]$ and $\tokT[1]$ are transferred
  from the AMM $\amm{\ammR[0]:\tokT[0]}{\ammR[1]:\tokT[1]}$ to $\pmvA$'s wallet,
  and that $\valV$ units of $\tokM{\tokT[0]}{\tokT[1]}$ are burned.

\end{itemize}
We denote with
$\txType{\txT}$ the type of $\txT$
(\ie, $\ammDepositOp$, $\ammSwapOp$, or $\ammRedeemOp$),
with $\txWal{\txT}$ the user whose wallet is affected by $\txT$,
and with $\txTok{\txT}$ the set of token types affected by $\txT$.
For example, if
$\txT = \actAmmSwapExact{\pmvA}{}{\valV}{\tokT[0]}{\tokT[1]}$,
then
$\txType{\txT} = \ammSwapOp$,
$\txWal{\txT} = \pmvA$, and
$\txTok{\txT} = \setenum{\tokT[0],\tokT[1]}$.

\mypar{Token supply}

We use $\supply[\confG]{\tokT}$ to denote the \keyterm{supply} of a token type $\tokT$ in a state $\confG$, defined 
as the sum of the reserves of $\tokT$ in all the wallets and the AMMs
in $\confG$.
Formally, we define $\supply[\confG]{\tokT}$
by induction on the structure of states as follows:
\begin{equation*}
  \supply[\walA{\tokBal}]{\tokT} = 
  \begin{cases}
    \tokBal \tokT & \text{if $\tokT \in \dom{\tokBal}$} \\
    0 &\text{otherwise }
  \end{cases}
  \qquad
  \supply[\amm{\ammR[0]:\tokT[0]}{\ammR[1]:\tokT[1]}]{\tokT} =
  \begin{cases}
    \ammR[i] & \text{if $\tokT = \tokT[i]$} \\
    0 & \text{otherwise}
  \end{cases}
  \qquad
  \supply[\confG \mid \confGi]{\tokT} = \supply[\confG]{\tokT} + \supply[\confGi]{\tokT}
\end{equation*}
For example, let
\(
\confG
=
\walA{1:\tokT[0], 2:\tokM{\tokT[0]}{\tokT[1]}} \mid 
\amm{3:\tokT[0]}{4:\tokT[1]} \mid 
\amm{5:\tokT[0]}{6:\tokT[2]}
\).
We have that
$\supply[\confG]{\tokT[0]} = 9$, 
$\supply[\confG]{\tokT[1]} = 4$,
$\supply[\confG]{\tokT[2]} = 6$,
while $\supply[\confG]{\tokT} = 0$ for 
$\tokT \not\in \setenum{\tokT[0],\tokT[1],\tokT[2]}$.
%
Note that $\supply[\confG]{\tokT}$ is always defined, since 
it is defined when $\confG$ is an atomic term (wallet or AMM),
and states $\confG$ are \emph{finite} compositions of atomic terms.

\subsection{AMM semantics}

We now formalise the transition rules between states.
We write $\confG \xrightarrow{\txT} \confGi$ for a state transition
from $\confG$ to $\confGi$, triggered by a transaction $\txT$.
When $\confG \xrightarrow{\txT} \confGi$ for some $\confGi$, we
say that $\txT$ is \emph{enabled} in $\confG$.
We denote with $\xrightarrow{}^*$ the reflexive and transitive
closure of $\xrightarrow{}$. 
Given a finite sequence of transactions $\bcB = \txT[1] \cdots \txT[k]$,
we write $\confG \xrightarrow{\bcB} \confGi$
when $\confG \xrightarrow{\txT[1]} \cdots \xrightarrow{\txT[k]} \confGi$,
and in this case we say that $\bcB$ is enabled in $\confG$.
We say that a state $\confG$ is \emph{reachable}
if $\confG[0] \xrightarrow{}^* \confG$ for some initial $\confG[0]$.
Hereafter, all the states mentioned in our results are
implicitly assumed to be reachable.
Given a partial map $\tokBal \in \TokU \rightharpoonup \RNN$,
a token type $\tokT \in \TokU$ and a partial operation
$\circ \in \RNN \times \RNN \rightharpoonup \RNN$
with $\circ \in \setenum{+,-}$,
we define the partial map $\tokBal \circ \valV:\tokT$ as follows:
\begin{equation*} \label{eq: balance update}
  \tokBal \circ \valV:\tokT = \begin{cases}
    \tokBal\setenum{\bind{\tokT}{(\tokBal \tokT) \;\circ\; \valV}}
    & \text{if $\tokT \in \dom{\tokBal}$ and $(\tokBal \tokT) \circ \valV \in \RNN$}
    \\
    \tokBal\setenum{\bind{\tokT}{\valV}}
    & \text{if $\tokT \not\in \dom{\tokBal}$ and $\circ = +$}
  \end{cases}
\end{equation*}

\noindent
These partial operations allow to increase/decrease the amount of tokens
in a balance.
For instance, if $\tokBal = 5:\tokT[0]$, then
$\tokBal + 1:\tokT[0] = 6:\tokT[0]$, and
$\tokBal + 1:\tokT[1] = 5:\tokT[0], 1:\tokT[1]$.

\mypar{Deposit}

Any user can create an AMM for two tokens $\tokT[0]$ and $\tokT[1]$,
if such an AMM is not already present in the state.
This is achieved by the transaction
$\actAmmDeposit{\pmvA}{\valV[0]}{\tokT[0]}{\valV[1]}{\tokT[1]}$,
through which $\pmvA$ transfers $\valV[0]:\tokT[0]$ and $\valV[1]:\tokT[1]$
to the new AMM.
In return for the deposit, $\pmvA$ receives a certain positive amount of units
of a new token type $\tokM{\tokT[0]}{\tokT[1]}$, 
which is minted by the AMM.%
\footnote{The actual amount of received units is irrelevant. 
Here we choose $\valV[0]$, but any other choice would be valid.}
We formalise this behaviour by the rule:
\[
\irule
{
  \begin{array}{l}
    \tokBal \tokT[i] \geq \valV[i] > 0
    \;\; (i \in \setenum{0,1})
    \qquad
    \supply[\confG]{\tokM{\tokT[0]}{\tokT[1]}} = 0
  \end{array}
}{
  \begin{array}{l}
    \walA{\tokBal} \mid \confG
    \xrightarrow{\actAmmDeposit{\pmvA}{\valV[0]}{\tokT[0]}{\valV[1]}{\tokT[1]}}
    \\[4pt]
    \walA{\tokBal - \valV[0]:\tokT[0] - \valV[1]:\tokT[1] + \valV[0]:\tokM{\tokT[0]}{\tokT[1]}} \mid
    \amm{\valV[0]:\tokT[0]}{\valV[1]:\tokT[1]} \mid
    \confG
  \end{array}
}
\;\nrule{[Dep0]}
\]

Note that the premise $\supply[\confG]{\tokM{\tokT[0]}{\tokT[1]}} = 0$
implies that $\tokT[0],\tokT[1]$ are distinct atomic tokens,
since otherwise $\tokM{\tokT[0]}{\tokT[1]}$ would not be a minted token.
If $\confG$ is reachable, then this premise also implies
that $\confG$ does \emph{not} contain an AMM 
for the token pair $\tokT[0]$, $\tokT[1]$.

Once an AMM is created, any user can deposit tokens into it ---
\emph{as long as} doing so preserves the ratio of the token reserves in the AMM.
When a user deposits $\valV[0]:\tokT[0]$ and $\valV[1]:\tokT[1]$
to an existing AMM, it receives in return an amount of minted tokens
of type $\tokM{\tokT[0]}{\tokT[1]}$.
This amount is the ratio between the deposited amount $\valV[0]$ and the 
\keyterm{redeem rate} of $\tokM{\tokT[0]}{\tokT[1]}$ 
in the current state $\confG$, 
which is defined as follows for $i \in \setenum{0,1}$:
\begin{equation}
  \label{eq:amm:rx}
  \RX{i}{\confG}{\tokT[0]}{\tokT[1]}
  =
  \frac{\ammR[i]}{\supply[\confG]{\tokM{\tokT[0]}{\tokT[1]}}}
  \qquad
  \text{if $\amm{\ammR[0]:\tokT[0]}{\ammR[1]:\tokT[1]} \in \confG$}
\end{equation}

\noindent
The effect of a deposit transaction on the state is then 
formalised by the following rule:
\[
\irule{
  \begin{array}{l}
    \tokBal{\tokT[i]} \geq \valV[i] > 0
    \;\; (i \in \setenum{0,1})
    \qquad
    \valV[i] = \valV \cdot \RX{i}{\confG}{\tokT[0]}{\tokT[1]}
  \end{array}
}
{
  \begin{array}{ll}
    \confG \; = \;
    & \walA{\tokBal}
      \; \mid \;
      \amm{\ammR[0]:\tokT[0]}{\ammR[1]:\tokT[1]}
      \; \mid \;
      \confD 
      \xrightarrow{\actAmmDeposit{\pmvA}{\valV[0]}{\tokT[0]}{\valV[1]}{\tokT[1]}}
    \\[4pt]
    & \walA{\tokBal - \valV[0]:\tokT[0] - \valV[1]:\tokT[1] + \valV:\tokM{\tokT[0]}{\tokT[1]}}
      \; \mid \;
      \amm{\ammR[0]+\valV[0]:\tokT[0]}{\ammR[1]+\valV[1]:\tokT[1]}
      \; \mid \;
      \confD
  \end{array}
}
\;\nrule{[Dep]}
\]

\noindent
We anticipate that
the premises of the \nrule{[Dep]} rule ensure that
deposits preserve some key quantities across state transitions, namely:
\begin{itemize}
  
\item the ratio between the reserves of $\tokT[0]$ and $\tokT[1]$ in the AMM 
  (see Lemma~\ref{lem:dep-rdm:const}\ref{lem:dep-rdm:const:ratio}).
  This ratio is always defined,
  since the reserves of a token in an AMM cannot be zeroed
  (see Lemma~\ref{lem:non-depletion});
  
\item the \emph{net worth} of the user performing the action
  (see Lemma~\ref{lem:amm:W-preservation}).
  In particular, the value of the minted tokens $\tokM{\tokT[0]}{\tokT[1]}$
  received by the user upon a deposit is equal to
  the value of the tokens $\tokT[0]$, $\tokT[1]$ transferred to the AMM;

\item the \emph{internal exchange rate} of the AMM
  (see Lemma~\ref{lem:sr-homogeneity:dep-rdm}).
  This preservation property holds for a relevant class of swap rate functions,
  called \emph{homogeneous} (see Definition~\ref{def:sr-homogeneity}).
 
\end{itemize}

\mypar{Redeem}

Any user can redeem units of a minted token $\tokM{\tokT[0]}{\tokT[1]}$,
obtaining in return units of the underlying 
atomic tokens $\tokT[0]$ and $\tokT[1]$.
Their actual amounts are determined by the redeem rate:
the idea is that each unit of the minted token $\tokM{\tokT[0]}{\tokT[1]}$ 
can be redeemed for equal fractions of $\tokT[0]$ and $\tokT[1]$ 
remaining in the AMM:
\[
\irule{
  \begin{array}{l}
    \tokBal{\tokM{\tokT[0]}{\tokT[1]}} \geq \valV > 0
    \qquad
    \valV < \supply[\confG]{\tokM{\tokT[0]}{\tokT[1]}}
    \qquad
    \valV[i] = \valV \cdot \RX{i}{\confG}{\tokT[0]}{\tokT[1]}
    \quad (i \in \setenum{0,1})
  \end{array}
}
{
  \begin{array}{ll}
    \confG \; = \;
    & \walA{\tokBal}
      \; \mid \;
      \amm{\ammR[0]:\tokT[0]}{\ammR[1]:\tokT[1]}
      \; \mid \;
      \confD
      \xrightarrow{\actAmmRedeem{\pmvA}{\valV:\tokM{\tokT[0]}{\tokT[1]}}}
    \\[4pt]
    & \walA{\tokBal + \valV[0]:\tokT[0] + \valV[1]:\tokT[1] - \valV:\tokM{\tokT[0]}{\tokT[1]}}
      \; \mid \;
      \amm{\ammR[0]-\valV[0]:\tokT[0]}{\ammR[1]-\valV[1]:\tokT[1]}
      \; \mid \;
      \confD
  \end{array}
}
\;\nrule{[Rdm]}
\]

\noindent
Note that the premise $\valV < \supply[\confG]{\tokM{\tokT[0]}{\tokT[1]}}$
ensures that the reserves are not depleted, \ie $\valV[i] < \ammR[i]$.
Similarly to the \nrule{[Dep]} rule,
the premises of \nrule{[Rdm]} ensure that:
\begin{itemize}
\item the net worth of the user performing the action is preserved
  (\ie, the net worth of burnt minted tokens is equal to 
  that of the tokens received by $\pmvA$);
\item  the internal exchange rate of the AMM
  is unaffected by the transition,
  if the swap rate function is homogeneous.
\end{itemize}

\mypar{Swap}

Any user $\pmvA$ can swap $\valV$ units of $\tokT[0]$ in her wallet
for some units of $\tokT[1]$ in an AMM 
\mbox{$\amm{\ammR[0]:\tokT[0]}{\ammR[1]:\tokT[1]}$}
through the transaction
\mbox{$\actAmmSwapExact{\pmvA}{}{\valV}{\tokT[0]}{\tokT[1]}$}.
Symmetrically, $\pmvA$ can swap $\valV$ of her units of $\tokT[1]$ 
for units of $\tokT[0]$ in the AMM through a transaction 
\mbox{$\actAmmSwapExact{\pmvA}{}{\valV}{\tokT[1]}{\tokT[0]}$}.
The \keyterm{swap rate} $\SX{x,\ammR[0],\ammR[1]}$
determines the amount of \emph{output tokens} $\tokT[1]$
that a user receives upon an amount of $x$ \emph{input tokens} $\tokT[0]$ 
in an AMM \mbox{$\amm{\ammR[0]:\tokT[0]}{\ammR[1]:\tokT[1]}$}.
\[
\irule
{
  \begin{array}{l}
    \tokBal{\tokT[0]} \geq x
    \qquad
    y = x \cdot \SX{x,\ammR[0],\ammR[1]} < \ammR[1]
  \end{array}
}
{\begin{array}{l}
   \walA{\tokBal}
   \mid
   \amm{\ammR[0]:\tokT[0]}{\ammR[1]:\tokT[1]}
   \mid
   \confG
   \xrightarrow{\actAmmSwapExact{\pmvA}{}{x}{\tokT[0]}{\tokT[1]}}
   \\[4pt]
   \walA{\tokBal - x:\tokT[0] + y:\tokT[1]}
   \mid
   \amm{\ammR[0]+ x:\tokT[0]}{\ammR[1]-y:\tokT[1]}
   \mid
   \confG
   \hspace{-5pt}
 \end{array}
}
\nrule{[Swap]}
\]

The swap rate function is a parameter of our model:
we will discuss in~\S\ref{sec:swap-rate} some desiderata for this function, 
and the behavioural properties they induce on the AMM semantics.
As an instance, we consider below 
the \keyterm{constant product swap rate}~\cite{rvammspec},
which is used in mainstream AMM implementations,
like \eg in Uniswap v2~\cite{uniswapimpl}, 
Mooniswap~\cite{mooniswapimpl} and SushiSwap~\cite{sushiswapimpl}.
We will use this swap rate function in all the examples in this paper.

\begin{defi}[Constant product swap rate]
  \label{def:const-prod}
  The constant product swap rate function is:  
  \[
    \SX{x,\ammR[0],\ammR[1]}
    \; = \;
    \frac{\ammR[1]}{\ammR[0] + x}
  \]
\end{defi}

The constant product swap rate ensures that,
if an AMM $\amm{\ammR[0]:\tokT[0]}{\ammR[1]:\tokT[1]}$ evolves into
$\amm{\ammR[0]+x:\tokT[0]}{\ammR[1]-y:\tokT[1]}$ upon a swap, 
then the product between the reserves is preserved:
\[
(\ammR[0]+x) (\ammR[1]-y)
\; = \;
(\ammR[0]+x) \Big( \ammR[1]-x \cdot \frac{\ammR[1]}{\ammR[0] + x} \Big)
\; = \;
\ammR[0] \ammR[1]
\]

Overall, the behaviour of the transition rules discussed above highlights
some landmark properties of AMMs, namely:
\begin{itemize}
  
\item since neither deposits nor redeems affect the net worth of the users
  performing them, the only way for users to increase their net worth
  is to perform swaps.
  Since, as we will see in Lemma~\ref{lem:amm:W-preservation},
  the \emph{global} net worth is constant, this means that increasing ones'
  net worth results in a decrease of someone else's net worth;
  
\item the internal exchange rate of an AMM is affected only by swap actions
  (provided that the swap rate function is homogeneous).
  This is a natural behaviour, because swaps reflect the value of tokens
  perceived by users.
  We will show later in \S\ref{sec:swap-rate} that
  the constant sum/product/mean swap rate functions are homogeneous.
  
\end{itemize}

\begin{figure}[t]
\small
\scalebox{0.95}{\parbox{\textwidth}{%
  \begin{align*}
      \nonumber
      & \walA{70:\tokT[0], 70:\tokT[1]} \mid
        \walB{30:\tokT[0],10:\tokT[1]}
      \\
      \xrightarrow{\actAmmDeposit{\pmvA}{70}{\tokT[0]}{70}{\tokT[1]}} \;
      & \walA{70:\tokM{\tokT[0]}{\tokT[1]}} \mid
        \walB{30:\tokT[0],10:\tokT[1]} \mid
        \amm{70:\tokT[0]}{70:\tokT[1]}
      \\
      \xrightarrow{\actAmmSwapExact{\pmvB}{}{30}{\tokT[0]}{\tokT[1]}} \;
      & \walA{70:\tokM{\tokT[0]}{\tokT[1]}} \mid
        \walB{0:\tokT[0],31:\tokT[1]} \mid
        \amm{100:\tokT[0]}{49:\tokT[1]}
      \\
      \xrightarrow{\actAmmSwapExact{\pmvB}{}{21}{\tokT[1]}{\tokT[0]}} \;
      & \walA{70:\tokM{\tokT[0]}{\tokT[1]}} \mid
        \walB{30:\tokT[0],10:\tokT[1]} \mid
        \amm{70:\tokT[0]}{70:\tokT[1]}
      \\
      \xrightarrow{\actAmmRedeem{\pmvA}{30:\tokM{\tokT[0]}{\tokT[1]}}} \;
      & \walA{30:\tokT[0],30:\tokT[1],40:\tokM{\tokT[0]}{\tokT[1]}} \mid
        \walB{30:\tokT[0],10:\tokT[1]} \mid
        \amm{40:\tokT[0]}{40:\tokT[1]}
      \\
      \xrightarrow{\actAmmSwapExact{\pmvB}{}{30}{\tokT[0]}{\tokT[1]}} \;
      & \walA{30:\tokT[0],30:\tokT[1],40:\tokM{\tokT[0]}{\tokT[1]}} \mid
        \walB{0:\tokT[0],27:\tokT[1]} \mid
        \amm{70:\tokT[0]}{23:\tokT[1]}
      \\
      \xrightarrow{\actAmmRedeem{\pmvA}{30:\tokM{\tokT[0]}{\tokT[1]}}} \;
      & \walA{82:\tokT[0],47:\tokT[1],10:\tokM{\tokT[0]}{\tokT[1]}} \mid
        \walB{0:\tokT[0],27:\tokT[1]} \mid
        \amm{18:\tokT[0]}{6:\tokT[1]}
    \end{align*}}}%
  \caption{Interactions between two users and an AMM.}
  \label{fig:amm}
\end{figure}


\begin{exa}
  \label{ex:amm}
  \Cref{fig:amm} shows a computation in our model.
  We discuss below the effect of the fired transactions, showing 
  in \Cref{fig:amm:graph} the evolution of the token reserves in the AMM:
  \begin{enumerate}
  
  \item $\actAmmDeposit{\pmvA}{70}{\tokT[0]}{70}{\tokT[1]}$.
    Starting from an initial state,
    $\pmvA$ creates a new AMM, depositing 
    \mbox{$70:\tokT[0]$} and \mbox{$70:\tokT[1]$}.
    In return, $\pmvA$ receives $70$ units of the minted token 
    $\tokM{\tokT[0]}{\tokT[1]}$.
    
  \item $\actAmmSwapExact{\pmvB}{}{30}{\tokT[0]}{\tokT[1]}$.
    $\pmvB$ swaps $30$ units of $\tokT[0]$ for an 
    amount $y$ of units of $\tokT[1]$ determined by the swap rate.
    Since we are assuming the constant product swap rate,
    we obtain $y = 30 \cdot \nicefrac{70}{70+30} = 21$.
    This swap rate function ensures that swaps preserve 
    the product between the token reserves in the AMM:
    in \Cref{fig:amm:graph}, we show indeed that the swap
    results in a traversal along the curve $k = 70\cdot 70$ from 
    \mbox{$\amm{70:\tokT[0]}{70:\tokT[1]}$} to 
    \mbox{$\amm{100:\tokT[0]}{49:\tokT[1]}$}.

  \item $\actAmmSwapExact{\pmvB}{}{21}{\tokT[1]}{\tokT[0]}$.
    $\pmvB$ reverses the effect of his previous action 
    by swapping $21$ units of $\tokT[1]$
    for $y = 21 \cdot \nicefrac{100}{49+21} = 30$ of $\tokT[0]$.
    \Cref{fig:amm:graph} shows that the swap results 
    in a traversal  from 
    \mbox{$\amm{100:\tokT[0]}{49:\tokT[1]}$} to
    \mbox{$\amm{70:\tokT[0]}{70:\tokT[1]}$} along the curve $k = 70\cdot 70$. 

  \item $\actAmmRedeem{\pmvB}{30:\tokM{\tokT[0]}{\tokT[1]}}$.
    $\pmvB$ redeems $30$ units of the minted token $\tokM{\tokT[0]}{\tokT[1]}$,
    accordingly reducing the token reserves in the AMM to
    \mbox{$\amm{40:\tokT[0]}{40:\tokT[1]}$}. 
    Note that the received tokens exhibit the same 1-to-1 ratio as after
    the initial deposit.

  \item $\actAmmSwapExact{\pmvB}{}{30}{\tokT[0]}{\tokT[1]}$.
    $\pmvB$ swaps $30$ units of $\tokT[0]$, receiving
    $y = 30 \cdot \nicefrac{40}{40+30} \approx 17$ units of $\tokT[1]$.
    Note that the swap rate, \ie $\nicefrac{40}{40+30} \approx 0.57$,
    has decreased \wrt the first swap, \ie $\nicefrac{70}{70+30} = 0.7$, 
    even though the AMM has the same 1-to-1 reserves ratio. 
    This is caused by the reduction in reserves occurred after $\pmvA$'s 
    redeem action: thus, the swap rate is sensitive not only to 
    the ratio of reserves in the AMM, but also on their actual values.

  \item $\actAmmRedeem{\pmvA}{30:\tokM{\tokT[0]}{\tokT[1]}}$.
    $\pmvA$ redeems $30$ units of the minted token 
    $\tokM{\tokT[0]}{\tokT[1]}$, thereby extracting $52$ units of $\tokT[0]$ 
    and $17$ units of $\tokT[1]$ from the AMM. 
    Note that the ratio of redeemed tokens is no longer 1-to-1 as 
    in the previous redeem action, 
    as the prior swap has changed the ratio between the 
    funds of $\tokT[0]$ and $\tokT[1]$ in the AMM.
    
  \end{enumerate}
  
  Finally, observe that the supply of both $\tokT[0]$ and $\tokT[1]$ 
  remains constant.
  We will show in Lemma~\ref{lem:supply:const} 
  that the supply of atomic token types is always preserved.
  \hfill\qedex
\end{exa}

\begin{figure}[t]
  \centering
  \includegraphics[width=\linewidth]{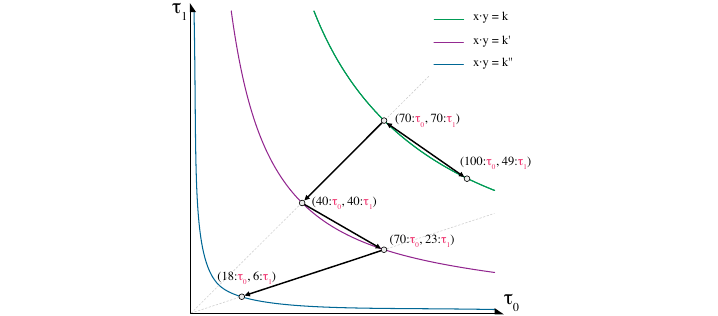}
  \caption{Evolution of reserves of AMM $(\tokT[0],\tokT[1])$ along the trace in~\Cref{fig:amm}.}
  \label{fig:amm:graph}
\end{figure}

\section{Prices, exchange rates and net worth}
\label{sec:econom-defs}

In this~\namecref{sec:econom-defs} we introduce some economic notions
which are pivotal for understanding the economic mechanisms of AMMs.

\mypar{Token prices and exchange rates}

We assume an external oracle that prices atomic tokens.
Formally, we model this oracle as a function
$\exchO{}: \TokU[0] \rightarrow \RP$,
assuming that the prices given by the oracle
are constant along executions
(see \autoref{sec:variants:price-updates} for a discussion about
dynamic price updates).
While the prices of atomic tokens are constant,
that of minted tokens may vary at run-time
as a function of the state.
More precisely, the price $\exchO[\confG]{\tokM{\tokT[0]}{\tokT[1]}}$
of a minted token $\tokM{\tokT[0]}{\tokT[1]}$ depends both
on the supply of the minted token in the users' wallets
and on the reserves of $\tokT[0]$ and $\tokT[1]$ in the AMM: 
\begin{equation}
  \label{def:price}
  \exchO[\confG]{\tokM{\tokT[0]}{\tokT[1]}} 
  \; = \;
  \dfrac
  {\ammR[0] \cdot \exchO{\tokT[0]} + \ammR[1] \cdot \exchO{\tokT[1]}}
  {\supply[\confG]{\tokM{\tokT[0]}{\tokT[1]}}}
  \mbox{\qquad}
  \text{if $\amm{\ammR[0]:\tokT[0]}{\ammR[1]:\tokT[1]} \in \confG$}
\end{equation}
\noindent
For uniformity, we define $\exchO[\confG]{\tokT} = \exchO{\tokT}$
when $\tokT \in \TokU[0]$.
Lemma~\ref{lem:non-depletion} will ensure that
$\ammR[0],\ammR[1] > 0$ and
$\supply[\confG]{\tokM{\tokT[0]}{\tokT[1]}} > 0$
in every reachable state $\confG$ containing an AMM 
for the token pair $\tokT[0]$, $\tokT[1]$.
Therefore, the price of the token $\tokM{\tokT[0]}{\tokT[1]}$
is always defined and positive in reachable states.

The idea underlying \Cref{def:price} is that
the price of one unit of minted token must be
equal to the value of the atomic tokens
that can be obtained by redeeming the minted token.
Indeed, by rule \nrule{[Rdm]} and \Cref{eq:amm:rx} we have that:
\[
  \exchO[\confG]{\tokM{\tokT[0]}{\tokT[1]}} 
  \; = \;
  \dfrac
  {\ammR[0]}
  {\supply[\confG]{\tokM{\tokT[0]}{\tokT[1]}}}
  \exchO{\tokT[0]}
  +
  \dfrac
  {\ammR[1]}
  {\supply[\confG]{\tokM{\tokT[0]}{\tokT[1]}}}
  \exchO{\tokT[1]}
  =
  \RX{0}{\confG}{\tokT[0]}{\tokT[1]} \cdot \exchO{\tokT[0]}
  +
  \RX{1}{\confG}{\tokT[0]}{\tokT[1]} \cdot \exchO{\tokT[1]}
\] 
which substantiates our desideratum.
This intuition will be formalized later in Lemma~\ref{lem:amm:W-TokUInit}.

The \keyterm{exchange rate} $\X{\tokT[0],\tokT[1]}$ 
between atomic token types $\tokT[0]$ and $\tokT[1]$
is the number of units of $\tokT[1]$
that one can buy with $1$ unit of $\tokT[0]$ at the price given by the external oracle:
\begin{equation}
  \label{eq:exchange-rate}
  \X{\tokT[0],\tokT[1]}
  \; = \;
  \frac
  {\exchO{\tokT[0]}}
  {\exchO{\tokT[1]}}
\end{equation}

Hence, assuming an exchange at the prices of the external oracle,
a user paying $x$ units of $\tokT[0]$ would receive
$x \cdot \X{\tokT[0],\tokT[1]}$ units of $\tokT[1]$.

Note that the exchange rate between two token types only depends on external
oracles, neglecting the state of AMMs.
However, AMMs themselves can act as (decentralised)
price oracles~\cite{Angeris20aft},
since they induce an exchange rate
based on the effect of swaps in the current state.
More precisely, the \keyterm{internal exchange rate}
$\X[\confG]{\tokT[0],\tokT[1]}$
between two atomic token types $\tokT[0]$ and $\tokT[1]$
in a state $\confG$ is
the limit of the swap rate function as $x$ approaches $0$:
\footnote{This notion is also dubbed as marginal price~\cite{Angeris20aft} or spot exchange rate~\cite{Xu21sok} in literature.}
\begin{equation}
  \label{eq:exchange-rate:internal}
  \X[\confG]{\tokT[0],\tokT[1]}
  \; = \;
  \lim_{x \rightarrow 0}
  \SX{x,\ammR[0],\ammR[1]}
  \qquad \text{if $\amm{\ammR[0]:\tokT[0]}{\ammR[1]:\tokT[1]}$}
\end{equation}
The intuition is similar to that in~\Cref{eq:exchange-rate}:
a user swapping $x$ units of $\tokT[0]$ for $\tokT[1]$ through the AMM
(for $x$ \emph{very small}) would expect to receive
$x \cdot \X[\confG]{\tokT[0],\tokT[1]}$ units of $\tokT[1]$.
We will see later in \autoref{sec:arbitrage}
that rational users will perform actions
that align the internal exchange rate to the one given by external oracles.

\keyterm{Slippage} measures the discrepancy between the internal exchange rate
and the actual ratio between the amounts of output and input tokens
obtained upon the swap~\cite{Xu21sok}:
\begin{equation}
  \label{eq:slippage}
  \SL[\confG]{x,\tokT[0],\tokT[1]}
  \; = \;
  \frac{\X[\confG]{\tokT[0],\tokT[1]}}
       {\SX{x,\ammR[0],\ammR[1]}} - 1
       \qquad \text{if $\amm{\ammR[0]:\tokT[0]}{\ammR[1]:\tokT[1]}$}
\end{equation}

Ideally, slippage should disadvantage large trades, \ie
trying to obtain a larger amount of tokens with a swap
should make them more expensive, increasing the slippage.
We will compute in sections \ref{sec:swap-rate:const-sum}-\ref{sec:swap-rate:const-mean}
the internal exchange rate and the slippage of some common AMMs.

\begin{exa}
  \label{ex:amm:price}
  Let
  \(
  \confG
  =
  \walA{82:\tokT[0],47:\tokT[1],10:\tokM{\tokT[0]}{\tokT[1]}} \mid 
  \amm{18:\tokT[0]}{6:\tokT[1]} \mid 
  \walB{\cdots}
  \)
  be the final state of the computation in~\Cref{fig:amm}.
  We have that $\supply[\confG]{\tokM{\tokT[0]}{\tokT[1]}} = 10$,
  since only $\pmvA$'s wallet contains units of the minted token.
  Assume that the prices of atomic tokens are 
  $\exchO{\tokT[0]} = 5$ and $\exchO{\tokT[1]} = 9$. 
  The price of the minted token $\tokM{\tokT[0]}{\tokT[1]}$ is then:
  \[
  \exchO[\confG]{\tokM{\tokT[0]}{\tokT[1]}}
  =
  \frac{18 \cdot \exchO{\tokT[0]} + 6 \cdot \exchO{\tokT[1]}}{10}
  =
  \frac{18 \cdot 5 + 6 \cdot 9}{10}
  =
  14.4
  \]
  The exchange rate between the two tokens is:
  \[
  \X{\tokT[0],\tokT[1]}
  \; = \;
  \frac
  {\exchO{\tokT[0]}}
  {\exchO{\tokT[1]}}
  \; = \;
  \frac
  {5}
  {9}
  \; = \;
  0.55
  \]
  which means that to buy $1$ unit of $\tokT[0]$, 
  one needs $0.55$ units of $\tokT[1]$.
  Note instead that the internal exchange rate is:
  \[
  \X[\confG]{\tokT[0],\tokT[1]}
  \; = \;
  \lim_{x \rightarrow 0} \SX{x,18,6}
  =
  \frac{6}{18}
  \approx
  0.33
  \]
  We will see in~Example~\ref{ex:arbitrage} that
  the discrepancy between internal and oracle exchange rate
  can by exploited by users to increase their gain.
  The slippage of a $\actAmmSwapExact{}{}{x}{\tokT[0]}{\tokT[1]}$ is:
  \[
  \SL[\confG]{x,\tokT[0],\tokT[1]}
  \; = \;
  \frac{\X[\confG]{\tokT[0],\tokT[1]}}{\SX{x,18,6}} - 1
  \; = \;
  \frac{x}{18}
  \]
  from which we can see that the slippage grows with the input amount $x$.
  \hfill\qedex
\end{exa}

\mypar{Net worth and gain}

The \keyterm{net worth} of a user $\pmvA$ is a measure of $\pmvA$'s 
wealth in tokens (both atomic and minted). 
Formally, we define the net worth of $\pmvA$ in a state $\confG$ as:
\begin{equation} 
  \label{eq:net-worth:user}
  W_{\pmvA}(\confG) \; = \; \begin{cases}
    \sum_{\tokT \in \dom{\tokBal}}
    \tokBal{\tokT} \cdot \exchO[\confG]{\tokT}
  & \text{if $\walA{\tokBal} \in \confG$} 
  \\[4pt]
  0 & \text{otherwise}
  \end{cases}
\end{equation}

\noindent
Note that $W_{\pmvA}(\confG) \in \RNN$, since 
balances $\tokBal$ are \emph{finite} maps, 
and $\exchO[\confG]{\tokT}$ is always defined.

The \keyterm{global net worth} $W(\confG)$ in a state $\confG$ 
is the sum of the net worth in users' wallets.
Note that the token reserves in AMMs are not accounted for by $W(\confG)$,
because their value is already recorded by minted tokens 
held in users' wallets. 
Indeed, the equality:
\[
  \supply[\confG]{\tokM{\tokT[0]}{\tokT[1]}} \cdot \exchO[\confG]{\tokM{\tokT[0]}{\tokT[1]}}
  \; = \;
  \ammR[0] \cdot \exchO{\tokT[0]} + \ammR[1] \cdot \exchO{\tokT[1]}
\]
between the net worth of a minted token and the value of the AMM
is a direct consequence of the definition of price in Equation~\eqref{def:price}.

We denote by $\gain[\confG]{\pmvA}{\bcB}$ the \keyterm{gain} of user $\pmvA$
upon performing a sequence of transactions $\bcB$ enabled in state $\confG$
(if $\bcB$ is not enabled in $\confG$, we stipulate that the gain is zero):
\begin{equation}
  \label{def:gain:A}
  \gain[\confG]{\pmvA}{\bcB}
  \; = \;
  W_{\pmvA}(\confGi) -W_{\pmvA}(\confG)
  \qquad \text{if $\confG \xrightarrow{\bcB} \confGi$}
\end{equation}

To maximize their gain, users can perform
different interactions with the AMM,
\eg, by investing tokens or trading units of differently priced token types.

The following lemma quantifies
the gain of users upon firing a $\ammSwapOp$ transaction.
Note that this quantification does not depend on any of the 
properties of the swap rate function introduced later on in \autoref{sec:swap-rate}:
actually, it holds for any swap rate function.

\begin{lem}[Swap gain]
  \label{lem:swap:gain}
  Let $\confG = \amm{\ammR[0]:\tokT[0]}{\ammR[1]:\tokT[1]} \mid \confD$,
  and let $\txT = \actAmmSwapExact{\pmvA}{}{x}{\tokT[0]}{\tokT[1]}$
  be enabled in $\confG$.
  Then:
  \begin{align*}
    \gain[{\confG}]{\pmvA}{\txT}
    & = \phantom{-}
      x \cdot \big(
      \SX{x,\ammR[0],\ammR[1]} \, \exchO{\tokT[1]}
      -
      \exchO{\tokT[0]}
      \big)
      \cdot
      \Big(
      1 - \frac{\tokBal[\pmvA]\tokM{\tokT[0]}{\tokT[1]}}{\supply[\confG]{\tokM{\tokT[0]}{\tokT[1]}}}
      \Big)
    && \text{if $\walA{\tokBal[\pmvA]} \in \confG$}
    \\
    \gain[{\confG}]{\pmvB}{\txT}
    & =
      - x \cdot
      \big( \SX{x,\ammR[0],\ammR[1]} \exchO{\tokT[1]} - \exchO{\tokT[0]} \big)
      \cdot 
      \frac
      {\tokBal[\pmvB]\tokM{\tokT[0]}{\tokT[1]}}
      {\supply[\confG]{\tokM{\tokT[0]}{\tokT[1]}}}
    && \text{if $\walB{\tokBal[\pmvB]} \in \confG$, $\pmvB \neq \pmvA$}
  \end{align*}
\end{lem}

A direct consequence of Lemma~\ref{lem:swap:gain} is that
if $\pmvA$ performs a $\ammSwapOp$ between $\tokT[0]$ and $\tokT[1]$
and she holds all the units of the minted token $\tokM{\tokT[0]}{\tokT[1]}$,
then her gain will be zero. 
Further, $\pmvA$ maximizes her gain when she has no minted tokens
of type $\tokM{\tokT[0]}{\tokT[1]}$. 
The lemma also implies that if the user performing the swap 
has a positive gain, then all the users who hold units of 
$\tokM{\tokT[0]}{\tokT[1]}$ will have a negative gain.

The following lemma states that a $\ammSwapOp$ transaction on an AMM
$\amm{\ammR[0]:\tokT[0]}{\ammR[1]:\tokT[1]}$
has a strictly positive gain \emph{if and only if} 
the swap rate is strictly greater than the oracle exchange rate between 
$\tokT[0]$ and $\tokT[1]$.
This holds for \emph{any} swap rate function, under the condition that 
the user who performs the swap has no minted tokens
of type $\tokM{\tokT[0]}{\tokT[1]}$.

\begin{lem}[Swap rate \emph{vs.} exchange rate]
  \label{lem:swap:gain-SX-X}
  Let $\confG = \walA{\tokBal} \mid \amm{\ammR[0]:\tokT[0]}{\ammR[1]:\tokT[1]} \mid \confD$ 
  be such that $\tokBal{\tokM{\tokT[0]}{\tokT[1]}} = 0$, and
  let $\txT = \actAmmSwapExact{\pmvA}{}{x}{\tokT[0]}{\tokT[1]}$ 
  be enabled in $\confG$.
  Then: 
  \[
    \gain[{\confG}]{\pmvA}{\txT} \circ 0
    \iff
    \SX{x,\ammR[0],\ammR[1]} \circ \X{\tokT[0],\tokT[1]}
    \tag*{for $\circ \in \setenum{<,=,>}$}
  \]
\end{lem}

\begin{exa}
  \label{ex:amm:networth}
  Let
  \(
  \confG[0]
  =
  \walA{70:\tokT[0], 70:\tokT[1]} \mid \walB{30:\tokT[0],10:\tokT[1]}
  \)
  be the initial state of the computation in~\Cref{fig:amm}.
  Let $\exchO{\tokT[0]} = 5$ and $\exchO{\tokT[1]} = 9$.
  The users' net worth in $\confG[0]$ and in the final state
  \(
  \confG
  =
  \walA{82:\tokT[0],47:\tokT[1],10:\tokM{\tokT[0]}{\tokT[1]}} \mid
  \walB{0:\tokT[0],27:\tokT[1]} \mid
  \amm{18:\tokT[0]}{6:\tokT[1]}
  \)
  is as follows:
  \begin{align*}
    & W_{\pmvA}(\confG[0])
      = 70 \cdot \exchO{\tokT[0]} + 70 \cdot \exchO{\tokT[1]} 
      = 980
    && W_{\pmvB}(\confG[0])
       = 30 \cdot \exchO{\tokT[0]} + 10 \cdot \exchO{\tokT[1]}
       = 240
    \\
    & W_{\pmvA}(\confG)
      =
      82 \cdot \exchO{\tokT[0]}
      + 47 \cdot \exchO{\tokT[1]}
      + 10 \cdot \exchO[{\confG}]{\tokM{\tokT[0]}{\tokT[1]}} = 977
    && W_{\pmvB}(\confG)
       = 27 \cdot \exchO{\tokT[1]} = 243
  \end{align*}
  Note that $\pmvA$'s net worth of  has decreased \wrt 
  the initial state, while that of $\pmvB$ has increased:
  indeed, the gain of $\pmvA$ upon the sequence of transactions $\bcB$ is
  $\gain[\confG]{\pmvA}{\bcB} = 977-980 = -3$,
  while that of $\pmvB$ is
  $\gain[\confG]{\pmvB}{\bcB} = 243-240 = 3$.
  One may think that $\pmvB$ has been more successful than $\pmvA$, 
  but this depends on the users' goals.
  Note, \eg, that $\pmvA$ holds $10$ units of the minted token 
  $\tokM{\tokT[0]}{\tokT[1]}$, 
  whose price may increase in the future. 
  \hfill\qedex
\end{exa}

\section{Structural properties of AMMs}
\label{sec:struct-properties}

We now establish some structural properties of AMMs,
which do not depend on the design of the economic mechanisms, 
\ie on the choice of the swap rate function.
These structural properties are the basis for AMM interactions
that occur in the wild, and that cumulatively give rise to
complex emerging behaviours like arbitrage and MEV.
Hence, establishing these structural properties is
a preliminary sanity check for our AMM model.
We will provide further support for the coherence between our model
and actual AMMs by showing that the above-mentioned complex behaviours
are expressible in our model (see \autoref{sec:arbitrage} and \autoref{sec:mev}).

First, we establish that the AMMs' transition system is deterministic.
This follows from the fact that,
given a state $\confG$ and a transaction $\txT$,
there is at most one applicable rule.
Note that determinism is a crucial property for blockchains,
since it ensures that all the nodes in the blockchain network
are able to reconstruct a common state from a sequence of transactions.
Therefore, it makes sense that determinism holds also for our AMM model.

\begin{lem}[Determinism]
  \label{lem:amm:determinism}
  If $\confG \xrightarrow{\txT} \confGi$ 
  and $\confG \xrightarrow{\txT} \confGii$, then
  $\confGi = \confGii$.
\end{lem}

We can lift the statement to sequences of transactions 
by using a simple inductive argument.
The same applies to other single-step results in this~\namecref{sec:struct-properties}.

Lemma~\ref{lem:non-depletion} ensures 
that the reserves in an AMM cannot be zeroed,
and that the same holds for the units of any minted token.
Summing up, this ensures that the price of any minted token
is always defined and positive.

\begin{lem}[Non depletion]
  \label{lem:non-depletion}
  For all states $\confG$,
  if $\amm{\ammR[0]:\tokT[0]}{\ammR[1]:\tokT[1]} \in \confG$
  then: 
  \begin{enumerate}[(a)]

  \item \label{lem:non-depletion:ammR}
    $\ammR[i] > 0$, for $i \in \setenum{0,1}$;

  \item \label{lem:non-depletion:supply}
    $\supply[\confG]{\tokM{\tokT[0]}{\tokT[1]}} > 0$.

  \end{enumerate}
\end{lem}

\subsection{Preservation properties}

Lemma~\ref{lem:supply:const} ensures that transactions preserve
the supply of \emph{atomic} tokens.
Minted tokens, instead, are preserved only by $\ammSwapOp$ transactions,
since deposit and redeem transactions, respectively, 
create and destroy minted tokens.
This fact will be instrumental to prove the preservation of
the net worth (see Lemma~\ref{lem:amm:W-preservation}).
\begin{lem}[Preservation of token supply]
  \label{lem:supply:const}
  Let $\confG \xrightarrow{\txT} \confGi$.
  Then: 
  \begin{enumerate}[(a)]
    
  \item \label{lem:supply:const:0}
    for all $\tokT \in \TokU[0]$,
    $\supply[\confG]{\tokT} = \supply[\confGi]{\tokT}$
    
  \item \label{lem:supply:const:1}
    if $\txType{\txT} = \ammSwapOp$, then for all $\tokT \in \TokU[1]$,
    $\supply[\confG]{\tokT} = \supply[\confGi]{\tokT}$
      
  \end{enumerate}
\end{lem}

Lemma~\ref{lem:dep-rdm:const} states
that deposit and redeem transactions preserve 
the reserves ratio in AMMs, the redeem rate, and the price of minted tokens.
These preservation properties will be exploited later on to
determine the solution to the arbitrage game after deposits and redeems
(see Theorems~\ref{thm:dep-arbitrage} and~\ref{thm:rdm-arbitrage}).

\begin{lem}[Preservation upon deposits/redeems]
  \label{lem:dep-rdm:const}
  Let $\confG \xrightarrow{\txT} \confGi$, with
  $\amm{\ammR[0]:\tokT[0]}{\ammR[1]:\tokT[1]} \in \confG$.
  If $\txType{\txT} \in \setenum{\ammDepositOp,\ammRedeemOp}$, then:
  \begin{enumerate}[(a)]
    
  \item \label{lem:dep-rdm:const:ratio}
    if
    $\amm{\ammRi[0]:\tokT[0]}{\ammRi[1]:\tokT[1]} \in \confGi$,
    then
    $\nicefrac{\ammR[1]}{\ammR[0]} = \nicefrac{\ammRi[1]}{\ammRi[0]}$
    
  \item \label{lem:dep-rdm:const:rx}
    $\RX{i}{\confG}{\tokT[0]}{\tokT[1]} = \RX{i}{\confGi}{\tokT[0]}{\tokT[1]}$,
    for $i \in \setenum{0,1}$

  \item \label{lem:dep-rdm:const:price}
      $\exchO[\confG]{\tokM{\tokT[0]}{\tokT[1]}} = \exchO[\confGi]{\tokM{\tokT[0]}{\tokT[1]}}$

  \end{enumerate}
\end{lem}

Lemma~\ref{lem:amm:W-preservation} ensures that transactions
(of any type) preserve the \emph{global} net worth, 
whereas the net worth of individual users is preserved 
only by redeem and deposit transactions. 
A direct consequence of this preservation result 
is that users can increase their net worth only by performing swaps:
Indeed, we will find in Theorem~\ref{thm:arbitrage} that the solution
of the arbitrage game only contains $\ammSwapOp$ transactions.
Furthermore, if a user has a positive gain, then some other user
must have a loss.

\begin{lem}[Preservation of net worth]
  \label{lem:amm:W-preservation}
  Let $\confG \xrightarrow{\txT} \confGi$. 
  Then: 
  \begin{enumerate}[(a)]
  \item \label{lem:amm:W-preservation:b}
    if $\txType{\txT} \neq \ammSwapOp$
    then, for all $\pmvA$: $W_{\pmvA}(\confG) = W_{\pmvA}(\confGi)$
  \item \label{lem:amm:W-preservation:a}
    $W(\confG) = W(\confGi)$
  \end{enumerate}
\end{lem}

The following lemma, which is a direct consequence of
Lemma~\ref{lem:amm:W-preservation}\ref{lem:amm:W-preservation:b},
supports the definition of the price of minted tokens in~\Cref{def:price}:
indeed, computing the net worth of a user $\pmvA$ under that price definition
corresponds to making $\pmvA$ first redeem all her minted tokens, and then
summing the price of the resulting atomic tokens. 

\begin{lem}
  \label{lem:amm:W-TokUInit}
  Let $\confG \xrightarrow{\bcB} \confGi$, 
  where $\bcB$ contains only $\ammRedeemOp$ actions of $\pmvA$.
  If $\walA{\tokBal} \in \confGi$ and $\dom{\tokBal} \cap \TokU[1] = \emptyset$,
  then:
  \[
  W_{\pmvA}(\confG)
  \; = \;
  \sum_{\tokT \in \dom \tokBal}
  \tokBal{\tokT} \cdot \exchO{\tokT}
  \]
\end{lem}

\begin{exa}
  \label{ex:amm:W-TokUInit}
  Let
  \(
  \confG =
  \walA{82:\tokT[0],47:\tokT[1],10:\tokM{\tokT[0]}{\tokT[1]}} \mid 
  \amm{27:\tokT[0]}{9:\tokT[1]} \mid 
  \walB{5:\tokM{\tokT[0]}{\tokT[1]}}
  \), 
  and let $\exchO{\tokT[0]} = 5$ and $\exchO{\tokT[1]} = 9$. 
  We have that $\exchO[\confG]{\tokM{\tokT[0]}{\tokT[1]}} = 14.4$,
  and $W_{\pmvA}(\confG) = 977$.
  Assume that $\pmvA$ performs a transaction from $\confG$
  to redeem all $10$ units of $\tokM{\tokT[0]}{\tokT[1]}$ in her wallet.
  The resulting state is 
  $\confGi = \walA{100:\tokT[0],53:\tokT[1]} \mid \amm{9:\tokT[0]}{3:\tokT[1]} \mid \cdots$.
  We compute $\pmvA$'s net worth in $\confGi$ using the oracle token prices:
  \[
    W_{\pmvA}(\confGi) 
    \; = \;
    100 \cdot \exchO{\tokT[0]}
    + 53 \cdot \exchO{\tokT[1]}
    \; = \;
    100 \cdot 5 + 53 \cdot 9 
    \; = \;
    977
  \]
  which is coherent with the net worth predicted by Lemma~\ref{lem:amm:W-TokUInit}.
  \hfill\qedex
\end{exa}

\subsection{Liquidity}

Lemma~\ref{lem:amm:liquidity} ensures that funds cannot be \emph{frozen} in an AMM,
\ie that users can always redeem arbitrary amounts of the tokens deposited
in an AMM, as long as the reserves are not zeroed.
Note that, since
$\amm{\ammR[0]:\tokT[0]}{\ammR[1]:\tokT[1]} = \amm{\ammR[1]:\tokT[1]}{\ammR[0]:\tokT[0]}$,
the statement also holds when swapping $\ammR[0]$ with $\ammR[1]$.

\begin{lem}[Liquidity]
  \label{lem:amm:liquidity}
  Let $\confG$ be such that
  $\amm{\ammR[0]:\tokT[0]}{\ammR[1]:\tokT[1]} \in \confG$.
  Then, for all $\ammRi[0] < \ammR[0]$, 
  there exist $\ammRi[1] < \ammR[1]$,
  $\confGi$ and
  $\bcB$ only containing $\ammRedeemOp$ transactions such that
  $\confG \xrightarrow{\bcB} \amm{\ammRi[0]:\tokT[0]}{\ammRi[1]:\tokT[1]} \mid \confGi$.
\end{lem}

\subsection{Reordering of transactions}

In general, given two transactions $\txT[0]$ and $\txT[1]$ and a state $\confG$,
executing $\txT[0] \txT[1]$ or $\txT[1] \txT[0]$ from $\confG$
yields different states.
However, under some conditions it is possible to invert the order
of the two transactions, preserving the resulting state.
This is always the case, \eg, of two transactions which operate on 
disjoint sets of tokens.
Lemma~\ref{lem:tx-concurrent} establishes sufficient conditions 
for preserving the state upon reordering.
Besides the case cited before, this is always possible if
both transactions are deposits, or if they are bot redeems
(case~\ref{lem:tx-concurrent:1} of the statement).
Note that, in these cases, the assumption that $\txT[0] \txT[1]$ is enabled
in $\confG$ implies that also $\txT[1] \txT[0]$ is such.
This is no longer true when one of the two transactions is a deposit
and the other one is a redeem.
For instance, if $\txT[1]$ redeems the minted tokens obtained upon a deposit $\txT[0]$, then $\txT[1]$ may not be enabled in $\confG$ because there are not
enough minted tokens in the user's wallet.
Therefore, case~\ref{lem:tx-concurrent:2} of the statement uses the additional
hypothesis that also $\txT[1] \txT[0]$ is enabled in $\confG$.

\begin{lem}[Reordering of transactions]
  \label{lem:tx-concurrent}
  Let $\confG \xrightarrow{\txT[0] \txT[1]} \confG[01]$.
  Then:
  \begin{enumerate}[(a)]

  \item \label{lem:tx-concurrent:1}
    if $\txTok{\txT[0]} \cap \txTok{\txT[1]} = \emptyset$
    or $\txType{\txT[0]} = \txType{\txT[1]} \in \setenum{\ammDepositOp,\ammRedeemOp}$,
    then $\confG \xrightarrow{\txT[1] \txT[0]} \confG[01]$;
    
  \item \label{lem:tx-concurrent:2}
    otherwise, if $\txType{\txT[0]},\txType{\txT[1]} \neq \ammSwapOp$
    and $\confG \xrightarrow{\txT[1] \txT[0]} \confG[10]$,
    then $\confG[01] = \confG[10]$.
   
  \end{enumerate}
\end{lem}

As we shall see in \autoref{sec:arbitrage}, 
it is actually desirable, and crucial for the economic mechanism of AMMs, 
that swaps interfere with other transactions 
that trade the same token type.

\subsection{Additivity of deposit and redeem actions}

Deposit and redeem actions satisfy an additivity property:
if a user performs two successive deposits (resp.\ redeems) on an AMM,
then the same result can be obtained through a single deposit (resp.\ redeem).
Instead, swap actions are not additive, in general:
we will study sufficient conditions for the additivity 
of swap actions in \autoref{sec:swap-rate} (see Theorem~\ref{thm:sr-additivity}).

\begin{thm}[Additivity]
  \label{thm:additivity}
  Let $\confG \xrightarrow{\txT[0]} \confG[0] \xrightarrow{\txT[1]} \confG[1]$.
  Then:
  \begin{enumerate}

  \item \label{thm:additivity:dep}
    if 
    $\txT[0] = \actAmmDeposit{\pmvA}{\valV[0]}{\tokT[0]}{\valV[1]}{\tokT[1]}$ 
    and 
    $\txT[1] = \actAmmDeposit{\pmvA}{\valVi[0]}{\tokT[0]}{\valVi[1]}{\tokT[1]}$, 
    then:
    \[
      \confG \xrightarrow{\actAmmDeposit{\pmvA}{\valV[0]+\valVi[0]}{\tokT[0]}{\valVi[1]+\valVi[1]}{\tokT[1]}} \confG[1]
    \]

  \item \label{thm:additivity:rdm}
    if 
    $\txT[0] = \actAmmRedeem{\pmvA}{\valV:\tokT}$ 
    and 
    $\txT[1] = \actAmmRedeem{\pmvA}{\valVi:\tokT}$, 
    then:
    \[
      \confG \xrightarrow{\actAmmRedeem{\pmvA}{\valV+\valVi:\tokT}} \confG[1]
    \]

  \end{enumerate}
\end{thm}

\subsection{Reversibility of deposit and redeem actions}

The following theorem establishes that
deposit and redeem transactions are \emph{reversible}:
more precisely, 
the effect of a deposit action can be reverted by a redeem action, 
and \emph{vice versa}, 
the effect of a redeem action can be reverted by a deposit action.
The only exception is a deposit action that creates an AMM,
through the rule \nrule{[Dep0]}.
Swap actions are not reversible, in general:
we will study sufficient conditions for their reversibility 
in \autoref{sec:swap-rate} (see Theorem~\ref{thm:sr-reversibility}).

\begin{thm}[Reversibility]
  \label{thm:reversibility}
  Let $\confG \xrightarrow{\txT} \confGi$,
  where $\txType{\txT} \in \setenum{\ammDepositOp,\ammRedeemOp}$ and 
  for all $\tokT \in \TokU[1]$, 
  if $\supply[\confG]{\tokT} = 0$ then $\supply[\confGi]{\tokT} = 0$. 
  Then there exists $\txT^{-1}$ such that
  $\confGi \xrightarrow{\txT^{-1}} \confG$.
\end{thm}

In general, the study of reversible computation models,
which dates back to~\cite{Bennett73lrc},
is an active area of research, 
which has led to a wide range of applications in software systems~\cite{Mezzina20cost}.
In particular, the reversibility of AMM actions has useful consequences
on their behaviour.
For instance, it guarantees that,
starting from a ``stable'' state where no arbitrage is possible,
after any transaction it is possible to return to the stable state.
More in general, if the swap rate function satisfies the conditions
of \autoref{sec:swap-rate} that ensure the additivity and reversibility
also for $\ammSwapOp$ actions, then for 
any sequence of transactions:
\[
\confG[0]
\; \xrightarrow{\txT[1]} \;
\confG[1]
\; \xrightarrow{\txT[2]} \;
\cdots
\confG[n]
\]
it is possible to fire another transaction and return
to the state $\confG[0]$.
Indeed, by additivity we obtain that the effect of the sequence
$\txT[1] \cdots \txT[n]$ can be emulated by a single transaction $\txT$,
and then reversibility ensures that $\txT$ can be reversed, \ie:
\[
\confG[0]
\; \xrightarrow{\txT} \;
\confG[n]
\; \xrightarrow{\txT^{-1}} \;
\confG[0]
\]

\section{The swap rate function}
\label{sec:swap-rate}

In the previous section we have established some key structural properties
of deposit and redeem actions, \eg their additivity and reversibility.
In general, these properties do not hold for swap actions:
it is easy to find swap rate functions $\SX{} \in \RNN^3  \rightarrow \RNN$
that make these properties false.
Throughout this section we introduce some general properties 
of swap rate functions,
and we discuss the properties they induce on the behaviour of AMMs.
In sections \ref{sec:swap-rate:const-sum}-\ref{sec:swap-rate:const-mean}
we then discuss the properties enjoyed by the swap rate functions
used in some concrete AMM implementations.
Coherently with these implementations,
we assume that a swap rate function is defined and non-negative for all $x>0$,
and that the internal exchange rate
(\ie, the limit of $\SX{}$ for $x$ leading to $0$, see~\eqref{eq:exchange-rate:internal})
is always defined.

\subsection{Output-boundedness}

Output boundedness guarantees that an AMM has always enough output tokens 
$\tokT[1]$ to send to the user who performs a  
$\actAmmSwapExact{}{}{x}{\tokT[0]}{\tokT[1]}$.

\begin{defi}[Output-boundedness]
  \label{def:sr-output-bound}
  A swap rate function $\SX{}$ is \emph{output-bounded} when, 
    for all $x,\ammR[0],\ammR[1]$ such that
    $x \geq 0$ and $\ammR[0], \ammR[1] > 0$:
  \[
  x \cdot \SX{x,\ammR[0],\ammR[1]} < \ammR[1]
  \]
\end{defi}

The following lemma establishes sufficient conditions
for a $\ammSwapOp$ action to be enabled.

\begin{lem}
  \label{lem:sr-output-bound}
  Let $\txT = \actAmmSwapExact{\pmvA}{}{x}{\tokT[0]}{\tokT[1]}$, and 
  let $\walA{\tokBal} \in \confG$.
  If $\supply[\confG]{\tokM{\tokT[0]}{\tokT[1]}} > 0$,
  $\tokBal(\tokT[0]) \geq x$ and
  $\SX{}$ is output-bounded, then $\txT$ is enabled in $\confG$.
\end{lem}

\subsection{Monotonicity}

Consider a transaction
\mbox{$\actAmmSwapExact{\pmvA}{}{x}{\tokT[0]}{\tokT[1]}$}
on an AMM $\amm{\ammR[0]:\tokT[0]}{\ammR[1]:\tokT[1]}$.
Without making any assumptions on the swap rate function,
there is no relation between the effect of this transaction
and that of a swap where the parameters have been varied.
Monotonicity, instead, ensures that there exists a meaninful relation:
the swap rate increases if
we decrease the input amount $x$ or the reserves of $\tokT[0]$,
and if we  increase the reserves of $\tokT[1]$.
The intuition is that
lower reserves of $\tokT[0]$ in the AMM 
make the $x:\tokT[0]$ paid by $\pmvA$ more ``valuable'' for the AMM,
hence the AMM will output more units of $\tokT[1]$ for the same input amount.
Increasing the reserves of $\tokT[1]$ in the AMM
(keeping those of $\tokT[0]$ unaltered) produces the same effect.
Monotonicity on $x$ also ensures that
the internal exchange rate of the AMM is defined, for each token pair.

\begin{defi}[Monotonicity]
  \label{def:sr-monotonicity}
  A swap rate function $\SX{}$ is \emph{monotonic} when:
  \[
    x' \leq x,\; \ammRi[0] \leq \ammR[0],\; \ammR[1] \leq \ammRi[1]
    \implies
    \SX{x,\ammR[0],\ammR[1]}
    \leq
    \SX{x',\ammRi[0],\ammRi[1]}
  \]
  Further, $\SX{}$ is \emph{strictly monotonic} when, 
  for $i \in \setenum{0,1,2}$ and $\lhd_i \in \setenum{<,\leq}$:
  \[
    x' \lhd_0 x,\; \ammRi[0] \lhd_1 \ammR[0],\; \ammR[1] \lhd_2 \ammRi[1]
    \implies
    \SX{x,\ammR[0],\ammR[1]}
    \lhd_3
    \SX{x',\ammRi[0],\ammRi[1]}
  \]
  where:
  \[
    \lhd_3 = \begin{cases}
      \leq & \text{if $\lhd_i = \,\leq$ for $i \in \setenum{0,1,2}$} \\
      < & \text{otherwise}
    \end{cases}
  \]
\end{defi}

\noindent
Note that strict monotonicity trivially implies monotonicity.
The following lemma relates monotonicity of the swap rate function
with the gain of swap transactions,
concretising the intuition given before from the point of view of $\pmvA$'s gain.

\begin{lem}
  \label{lem:swap-gain-monotonicity}
  Let \mbox{$\confG = \amm{\ammR[0]:\tokT[0]}{\ammR[1]:\tokT[1]} \mid \confD$}
  and \mbox{$\confGi = \amm{\ammRi[0]:\tokT[0]}{\ammRi[1]:\tokT[1]} \mid \confD$},
  with $\ammRi[0] \leq \ammR[0]$ and $\ammR[1] \leq \ammRi[1]$,
  and let $\txT = \actAmmSwapExact{\pmvA}{}{x}{\tokT[0]}{\tokT[1]}$
  be enabled in $\confG$ and in $\confGi$.
  If $\SX{}$ is monotonic, then
  $\gain[\confG]{\pmvA}{\txT} \leq \gain[\confGi]{\pmvA}{\txT}$.
\end{lem}

\subsection{Additivity}

To extend the additivity property of Theorem~\ref{thm:additivity}
to swap actions, we must require that the swap rate function is additive.

\begin{defi}[Additivity]
  \label{def:sr-additivity}
  A swap rate function $\SX{}$ is \emph{additive} when:
  \[
    \valSXa = \SX{x,\ammR[0],\ammR[1]},\;
    \valSXb = \SX{y,\ammR[0]+x,\ammR[1]-\valSXa x}
    \implies
    \SX{x+y,\ammR[0],\ammR[1]} = 
    \frac{\valSXa x + \valSXb y}{x+y}
  \]
\end{defi}

The idea here is that a user fires a swap transaction (say, $\txT[0]$)
with input amount $x$ in a state $\confG$,
and then in the state reached after firing $\txT[0]$,
she fires another swap transaction (say, $\txT[1]$)
with input amount $y$ on the same AMM.
The definition of additivity requires that the swap rate
of a swap transaction with input amount $x+y$ in $\confG$
is in a given relation with
the swap rates computed for $\txT[0]$ and $\txT[1]$
and with the input amounts $x$ and $y$.
Theorem~\ref{thm:sr-additivity} states that if this relation holds,
then a single swap with input amount $x+y$ in $\confG$ produces exactly
the same effect of performing first $\txT[0]$ and then $\txT[1]$.
Then, Lemma~\ref{lem:swap-gain:additivity} allows us to compute
the gain of this transaction as the sum of the gains of $\txT[0]$ and $\txT[1]$.

\begin{thm}[Additivity of swap]
  \label{thm:sr-additivity}
  Let $\confG \xrightarrow{\txT[0]} \confG[0] \xrightarrow{\txT[1]} \confG[1]$,
  with
  \mbox{$\txT[i] = \actAmmSwapExact{\pmvA}{}{x_i}{\tokT[0]}{\tokT[1]}$} 
  for $i \in \setenum{0,1}$.
  If $\SX{}$ is additive, then:
  \[
    \confG \xrightarrow{\actAmmSwapExact{\pmvA}{}{x_0+x_1}{\tokT[0]}{\tokT[1]}} \confG[1]
  \]
\end{thm}

\begin{lem}[Additivity of swap gain]
  \label{lem:swap-gain:additivity}
  Let \mbox{$\txT(x) = \actAmmSwapExact{\pmvA}{}{x}{\tokT[0]}{\tokT[1]}$},
  and let $\confG \xrightarrow{\txT(x_0)} \confGi$.
  If $\SX{}$ is output-bounded and additive, then:
  \[
    \gain[\confG]{\pmvA}{\txT(x_0+x_1)} 
    \; = \;
    \gain[\confG]{\pmvA}{\txT(x_0)} + \gain[\confGi]{\pmvA}{\txT(x_1)} 
  \]
\end{lem}

\subsection{Reversibility}

The reversibility property in Theorem~\ref{thm:reversibility}
states that the effect of deposit and redeem transactions can be reverted.
We now devise a property of swap rate functions that give the same
guarantee for swap transactions.

\begin{defi}[Reversibility]
  \label{def:sr-reversibility}
  A swap rate function $\SX{}$ is \emph{reversible} when:
  \[
    \valSXa = \SX{x,\ammR[0],\ammR[1]}
    \implies
    \SX{\valSXa x,\ammR[1]-\valSXa x ,\ammR[0]+x} 
    \; = \;
    \frac{1}{\valSXa}
  \]
\end{defi}

Consider now a state
\mbox{$\confG = \amm{\ammR[0]:\tokT[0]}{\ammR[1]:\tokT[1]} \mid \confD$},
and let $\valSXa = \lim_{x \rightarrow 0} \SX{x,\ammR[0],\ammR[1]}$
be the internal exchange between $\tokT[0]$ and $\tokT[1]$ in $\confG$.
If the swap rate function is reversible, then:
\[
\lim_{x \rightarrow 0} \SX{x,\ammR[1],\ammR[0]}
\; = \;
\lim_{x \rightarrow 0} \SX{\valSXa x,\ammR[1] - \valSXa x,\ammR[0] + x}
 \; = \;
\lim_{x \rightarrow 0} \, \frac{1}{\valSXa}
\; = \;
\frac{1}{\valSXa}
\]
from which we obtain:
\begin{equation}
  \label{eq:sr-reversibility:internal-exchange-rate}
  \X[\confG]{\tokT[1],\tokT[0]}
  \; = \;
  \frac{1}{\X[\confG]{\tokT[0],\tokT[1]}}
\end{equation}

The intuition of Definition~\ref{def:sr-reversibility} is that,
to reverse the effect of a swap transaction $\txT$ 
that pays $x:\tokT[0]$ to receive $y:\tokT[1]$,
one must fire a swap transaction $\txT^{-1}$
that pays $y:\tokT[1]$ to receive $x:\tokT[0]$.
Of course, this results in the same AMM state that we had before performing $\txT$.
Writing $\valSXa$ for the swap rate $\SX{x,\ammR[0],\ammR[1]}$,
the \nrule{[Swap]} rule fixes $y = \valSXa x$.
Hence, assuming that in the initial state the AMM has reserves
$\ammR[0]:\tokT[0]$ and $\ammR[1]:\tokT[1]$, after performing $\txT$
its reserves will be $\ammR[0]+x:\tokT[0]$ and
$\ammR[1]-\valSXa x:\tokT[1]$.
In this state, requiring that the swap rate for an input of
$y:\tokT[1]$ is $\tfrac{1}{\valSXa}$
(as done by Definition~\ref{def:sr-reversibility})
implies that the AMM outputs $x:\tokT[0]$,
reverting the reserves of the AMM to the initial values.

The following theorem formalises the intuition above, establishing that, when 
the swap rate function is reversible,
$\ammSwapOp$ transactions are reversible.
Together with Theorem~\ref{thm:reversibility}, 
all the AMM actions are reversible under this hypothesis.

\begin{thm}[Reversibility of swap]
  \label{thm:sr-reversibility}
  Let $\txT = \actAmmSwapExact{\pmvA}{}{x}{\tokT[0]}{\tokT[1]}$, and 
  let $\confG \xrightarrow{\txT} \confGi$.
  If $\SX{}$ is reversible, 
  then there exists $\txT^{-1}$ such that
  $\confGi \xrightarrow{\txT^{-1}} \confG$.
\end{thm}

Lemma~\ref{lem:swap-gain:reversibility} allows us to compute the
gain of the reverse transaction $\txT^{-1}$ in the state
reached after performing $\txT$
as a function of the gain of $\txT$.
As expected by preservation of the global net worth,
the gain of $\txT^{-1}$ is the opposite of that of $\txT$.

\begin{lem}
  \label{lem:swap-gain:reversibility}
  Let $\txT = \actAmmSwapExact{\pmvA}{}{x}{\tokT[0]}{\tokT[1]}$, and
  let $\confG \xrightarrow{\txT} \confGi$.
  If $\SX{}$ is reversible, then
  $\gain[\confG]{\pmvA}{\txT} = -\gain[\confGi]{\pmvA}{\txT^{-1}}$.
\end{lem}

\subsection{Homogeneity}

A swap rate function is homogeneous when the swap rate 
is not affected by a linear scaling of the three parameters.
Homogeneity is useful to relate the swap rate 
before and after deposit or redeem transactions, since their
effect is a linear scaling of the AMM reserves.
Lemma~\ref{lem:sr-homogeneity:dep-rdm} establishes one the the landmark
properties of AMMs we have anticipated in~\autoref{sec:amm}:
when the swap rate function is homogeneous,
deposits and redeems do not affect the internal swap rate.

\begin{defi}[Homogeneity]
  \label{def:sr-homogeneity}
  A swap rate function $\SX{}$ is \emph{homogeneous} when,
  for $a > 0$:
  \[
    \SX{a x, a \ammR[0],a \ammR[1]}
    \; = \;
    \SX{x,\ammR[0],\ammR[1]}
  \]
\end{defi}

\begin{lem}[Preservation of internal exchange rate upon deposits/redeems]
  \label{lem:sr-homogeneity:dep-rdm}
  Let $\confG \xrightarrow{\txT} \confGi$,
  where
  $\txTok{\txT} = \setenum{\tokT[0],\tokT[1]}$
  and
  $\txType{\txT} \in \setenum{\ammDepositOp,\ammRedeemOp}$.
  If $\SX{}$ is homogeneous, then:
  \[
  \X[\confG]{\tokT[0],\tokT[1]} = \X[\confGi]{\tokT[0],\tokT[1]}
  \]
\end{lem}

The following lemma shows that deposits increase swap rates, 
whilst redeems have the opposite effect.
Dually, deposits decrease the slippage, while redeems increase it.
In \autoref{sec:arbitrage} we will exploit this fact to show that
deposits incentivize swaps, while redeems disincentivize them
(see Theorems~\ref{thm:swap-after-dep} and~\ref{thm:swap-after-rdm}).

\begin{lem}
  \label{lem:sr-reduced-slippage}
  Let $\confG \xrightarrow{\txT} \confGi$,
  where
  $\amm{\ammR[0]:\tokT[0]}{\ammR[1]:\tokT[1]} \in \confG$,
  $\amm{\ammRi[0]:\tokT[0]}{\ammRi[1]:\tokT[1]} \in \confGi$
  and $\txTok{\txT} = \setenum{\tokT[0],\tokT[1]}$.
  If $\SX{}$ is homogeneous and strictly monotonic, then for all $x > 0$:
  \begin{enumerate}[(a)]
  \item \label{lem:sr-reduced-slippage:dep}
    \(
    \txType{\txT} = \ammDepositOp\,
    \implies
    \SX{x,\ammR[0],\ammR[1]}
    <
    \SX{x,\ammRi[0],\ammRi[1]}
    \text{ and }
    \SL[\confG]{x,\tokT[0],\tokT[1]}
    >
    \SL[\confGi]{x,\tokT[0],\tokT[1]}
    \)
  \item \label{lem:sr-reduced-slippage:rdm}
    \(
    \txType{\txT} = \ammRedeemOp
    \implies
    \SX{x,\ammR[0],\ammR[1]}
    >
    \SX{x,\ammRi[0],\ammRi[1]}
    \text{ and }
    \SL[\confG]{x,\tokT[0],\tokT[1]}
    <
    \SL[\confGi]{x,\tokT[0],\tokT[1]}
    \)
  \end{enumerate}
\end{lem}

It is easy to find swap rate functions that violate
the properties discussed before:
for instance $\SX{x,\ammR[0],\ammR[1]} = \nicefrac{1}{x}$
violates output-boundedness, additivity, reversibility and homogeneity.
In the rest of the section we discuss some notable swap rate functions,
used in actual AMM implementations, showing that they 
satisfy most of our properties.

\subsection{Constant sum swap rate}
\label{sec:swap-rate:const-sum}

The \emph{constant sum} function
mandates the sum of the token reserves in an AMM to remain constant,
\ie $\ammR[0] + \ammR[1] = k$, where the constant $k$ is fixed
upon the first deposit in the AMM.

\begin{thm}[Constant sum swap rate]
  \label{def:const-sum}
  \label{thm:const-sum}
  The \emph{constant sum} swap rate function:  
  \[
  \SX{x,\ammR[0],\ammR[1]}
  \; = \;
  1
  \]
  is 
  monotonic,
  reversible,
  additive, 
  and homogeneous.
  Furthermore, 
  its internal swap rate and its slippage are given by:
  \[
  \X[\confG]{\tokT[0],\tokT[1]}
  \; = \;
  1
  \qquad\qquad
  \SL[\confG]{x,\tokT[0],\tokT[1]}
  \; = \;
  0
  \]  
\end{thm}

Note that the constant sum function is \emph{not} output-bounded,
since the output amount may exceed the reserves of the output token. 
A positive aspect of constant sum AMMs
is that they do not suffer from slippage.
With constant sum AMMs, the internal exchange rate is always $1$,
and so there is zero slippage (see \Cref{eq:exchange-rate:internal,eq:slippage}).
A negative aspect is that constant sum AMMs do not allow
the token reserves to grow unboundedly:
indeed, the bound is fixed with the first deposit.
This makes constant sum AMMs unsuitable for scenarios where
one wants the liquidity of the AMM to increase over time,
and to incentivise users to deposit through minted tokens.
When the oracle and internal exchange rates are not aligned
(\ie, when the prices of the two tokens are different), then
rational users will drain the reserves
of the most expensive token type held by the AMM.
Despite these drawbacks, the constant sum swap rate is suitable situations
where the two token types in the AMM are supposed to be equally prices,
like for stablecoins.
This is the case \eg for mStable~\cite{mStable}.

\subsection{Constant product swap rate}
\label{sec:swap-rate:const-prod}

The constant product swap rate function
(introduced before in Definition~\ref{def:const-prod})
enjoys all the properties discussed previously in this section.%
\footnote{The existence of other classes of swap rate functions enjoying all the six properties is an open question.}

\begin{thm}[Constant product]
  \label{lem:swap-rate:const-prod}
  The constant product swap rate function is 
  output-bounded,
  strictly monotonic,
  reversible, 
  additive, and
  homogeneous.
  Furthermore, 
  its internal swap rate and its slippage are given by:
  \[
  \X[\confG]{\tokT[0],\tokT[1]}
  \; = \;
  \frac{\ammR[1]}{\ammR[0]}
  \qquad\qquad
  \SL[\confG]{x,\tokT[0],\tokT[1]}
  \; = \;
  \frac{x}{\ammR[0]}
  \]
\end{thm}

Compared to the constant sum swap rate, a point in favour of the
constant product is output-boundedness,
which allows users to add unbounded liquidity to the AMM.
A point against is slippage, which grows linearly with the amount
of the input token.
Therefore, when the internal exchange rate is aligned with the oracle's,
users are disincentivised from performing large swaps.
The most prominent AMM platform adopting the constant product
is Uniswap v2~\cite{uniswapimpl}.
Curve~\cite{curve} uses a hybrid swap rate function,
which approximates a constant sum for an interval of input values $x$, and 
behaves as a constant product outside the interval.
In this way, it achieves a small slippage within the interval,
at the same time allowing unbounded liquidity thanks to output-boundedness.

\subsection{Constant mean swap rate}
\label{sec:swap-rate:const-mean}

The constant mean swap rate function generalises the constant product
by associating weights $w_0, w_1 \in \RP$ to the token types held by the AMM,
so to preserve the following equality:
\[
\ammR[0]^{w_0} \ammR[1]^{w_1} = (\ammR[0] + x)^{w_0} (\ammR[1] + y)^{w_1}
\qquad
\text{where } y = x \cdot \SX{x,\ammR[0],\ammR[1]}
\]

The following theorem shows that the constant mean function
enjoys most of the properties of the constant product, except reversibility.

\begin{thm}[Constant mean swap rate]
  \label{def:const-mean}
  \label{thm:const-mean}
  The \emph{constant mean} swap rate function:  
  \[
  \SX{x,\ammR[0],\ammR[1]}
  \; = \;
  \frac{\ammR[1]}{x} \bigg( 1 - \Big(\frac{\ammR[0]}{\ammR[0]+x} \Big)^{\frac{w_0}{w_1}} \bigg)
  \]
  is
  output-bounded,
  monotonic,
  additive, 
  and homogeneous.
  Furthermore, 
  its internal swap rate and its slippage are given by:
  \[
  \X[\confG]{\tokT[0],\tokT[1]}
  \; = \;
  \frac{\ammR[1] w_0}{\ammR[0] w_1}
  \qquad\qquad
  \SL[\confG]{x,\tokT[0],\tokT[1]}
  \; = \;
  \frac{x w_0}{\ammR[0] w_1 \Big( 1 - \big( \frac{\ammR[0]}{\ammR[0]+x} \big)^{\frac{w_0}{w_1}} \Big)} - 1
  \]  
\end{thm}

The most prominent AMM plaform using the constant mean swap rate
is Balancer~\cite{balancerpaper}.
Users fix the weights $w_0,w_1$ of token types when an AMM is created;
once fixed, these weights cannot be changed.
The constant product swap rate can be seen as the special case of the constant
mean where the two weights are equal.

\section{The economic mechanism of AMMs} 
\label{sec:arbitrage}

AMMs can be seen as games where users compete to increase their net worth.
We now study the incentive mechanisms of AMMs from 
a game-theoretic perspective.

\subsection{Arbitrage}

The \keyterm{arbitrage game} is a single-player, single-round game,
where the player can perform a single move on a given AMM pair
$\tokT[0],\tokT[1]$ in order to maximize her gain.
The initial game states have the form
\mbox{$\confG[0] = \walA{\tokBal} \mid \amm{\ammR[0]:\tokT[0]}{\ammR[1]:\tokT[1]} \mid \confD$},
where $\pmvA$ is the player;
the \emph{moves} of $\pmvA$ are all the possible transactions 
that can be fired by $\pmvA$
(we also consider doing nothing as a possible move).
More formally, a move is a sequence $\bcB$ such that
either $\bcB = \emptyseq$ (the empty sequence), 
or $\bcB = \txT$ with $\txWal{\txT} = \pmvA$.
The goal of $\pmvA$ is to maximize her gain
$\gain[{\confG[0]}]{\pmvA}{\bcB}$
on the AMM pair $\tokT[0],\tokT[1]$.
A \emph{solution} to the game is a move $\bcB$ that satisfies such goal.
We study the arbitrage game under the assumption 
that $\pmvA$ holds no minted tokens $\tokM{\tokT[0]}{\tokT[1]}$.
In this way, by Lemma~\ref{lem:swap:gain},
$\pmvA$'s gain only depends on
the input amount of $\pmvA$'s swap,
on the reserves of $\tokT[0]$ and $\tokT[1]$ in the AMM, 
and on their prices.
In practice, AMM users are logically partitioned in two groups,
\eg liquidity providers (who perform deposits and redeems) and traders
(who perform swaps), so basically here we are considering the
arbitrage game from the traders' point of view.
We further assume that $\pmvA$'s balance is enough to allow $\pmvA$
to perform the optimal swap.
This is a common assumption in formulations of the arbitrage game:
in practice, this can be achieved by borrowing the needed amount of 
the input token from a lending pool via 
a flash-loan~\cite{Qin21fc,wang2020towards}.
Theorem~\ref{thm:arbitrage} shows that
a rational agent is incentivized to perform a swap
to realign the internal and the oracle's exchange rate.
The optimal solution to the arbitrage game can be approximated 
by multiple users who swap smaller amounts than the optimal one.

Before devising a solution to the arbitrage game, 
we examine the potential candidates for the solution.
Observe that doing nothing (\ie, $\bcB = \emptyseq$) has clearly zero gain,
as well as depositing or redeeming, 
as established by Lemma~\ref{lem:amm:W-preservation}.
Hence, if one of such moves is a solution, so are the other two:
without loss of generality, we assume that $\pmvA$'s move will be
$\bcB = \emptyseq$ when there is no strategy which allows 
$\pmvA$ to increase her gain.
%

We first show in Lemma~\ref{lem:swap:gain-0-xor-1} that, 
if a swap with input $\tokT[0]$ and output $\tokT[1]$ has a positive gain,
then a swap with input $\tokT[1]$ and output $\tokT[0]$
will have a negative gain, whatever input amount is chosen.
This holds whenever the swap rate function is monotonic and reversible.
Lemma~\ref{lem:SX-X} is instrumental to prove Lemma~\ref{lem:swap:gain-0-xor-1},
as it finds the needed relation between the swap rate function and the exchange rate.
Passing from this relation to the gain of the swap transaction
is obtained by means of Lemma~\ref{lem:swap:gain-SX-X}.

\begin{lem}
  \label{lem:SX-X}
  If $\SX{}$ is strictly monotonic and reversible, then for all $x > 0$:
  \[
    \SX{x,\ammR[0],\ammR[1]} \geq 
    \X{\tokT[0],\tokT[1]}
    \implies
    \forall y > 0.\;
    \SX{y,\ammR[1],\ammR[0]} <
    \X{\tokT[1],\tokT[0]}
  \]
\end{lem}

\begin{lem}[Unique direction for swap gain]
  \label{lem:swap:gain-0-xor-1}
  Let 
  $\confG = \walA{\tokBal} \mid \amm{\ammR[0]:\tokT[0]}{\ammR[1]:\tokT[1]} \mid \confD$
  be such that $\tokBal{\tokM{\tokT[0]}{\tokT[1]}} = 0$, and
  let $\txT[d](x) = \actAmmSwapExact{\pmvA}{}{x}{\tokT[d]}{\tokT[1-d]}$,
  for $x > 0$ and $d \in \setenum{0,1}$.
  If $\SX{}$ is output-bounded, strictly monotonic and reversible, then
  for all $y>0$ such that $\tokBal{\tokT[1-d]} \geq y$:
  \[
    \gain[\confG]{\pmvA}{\txT[d](x)} > 0
    \implies
    \gain[\confG]{\pmvA}{\txT[1-d](y)} < 0
  \]
\end{lem}

Theorem~\ref{thm:arbitrage}
devises a general solution to the arbitrage game,
determining the swap transaction that maximizes $\pmvA$'s gain.
This is the transaction
\mbox{$\actAmmSwapExact{\pmvA}{}{x_0}{\tokT[0]}{\tokT[1]}$}
such that, in the state $\confGi$ reached after performing it
from the initial state, 
the internal exchange rate between $\tokT[0]$ and $\tokT[1]$
is aligned to the oracle's exchange rate.
By Lemma~\ref{lem:swap:gain-SX-X}, 
no move from $\confGi$ can increase $\pmvA$'s gain,
\ie the solution for the arbitrage game in $\confGi$ is to do nothing.
Lemma~\ref{lem:swap:gain-0-xor-1} guarantees that
swaps in the other direction are not solutions,
since they decrease $\pmvA$'s gain.
Note that 
if the internal exchange rate is already aligned to the oracle's,
or if $\pmvA$ has not enough balance to perform the optimal swap,
then the solution to the arbitrage problem is to do nothing.

\begin{thm}[Arbitrage]
  \label{thm:arbitrage}  
  Let $\confG = \walA{\tokBal} \mid \amm{\ammR[0]:\tokT[0]}{\ammR[1]:\tokT[1]} \mid \confD$ 
  be such that $\tokBal{\tokM{\tokT[0]}{\tokT[1]}} = 0$.
  For all $x > 0$, 
  let $\txT(x) = \actAmmSwapExact{\pmvA}{}{x}{\tokT[0]}{\tokT[1]}$.
  Let $x_0$ be such that:
  \begin{equation}
    \label{eq:arbitrage:max:x0}
    \X[\confGi]{\tokT[0],\tokT[1]} = \X{\tokT[0],\tokT[1]}
    \qquad
    \text{ where }
    \confG \xrightarrow{\txT(x_0)} \confGi
  \end{equation}
  If $\SX{}$ is output-bounded, strictly monotonic, additive and reversible, then:
  \[
    \forall x \neq x_0
    \; : \;
    \gain[{\confG}]{\pmvA}{\txT(x_0)}
    \; > \;
    \gain[{\confG}]{\pmvA}{\txT(x)}
  \]
  Furthermore, if an $x_0$ satisfying~\Cref{eq:arbitrage:max:x0} exists,
  it is unique.
\end{thm}

An implicit desideratum on these solutions is that,
given a specific instance of the swap rate function,
they are efficiently computable:
this is the case, \eg, for the constant product,
for which Lemma~\ref{lem:arbitrage:const-prod} finds a closed formula
for the arbitrage solution.

\begin{lem}[Arbitrage and constant product]
  \label{lem:arbitrage:const-prod}
  Let $\confG = \walA{\tokBal} \mid \amm{\ammR[0]:\tokT[0]}{\ammR[1]:\tokT[1]}$,
  and let:
  \begin{equation}
    \label{eq:arbitrage:const-prod}
    x_0 
    \; = \;
    \sqrt{\frac{\exchO{\tokT[1]}}{\exchO{\tokT[0]}} \ammR[0] \ammR[1]} - \ammR[0]
  \end{equation}
  If $\SX{}$ is the constant product swap rate and $x_0 > 0$,
  then $\actAmmSwapExact{\pmvA}{}{x_0}{\tokT[0]}{\tokT[1]}$
  is the solution to the arbitrage game in $\confG$.
\end{lem}

\begin{exa}
  \label{ex:arbitrage}
  Consider an initial state
  \(
  \confG =
  \walA{\tokBal} \mid
  \amm{18:\tokT[0]}{6:\tokT[1]} \mid \confD
  \).
  Assuming the constant product swap rate,
  and $\exchO{\tokT[0]} = 3$, $\exchO{\tokT[1]} = 4$,
  we have that:
  \begin{align*}
    & \X[\confG]{\tokT[0],\tokT[1]}
      \; = \;
      6/18
      \; < \;
      3/4
      \; = \;
      \X{\tokT[0],\tokT[1]}
    \\[5pt]
    & \X[\confG]{\tokT[1],\tokT[0]} 
      \; = \;
      18/6
      \; > \;
      4/3
      \; = \;
      \X{\tokT[1],\tokT[0]}
  \end{align*}
  By Theorem~\ref{thm:arbitrage} it follows that 
  the solution to the arbitrage game is 
  $\txT(x) = \actAmmSwapExact{\pmvA}{}{x}{\tokT[1]}{\tokT[0]}$, 
  for suitable $x$.
  By Lemma~\ref{lem:arbitrage:const-prod}, we find that 
  the optimal input value is:
  \[
  x_1
  \; = \;
  \sqrt{\frac{3}{4} \cdot 18 \cdot 6} - 6
  =
  3
  \]
  and the corresponding output value is
  $x_1 \cdot \SX{x_1,6,18} = 6$.
  We then obtain:
  \[
  \confG
  \xrightarrow{\txT(x_1)}
  \confGi
  \; = \;
  \walA{\tokBal - 3:\tokT[1] + 6:\tokT[0]} \mid
  \amm{12:\tokT[0]}{9:\tokT[1]}
  \]
  This action maximizes $\pmvA$'s gain
  $\gain[\confG]{\pmvA}{\txT(x_1)} = W_{\pmvA}(\confGi) - W_{\pmvA}(\confG) = 6 \exchO{\tokT[0]} - 3 \exchO{\tokT[1]} = 6$.
  Any other action would result in a lower gain for $\pmvA$.
  Note that the internal exchange rate in $\confGi$ is aligned to the oracle's:
  $\X[\confGi]{\tokT[0],\tokT[1]} = 9/12 = 3/4 = \X{\tokT[0],\tokT[1]}$.
  \hfill\qedex
\end{exa}

\subsection{Swaps after deposits}

We show in Theorem~\ref{thm:swap-after-dep} that deposits incentivise swaps.
Namely, if a user $\pmvB$ performs a deposit on an AMM for the token pair 
$\tokT[0],\tokT[1]$,
and then a \emph{different} user $\pmvA$ performs a swap in the resulting state,
then $\pmvA$'s gain is increased \wrt the gain that she would have
obtained by performing the same transaction \emph{before} $\pmvB$'s deposit.
The intuition is that larger amounts of tokens in an AMM
provide decrease the slippage, therefore attracting users interested in swaps.

\begin{thm}[Swap after deposit]
  \label{thm:swap-after-dep}
  Let $\txT[\ammSwapOp]$ and $\txT[\ammDepositOp]$ be two transactions
  such that 
  $\txWal{\txT[\ammSwapOp]} = \pmvA \neq \txWal{\txT[\ammDepositOp]}$ and,
  for $\ell \in \setenum{\ammSwapOp,\ammDepositOp}$, 
  $\txType{\txT[\ell]} = \ell$ and
  $\txTok{\txT[\ell]} = \setenum{\tokT[0],\tokT[1]}$.
  Let $\confG$ be such that both
  $\txT[\ammSwapOp]$ and $\txT[\ammDepositOp]\txT[\ammSwapOp]$
  are enabled in $\confG$.
  If the swap rate function is homogeneous and strictly monotonic, then:
  \[
    \gain[\confG]{\pmvA}{\txT[\ammDepositOp]\txT[\ammSwapOp]}
    >
    \gain[\confG]{\pmvA}{\txT[\ammSwapOp]}
  \]
\end{thm}

\begin{exa}
  \label{ex:swap-after-dep}
  Let $\confG = \walA{5:\tokT[0]} \mid \amm{5:\tokT[0]}{10:\tokT[1]} \mid \confD$,
  let $\txT[\ammDepositOp] = \actAmmDeposit{\pmvB}{40}{\tokT[0]}{80}{\tokT[1]}$,
  and let $\txT[\ammSwapOp] = \actAmmSwapExact{\pmvA}{}{5}{\tokT[0]}{\tokT[1]}$.
  Assuming the constant product swap rate,
  we have that:
  \begin{align*}
    & \confG \xrightarrow{\txT[\ammSwapOp]}
    \confG[s] = \walA{5:\tokT[1]} \mid \amm{10:\tokT[0]}{5:\tokT[1]} \mid \confD
    \\
    & \confG \xrightarrow{\txT[\ammDepositOp]}
    \confG[d] = \walA{5:\tokT[0]} \mid \amm{45:\tokT[0]}{90:\tokT[1]} \mid \confDi
    \xrightarrow{\txT[\ammSwapOp]}
    \confG[ds] = \walA{9:\tokT[1]} \mid \amm{50:\tokT[0]}{81:\tokT[1]} \mid \confDi    
  \end{align*}
  Now, assuming $\exchO{\tokT[0]} = 1$ and $\exchO{\tokT[1]} = 1$,
  we have the following gains for $\pmvA$:
  \[
  \gain[\confG]{\pmvA}{\txT[\ammDepositOp]\txT[\ammSwapOp]}
  \; = \;  
  4
  \; > \;
  0  
  \; = \;
  \gain[\confG]{\pmvA}{\txT[\ammSwapOp]}
  \]
  as correctly predicted by Theorem~\ref{thm:swap-after-dep}.
  Note that in the state $\confG$ before the deposit,
  $\pmvA$ has zero gain from her swap,
  while the same transaction has a positive gain after the deposit.
  \hfill\qedex
\end{exa}

Theorem~\ref{thm:dep-arbitrage} finds
the solution of the arbitrage game after a deposit of another user.
More precisely, let $\bcB$ be the solution in $\confG$,
and let $\bcBi$ be the solution in the state $\confGi$
reached after a deposit. 
If $\bcB$ is empty, then also $\bcBi$ is such.
If $\bcB$ is a $\ammSwapOp$ with input $\tokT[0]$ and output $\tokT[1]$,
then also $\bcBi$ is such (but for the input amount).

\begin{thm}[Arbitrage after deposit]
  \label{thm:dep-arbitrage}
  Let $\confG = \walA{\tokBal} \mid \amm{\ammR[0]:\tokT[0]}{\ammR[1]:\tokT[1]} \mid \confD$,
  and let:
  \[
    \confG
    \xrightarrow{\actAmmDeposit{\pmvB}{\valV[0]}{\tokT[0]}{\valV[1]}{\tokT[1]}} 
    \confG[d]
    \qquad
    \text{where $\confG[d] = \walA{\tokBali} \mid \amm{\ammRi[0]:\tokT[0]}{\ammRi[1]:\tokT[1]} \mid \confDi$ and $\pmvB \neq \pmvA$}
  \]
  Let $\bcB$ and $\bcB[d]$ be the solutions of the arbitrage game
  in $\confG$ and in $\confG[d]$, respectively.
  If $\SX{}$ is output-bounded, strictly monotonic, additive, reversible,
  and homogeneous, then:
  \begin{enumerate}

  \item \label{thm:dep-arbitrage:swap}
    if $\bcB = \actAmmSwapExact{\pmvA}{}{x}{\tokT[0]}{\tokT[1]}$, then 
    \[
      \bcB[d] = \actAmmSwapExact{\pmvA}{}{a x}{\tokT[0]}{\tokT[1]}
      \qquad
      \gain[{\confG[d]}]{\pmvA}{\bcB[d]} = a \gain[\confG]{\pmvA}{\bcB}
      \qquad
      \text{where }\; 
      a = \frac{\ammR[1]+\valV[1]}{\ammR[1]}
    \]

  \item \label{thm:dep-arbitrage:emptyseq}
    if $\bcB = \emptyseq$, then $\bcB[d] = \emptyseq$.

  \end{enumerate}
\end{thm}

\subsection{Swaps after redeems}

We now study swaps and arbitrage after redeems.
Conversely to what we have shown before in Theorem~\ref{thm:swap-after-dep},
we find that redeems disincentivise swaps
(Theorem~\ref{thm:swap-after-rdm}).
Similarly to Theorem~\ref{thm:dep-arbitrage},
if the solution to the arbitrage game in a state $\confG$ is a swap,
then after a redeem in $\confG$ the solution is still a swap
which only differs in the input amount
(Theorem~\ref{thm:rdm-arbitrage}).

\begin{thm}[Swap after redeem]
  \label{thm:swap-after-rdm}
  Let $\txT[\ammSwapOp]$ and $\txT[\ammRedeemOp]$ be two transactions
  such that 
  $\txWal{\txT[\ammSwapOp]} = \pmvA \neq \txWal{\txT[\ammRedeemOp]}$ and,
  for $\ell \in \setenum{\ammSwapOp,\ammRedeemOp}$, 
  $\txType{\txT[\ell]} = \ell$ and
  $\txTok{\txT[\ell]} = \setenum{\tokT[0],\tokT[1]}$.
  Let $\confG$ be such that both
  $\txT[\ammSwapOp]$ and $\txT[\ammRedeemOp]\txT[\ammSwapOp]$
  are enabled in $\confG$.
  If the swap rate function is homogeneous and strictly monotonic, then:
  \[
    \gain[\confG]{\pmvA}{\txT[\ammRedeemOp]\txT[\ammSwapOp]}
    <
    \gain[\confG]{\pmvA}{\txT[\ammSwapOp]}
  \]
\end{thm}

\begin{thm}[Arbitrage after redeem]
  \label{thm:rdm-arbitrage}
  Let $\confG = \walA{\tokBal} \mid \amm{\ammR[0]:\tokT[0]}{\ammR[1]:\tokT[1]} \mid \confD$,
  and let:
  \[
    \confG
    \xrightarrow{\actAmmRedeem{\pmvB}{\valV:\tokM{\tokT[0]}{\tokT[1]}}}
    \confG[d]
    \qquad
    \text{where $\confG[d] = \walA{\tokBali} \mid \amm{\ammRi[0]:\tokT[0]}{\ammRi[1]:\tokT[1]} \mid \confDi$ and $\pmvB \neq \pmvA$}
  \]
  Let $\bcB$ and $\bcB[d]$ be the solutions of the arbitrage game
  in $\confG$ and in $\confG[d]$, respectively.
  If $\SX{}$ is output-bounded, strictly monotonic, additive, reversible,
  and homogeneous, then:
  \begin{enumerate}

  \item \label{thm:rdm-arbitrage:swap}
    if $\bcB = \actAmmSwapExact{\pmvA}{}{x}{\tokT[0]}{\tokT[1]}$, then 
    \[
      \bcB[d] = \actAmmSwapExact{\pmvA}{}{a x}{\tokT[0]}{\tokT[1]}
      \qquad
      \gain[{\confG[d]}]{\pmvA}{\bcB[d]} = a \gain[\confG]{\pmvA}{\bcB}
      \qquad
      \text{where }\; 
      a = 1 - \frac{\valV}{\supply[\confG]{\tokM{\tokT[0]}{\tokT[1]}}}
    \]

  \item \label{thm:rdm-arbitrage:emptyseq}
    if $\bcB = \emptyseq$, then $\bcB[d] = \emptyseq$.

  \end{enumerate}
\end{thm}

\section{Maximal extractable value}
\label{sec:mev}

Maximal Extractable Value (MEV) refers to a class of attacks to smart contracts
where miners/validators exploit their power to reorder, drop or insert transactions in a block
to ``extract'' value from the \emph{mempool}
(\ie, the set of transactions sent to the blockchain network, but not appearing yet in a block). 
Empirical research has shown that AMMs
are routinely targeted by MEV attacks~\cite{Daian19flash,Qin21quantifying,Zhou21discovery,Zhou21high},
and indeed recent versions of the Ethereum protocol implementation
include a MEV extraction mechanism~\cite{mev-geth}.
This has negative effects on AMM users, as well as on transaction fees
and network congestion.

We show that our AMM model makes it possible to faithfully express MEV attacks.
Consider a constant product AMM for two token types $\tokT[0], \tokT[1]$ with the same price, \eg
$\exchO{\tokT[0]} = \exchO{\tokT[1]} = 1$, and
consider a state:
\[
  \confG 
  \; = \;
  \wal{\pmvM}{\cdots} \mid \wal{\pmvA}{50:\tokT[0]}
  \mid
  \amm{10:\tokT[0]}{10:\tokT[1]}
  \mid \cdots
\]
where we use $\pmvA$ to impersonate a honest user,
and $\pmvM$ for a miner, acting as an adversary. 
By Lemma~\ref{lem:swap:gain-SX-X} we know that the AMM is in equilibrium
in $\confG$, because, for each $x>0$:
\[
\SX{x,10,10} = \frac{10}{10+x} < 1 = \X{\tokT[0],\tokT[1]}
\]
Therefore, neither a miner nor any other user
can increase their net worth in $\confG$.


Assume now that $\pmvA$ sends the transaction
$\txT[\pmvA] = \actAmmSwapExact{\pmvA}{}{50}{\tokT[0]}{\tokT[1]}$
to the blockchain network.
Before being included in a block, $\txT[\pmvA]$ is added to the mempool,
from where miners gather transactions to construct blocks.
Any miner owning enough token units can increase their gain by firing 
$\pmvA$'s transaction within a \emph{sandwich} of $\pmvM$'s swaps. 
For instance, assume that $\pmvM$'s wallet is
$\wal{\pmvM}{40:\tokT[0], 1:\tokT[1]}$.
Then $\pmvM$ can construct a block:
\[
\bcB
\; = \;
\actAmmSwapExact{\pmvM}{}{40}{\tokT[0]}{\tokT[1]}
\;\;
\txT[\pmvA]
\;\;
\actAmmSwapExact{\pmvM}{}{9}{\tokT[1]}{\tokT[0]}
\]
We have that $\confG \xrightarrow{\bcB} \confGi$, where:
\begin{align*}
  \confG \;
  & \xrightarrow{\actAmmSwapExact{\pmvM}{}{40}{\tokT[0]}{\tokT[1]}} 
  \wal{\pmvM}{0:\tokT[0],9:\tokT[1]} \mid \wal{\pmvA}{50:\tokT[0]} \mid
  \amm{50:\tokT[0]}{2:\tokT[1]}
  \mid \cdots
  \\
  & \xrightarrow{\actAmmSwapExact{\pmvA}{}{50}{\tokT[0]}{\tokT[1]}\,} 
  \wal{\pmvM}{0:\tokT[0],9:\tokT[1]} \mid \wal{\pmvA}{0:\tokT[0],1:\tokT[1]} \mid
  \amm{100:\tokT[0]}{1:\tokT[1]}
  \mid \cdots
  \\
  & \xrightarrow{\actAmmSwapExact{\pmvM}{}{9}{\tokT[1]}{\tokT[0]}\;} 
  \wal{\pmvM}{90:\tokT[0],0:\tokT[1]} \mid 
  \wal{\pmvA}{0:\tokT[0],1:\tokT[1]} \mid
  \amm{10:\tokT[0]}{10:\tokT[1]}
  \mid \cdots
  \; = \confGi
\end{align*}
This results in a positive gain for $\pmvM$, since:
\begin{align*}
  \gain[\confG]{\pmvM}{\bcB}
  & = W_{\pmvM}(\confGi) - W_{\pmvM}(\confG)
  = 90 \cdot \exchO{\tokT[0]} - (40 \cdot \exchO{\tokT[0]} + 1 \cdot \exchO{\tokT[1]})
  = 49
  \\
  \gain[\confG]{\pmvA}{\bcB}
  & = W_{\pmvA}(\confGi) - W_{\pmvA}(\confG)
  = 1 \cdot \exchO{\tokT[1]} - 50 \cdot \exchO{\tokT[0]}
  = -49  
\end{align*}
Summing up, $\pmvM$ has managed to extract value from $\pmvA$'s transaction in the mempool, 
improving her gain to the detriment of $\pmvA$'s net worth.

The mechanism of \emph{guarded transactions},
which allows users to specify a lower bound to the amount of tokens
outputted upon a swap (see~\autoref{sec:variants}),
is a partial countermeasure against MEV attacks.
For instance, in the scenario above
$\pmvA$ could have sent a guarded transaction
$\txTi[\pmvA] = \actAmmSwapGuarded{\pmvA}{}{50}{\tokT[0]}{8.3}{\tokT[1]}$,
which would have ensured $\pmvA$ to receive at least $8.3:\tokT[1]$ upon the swap.
This would have neutralised the sandwich attack described before,
since after the first $\pmvM$'s transaction, $\txTi[\pmvA]$ is no longer valid.
Even though guarded transactions mitigate the issue of not knowing the state
where one's transaction will be fired,
they are not a complete defence against MEV attacks.
Indeed, in~\cite{BCL22fc} it is shown that 
adversaries can craft sandwiches 
that extract value from \emph{any} non-empty mempool
of $\ammSwapOp$ and $\ammDepositOp$ (guarded) transactions.
Further analyses the effect of MEV on constant-function AMMs
are developed in~\cite{Kulkarni22arxiv}.
Several approaches to prevent MEV attacks are discussed 
in~\cite{Heimbach22sok,Baum21sok}.

\section{Variants of the basic model}
\label{sec:variants}

Our AMM model abstracts from implementation-specific features, 
and from the features that are orthogonal 
to the core functionality of AMMs (\eg, governance). 
We discuss below some extensions and variants of our model to make it closer to actual
implementations, and their impact on our theory.

\subsection{Fees}
\label{sec:variants:fees}

In actual AMM implementations, 
the swap rate --- and consequently, the semantics of \nrule{[Swap]} actions ---
also depends on a \emph{trading fee} $1 - \fee$.
For instance, incorporating this fee in the 
constant product swap rate function is usually done as follows:
\[
  \SX[\fee]{x,\ammR[0],\ammR[1]}
  \; = \;
  \frac{\fee \, \ammR[1]}{\ammR[0] + \fee \, x}
  \qquad\text{where } 
  \fee \in [0,1]
\]
In this case, when the trading fee is zero (\ie, $\fee = 1$), 
the swap rate preserves the product between AMM reserves;
a higher fee, instead, results in reduced amounts of output tokens
received from swap actions.
Intuitively, the AMM retains a portion of the swapped amounts, 
but the overall reserves are still distributed among all minted tokens, 
thereby increasing the redeem rate of minted tokens.
The structural properties in \autoref{sec:struct-properties} are
not affected by swap fees.

\subsection{Price updates}
\label{sec:variants:price-updates}

An underlying assumption of our model is that the price of atomic tokens 
is constant, and consequently that exchange rates are stable. 
In the wild, prices and exchange rates can vary over time,
possibly making the net worth of users holding minted tokens decrease ---
a phenomenon commonly referred to as \emph{impermanent loss}~\cite{impermanentloss}.

Introducing price updates in our AMM model is straightforward:
it suffices to extend states $\confG$ with price oracles, 
parameterize with $\confG$ the exchange rate $\X{}$,
and extend the AMM semantics with a rule to non-deterministically
update token prices.
Most of the structural properties in \autoref{sec:struct-properties}
would not be affected by this extension: the exceptions are
determinism (Lemma~\ref{lem:amm:determinism}) and
net worth preservation 
(Lemma~\ref{lem:amm:W-preservation}\ref{lem:amm:W-preservation:a},
while part \ref{lem:amm:W-preservation:b} would still be true for deposits
and redeems).
Technically, also the properties about swaps and incentives 
in \autoref{sec:swap-rate} and \autoref{sec:arbitrage} are preserved,
although this happens because most of these properties
assume sequences of deposits, redeems and swaps.
If we allow these actions to be interleaved with price updates, some
properties no longer hold:
notably, the optimality of the solution $\bcB$ to the arbitrage problem 
(Theorem~\ref{thm:arbitrage})
is lost if $\bcB$ is front-run by a price update that
alters the exchange rates, since this affects the condition 
provided by Theorem~\ref{thm:arbitrage}.

In practice, the assumption of constant exchange rates 
assumed by Theorem~\ref{thm:arbitrage} may hold
in the case of exchanges between stable coins \cite{makerdao}.
Here, arbitrage ensures the alignment between swap rates and 
exchange rates, so users are hence incentivized to provide liquidity to AMMs, 
as the redeem rate is likely to increase over time.

\subsection{Guarded transactions}
\label{sec:variants:guarded-transactions}

The semantics of AMMs in \autoref{sec:amm} defines how the state evolves
upon transactions.
In practice, when a user emits a transaction,
she cannot predict the exact state in which it will be actually committed.
This may lead to unexpected or unwanted behaviours.
For instance, the gain of a swap transaction sent by $\pmvA$ 
may be reduced if the transaction is front-run by a redeem transaction
sent by $\pmvB$, as established by Theorem~\ref{thm:swap-after-rdm}.
The problem here is that redeems decrease the swap rate 
(by Lemma~\ref{lem:sr-reduced-slippage}),
and consequently the amount of output tokens received by $\pmvA$.
As a partial countermeasure to this issue, Uniswap allows users to
specify a lower bound $y^{\mathit{min}}$ to the amount of received tokens.
In our model, we could formalise this behaviour by amending
the \nrule{[Swap]} rule as follows:
\[
\irule
{
  \begin{array}{l}
    \tokBal{\tokT[0]} \geq x > 0
    \qquad
    y = x \cdot \SX{x,\ammR[0],\ammR[1]}
    \qquad
    y^{\textit{min}} \leq y < \ammR[1]
  \end{array}
}
{\begin{array}{l}
   \walA{\tokBal}
   \mid
   \amm{\ammR[0]:\tokT[0]}{\ammR[1]:\tokT[1]}
   \mid
   \confG
   \xrightarrow{\actAmmSwapGuarded{\pmvA}{}{x}{\tokT[0]}{y^{\textit{min}}}{\tokT[1]}}
   \\[4pt]
   \walA{\tokBal - x:\tokT[0] + y:\tokT[1]}
   \mid
   \amm{\ammR[0]+x:\tokT[0]}{\ammR[1]-y:\tokT[1]}
   \mid
   \confG
   \hspace{-5pt}
 \end{array}
}
\nrule{[Swap]}
\]

\noindent
Similar countermeasures apply to \nrule{[Rdm]} and \nrule{[Dep]} rules.
For redeems, the user can enforce lower bounds 
$\valV[0]^{\mathit{min}}$, $\valV[1]^{\mathit{min}}$ on the amount
of received tokens $\tokT[0]$, $\tokT[1]$ as follows:
\[
\irule{
   \tokBal{\tokM{\tokT[0]}{\tokT[1]}} \geq \valV > 0
   \qquad
   \valV < \supply[\confG]{\tokM{\tokT[0]}{\tokT[1]}}
   \qquad
   \valV[i] = \valV \cdot \RX{i}{\confG}{\tokT[0]}{\tokT[1]}
   \qquad
   \valV[i]^{\textit{min}} \leq \valV[i]
}
{
 \begin{array}{ll}
   \confG \; = \;
   & \walA{\tokBal}
     \; \mid \;
     \amm{\ammR[0]:\tokT[0]}{\ammR[1]:\tokT[1]}
     \; \mid \;
     \confGi
     \xrightarrow{\actAmmRedeem{\pmvA}{\valV:\tokM{\tokT[0]}{\tokT[1]},\valV[0]^{\textit{min}}:\tokT[0],\valV[1]^{\textit{min}}:\tokT[1]}}
   \\[4pt]
   & \walA{\tokBal + \valV[0]:\tokT[0] + \valV[1]:\tokT[1] - \valV:\tokM{\tokT[0]}{\tokT[1]}}
     \; \mid \;
     \amm{\ammR[0]-\valV[0]:\tokT[0]}{\ammR[1]-\valV[1]:\tokT[1]}
     \; \mid \;
     \confGi
 \end{array}
}
\nrule{[Rdm]}
\]

\noindent
Amending the \nrule{[Dep]} rule is more complex, since here  
we must define ranges for the deposited amounts $\valV[0]$, $\valV[1]$,
and we must preserve the ratio between the AMM reserves.
A possible way to achieve this behaviour is the following rule:
\[
\irule{
 \begin{array}{l}
   \tokBal{\tokT[i]} \geq \valV[i] > 0
   \quad
   \valV = \frac{\valV[i]}{\RX{i}{\confG}{\tokT[0]}{\tokT[1]}}
   \quad
   (\valV[0],\valV[1]) = \begin{cases}
     (\valV[0]^{\textit{max}},\valV[0]^{\textit{max}}\cdot\frac{\ammR[1]}{\ammR[0]}) &
     \text{if $\valV[1]^{\textit{min}} \leq \valV[0]^{\textit{max}} \cdot \frac{\ammR[1]}{\ammR[0]} \leq \valV[1]^{\textit{max}}$} \\
     (\valV[1]^{\textit{max}}\cdot\frac{\ammR[0]}{\ammR[1]},\valV[1]^{\textit{max}}) &
     \text{if $\valV[0]^{\textit{min}} \leq \valV[1]^{\textit{max}} \cdot \frac{\ammR[0]}{\ammR[1]} \leq \valV[0]^{\textit{max}}$} \\
   \end{cases}
 \end{array}
}
{
 \begin{array}{ll}
   \confG \; = \;
   & \walA{\tokBal}
     \; \mid \;
     \amm{\ammR[0]:\tokT[0]}{\ammR[1]:\tokT[1]}
     \; \mid \;
     \confGi
     \xrightarrow{\actAmmDepositGuarded{\pmvA}{\valV[0]^{\textit{min}},\valV[0]^{\textit{max}}}{\tokT[0]}{\valV[1]^{\textit{min}},\valV[1]^{\textit{max}}}{\tokT[1]}}
   \\[4pt]
   & \walA{\tokBal - \valV[0]:\tokT[0] - \valV[1]:\tokT[1] + \valV:\tokM{\tokT[0]}{\tokT[1]}}
     \; \mid \;
     \amm{\ammR[0]+\valV[0]:\tokT[0]}{\ammR[1]+\valV[1]:\tokT[1]}
     \; \mid \;
     \confGi
 \end{array}
}
\nrule{[Dep]}
\]

These amendments, 
which are coherent with Uniswap implementation~\cite{uniswapimpl},
preserve all the properties, both structural and economic, 
established in the previous sections,
modulo a restatement of the properties which have transactions in their
hypotheses. 
For instance, in Theorem~\ref{thm:dep-arbitrage},
the scaling factor $a$ will be computed on the actual deposited value,
rather than on the parameter of the transaction.
Note that, although the new rules can disable some transactions
which were enabled with the rules in \autoref{sec:amm},
this does not affect the transactions reordering result
(Lemma~\ref{lem:tx-concurrent}).

\subsection{Other variants}

There are further differences between our model and the existing AMM platforms,
that could be accounted for in extensions of our model.
Uniswap implements flash-loans as part of the swap actions: 
namely, the user can optionally borrow available pair funds \cite{uniswapflash} 
whilst returning these within the same \emph{atomic group} of actions. 
Further, Uniswap implements an exchange rate oracle, allowing 
smart contracts to interpret (averages of) recent swap rates
as exchange rates \cite{uniswaporacle}. 
Balancer~\cite{balancerpaper} extends token pairs to token \emph{tuples}: 
a user can swap any two non-coinciding sets of supported tokens, 
such that the swap rate is maintained. 
In all AMM implementations, token balances are represented as integers: 
consequently, they are subject to rounding errors \cite{rvammspec}. 
AMM platforms frequently implement a governance logic, 
which allow ``governance token'' holders to coordinate changes 
to AMM fee-rates or swap rate parameters.

\section{Conclusions}
\label{sec:conclusions}

We have proposed a theory of AMMs, 
which encompasses and generalizes the main 
functional and economic aspects of the mainstream AMM implementations, 
providing solid grounds for the design of future AMMs. 

The core of our theory is a formal model of AMMs (\autoref{sec:amm}), 
based on a thorough inspection of leading AMM implementations 
like Uniswap \cite{uniswapimpl}, Curve \cite{curveimpl}, 
and Balancer \cite{balancerpaper}. 
An original aspect of our model is that 
it is parametric with respect to the key economic mechanism
--- the \emph{swap rate function} --- 
that algorithmically determines exchange rates between tokens. 
Our model features an \emph{executable semantics}, 
which can support future implementations and analysis tools;
an open-source implementation of our semantics is available as a companion of this paper.

Building upon our model, we prove a set of properties characterizing 
both structural (\autoref{sec:struct-properties}) 
and economic~(\autoref{sec:econom-defs}, \autoref{sec:arbitrage}) aspects of AMMs. 
Structural properties include, \eg,
that value cannot be created or destroyed (Lemma~\ref{lem:amm:W-preservation}), 
that tokens cannot be frozen within an AMM (Lemma~\ref{lem:amm:liquidity}) 
and that some sequences of transactions can be reordered without
affecting their semantics (Lemma~\ref{lem:tx-concurrent}).
Concerning the economic properties, we address the \emph{arbitrage problem}, 
the main game-theoretic foundation behind the economic incentives of AMMs.
Theorem~\ref{thm:arbitrage} provides sufficient conditions 
for the existence of solutions,
and links the solutions to the expected relation between 
internal exchange rate and oracle's exchange rate.
We show that deposits incentivize swaps, 
while redeems have the opposite effect. 
With respect to previous works, which focus on specific economic mechanisms, 
all our results are parametric with respect to the swap rate function. 
We identify indeed, for each property, a set of conditions 
on swap rate functions that are sufficient for the property to hold 
(\autoref{sec:swap-rate}).

AMM platforms like Uniswap~\cite{uniswapimpl} and Curve~\cite{curvepaper} 
have overtaken centralized cryptocurrency markets in size and usage. 
On the one hand, a better understanding of AMM design 
in cases where AMMs host the majority of the token's global swap volume 
is critical~\cite{angeris2020does}. 
On the other hand, the growth of AMMs is making them more attractive 
for malicious users, even if it is difficult to exactly quantify
the effect of attacks. 

This paper, together with our work on formalizing another DeFi archetype
called \emph{lending pool}~\cite{SokLP}, 
is the first step towards a general theory of DeFi~\cite{BCL21defi}.
We believe that a general theory encompassing interactions 
between different DeFi archetypes is crucial 
to be able to reason about their structural, economic and security aspects,
as typical DeFi applications operate within a wider ecosystem, 
composed by a set of collaborating or competing agents, 
which interact through possibly separate execution environments.

\paragraph*{Acknowledgements}
Massimo Bartoletti is partially supported by
Conv.\ Fondazione di Sardegna \& Atenei Sardi project
F75F21001220007 \emph{ASTRID}. 
James Hsin-yu Chiang is supported by the PhD School of DTU Compute. 
Alberto {Lluch Lafuente} is partially supported by the EU H2020-SU-ICT-03-2018 Project No. 830929 CyberSec4Europe (\href{https://www.cybersec4europe.eu}{cybersec4europe.eu}).
We thank the anonymous reviewers of COORDINATION 2021 and LMCS,
and Elvis Sikora for their insightful comments on preliminary versions
of this paper.

\bibliographystyle{alphaurl}
\bibliography{main}

\pagebreak
\appendix
%
%

\section{Proofs for Section~\ref{sec:econom-defs}}
\label{proofs:econom-defs}

\begin{proofof}{Lemma}{lem:swap:gain}
  Let $\confG$ and $\txT$ be as in the hypotheses,
  let $\confG \xrightarrow{\txT} \confGi$, and 
  let $y = x \cdot \SX{x,\ammR[0],\ammR[1]}$.
  By definition of gain (Equation~\ref{def:gain:A}), we have that:
  \[
  \gain[\confG]{\pmvA}{\txT}
  \; = \;
  W_{\pmvA}(\confGi) - W_{\pmvA}(\confG)
  \]
  By definition of net worth (Equation~\ref{eq:net-worth:user}), we have that:
  \begin{align*}
    W_{\pmvA}(\confG)
    & = \tokBal[\pmvA](\tokT[0]) \cdot \exchO{\tokT[0]}
      \; + \;
      \tokBal[\pmvA](\tokT[1]) \cdot \exchO{\tokT[1]}
    \\
    & + 
      \tokBal[\pmvA] \tokM{\tokT[0]}{\tokT[1]}
      \cdot \dfrac
      {\ammR[0] \cdot \exchO{\tokT[0]} + \ammR[1] \cdot \exchO{\tokT[1]}}
      {\supply[{\confG}]{\tokM{\tokT[0]}{\tokT[1]}}}
    \\
    & +
      \textstyle \sum_{\tokT \not\in \setenum{\tokT[0],\tokT[1], \tokM{\tokT[0]}{\tokT[1]}}}
      \tokBal[\pmvA](\tokT) \cdot \exchO[\confG]{\tokT}
    \\
    W_{\pmvA}(\confGi)
    & = (\tokBal[\pmvA](\tokT[0]) - x) \cdot \exchO{\tokT[0]}
      \; + \;
      (\tokBal[\pmvA](\tokT[1]) + y) \cdot \exchO{\tokT[1]}
    \\
    & + 
      \tokBal[\pmvA] \tokM{\tokT[0]}{\tokT[1]}
      \cdot \dfrac
      {(\ammR[0] + x) \cdot \exchO{\tokT[0]} + (\ammR[1] - y) \cdot \exchO{\tokT[1]}}
      {\supply[\confGi]{\tokM{\tokT[0]}{\tokT[1]}}}
    \\
    & +
      \textstyle \sum_{\tokT \not\in \setenum{\tokT[0],\tokT[1], \tokM{\tokT[0]}{\tokT[1]}}}
      \tokBal[\pmvA](\tokT) \cdot \exchO[\confGi]{\tokT}
  \end{align*}
  Since 
  $\supply[\confG]{\tokM{\tokT[0]}{\tokT[1]}} = \supply[\confGi]{\tokM{\tokT[0]}{\tokT[1]}}$
  and
  $\exchO[\confG]{\tokT} = \exchO[\confGi]{\tokT}$
  for all $\tokT \neq \tokM{\tokT[0]}{\tokT[1]}$:
  \begin{align*}
    W_{\pmvA}(\confGi) - W_{\pmvA}(\confG)
    & = y \cdot \exchO{\tokT[1]}
      - x \cdot \exchO{\tokT[0]}
      + \tokBal[\pmvA]\tokM{\tokT[0]}{\tokT[1]} \frac
      {x \cdot \exchO{\tokT[0]} - y \cdot \exchO{\tokT[1]}}
      {\supply[\confG]{\tokM{\tokT[0]}{\tokT[1]}}}
    \\
    & = \big( y \cdot \exchO{\tokT[1]} - x \cdot \exchO{\tokT[0]} \big)
      \Big( 1 - \frac{\tokBal[\pmvA]\tokM{\tokT[0]}{\tokT[1]}}{\supply[\confG]{\tokM{\tokT[0]}{\tokT[1]}}} \Big)
    \\
    & =     x \cdot \big(
      \SX{x,\ammR[0],\ammR[1]} \, \exchO{\tokT[1]}
      -
      \exchO{\tokT[0]}
      \big)
      \Big(
      1 - \frac{\tokBal[\pmvA]\tokM{\tokT[0]}{\tokT[1]}}{\supply[\confG]{\tokM{\tokT[0]}{\tokT[1]}}}
      \Big)
  \end{align*}
  Using similar calculations, for $\pmvB \neq \pmvA$, we obtain:
  \begin{align*}
    \gain[\confG]{\pmvB}{\txT}
    & = 
      \tokBal[\pmvB]\tokM{\tokT[0]}{\tokT[1]} \frac
      {x \cdot \exchO{\tokT[0]} - y \cdot \exchO{\tokT[1]}}
      {\supply[\confG]{\tokM{\tokT[0]}{\tokT[1]}}} \tag*{\qedhere} 
  \end{align*}
\end{proofof}

\begin{proofof}{Lemma}{lem:swap:gain-SX-X}
  Let $y = x \cdot \SX{x,\ammR[0],\ammR[1]}$.
  Since $\tokBal{\tokM{\tokT[0]}{\tokT[1]}} = 0$,
  by Lemma~\ref{lem:swap:gain} we have that:
  \begin{align*}
    \gain[{\confG}]{\pmvA}{\txT}
    \circ 0
    \iff
    & y \ \exchO{\tokT[1]}
      -
      x \ \exchO{\tokT[0]}
      \circ 0
    \\
    \iff
    & \frac{y}{x} \circ \frac{\exchO{\tokT[0]}}{\exchO{\tokT[1]}}
    \\
    \iff
    & \SX{x,\ammR[0],\ammR[1]} \circ \X{\tokT[0],\tokT[1]} \tag*{\qedhere} 
  \end{align*}
\end{proofof}

%
%

\section{Proofs for Section~\ref{sec:struct-properties}}
\label{proofs:struct-properties}

\begin{proofof}{Lemma}{lem:amm:determinism}
  Straightforward inspection of the rules 
  \nrule{[Dep0]}, \nrule{[Dep]}, \nrule{[Rdm]}, \nrule{[Swap]} in~\autoref{sec:amm}.
\end{proofof}

\begin{proofof}{Lemma}{lem:non-depletion}
  For item~\ref{lem:non-depletion:ammR},
  we proceed by induction on the length of a computation
  $\confG[0] \xrightarrow{}^* \confG$, where $\confG[0]$ is initial.
  The base case (computation of zero steps) is trivial,
  since initial states does not contain AMMs. For the inductive case,
  note that rule \nrule{[Dep0]} requires that the initial reserves
  of an AMM are strictly greater than zero.
  The rules that decrease the token reserves in AMMs, 
  \ie\ \nrule{[Rdm]} and \nrule{[Swap]}, have premises that ensure
  that the reserves cannot be zeroed.

  For item~\ref{lem:non-depletion:supply}, 
  we proceed by induction on the length of a computation
  $\confG[0] \xrightarrow{}^* \confG$, where $\confG[0]$ is initial.
  The base case is trivial, since initial states do not contain AMMs.
  For the inductive case, 
  we assume that $\confG$ satisfies the property,
  and we prove that it is preserved 
  by a transition $\confG \xrightarrow{} \confGi$.
  Assume that $\confGi$ contains an AMM
  $\amm{\ammRi[0]:\tokT[0]}{\ammRi[1]:\tokT[1]}$.
  By item~\eqref{lem:non-depletion:ammR}, 
  $\ammRi[0]>0$ and $\ammRi[1]>0$.
  There are the following cases, depending on the rule used to
  infer $\confG \xrightarrow{} \confGi$:
  \begin{itemize}

  \item \nrule{[Dep0]}, \nrule{[Dep]}.
    Trivial, because deposits can only increase the supply of minted tokens.

  \item \nrule{[Swap]}.
    Trivial, because $\ammSwapOp$ actions do not affect the supply of minted tokens.

  \item \nrule{[Rdm]}.
    Assume that $\amm{\ammR[0]:\tokT[0]}{\ammR[1]:\tokT[1]} \in \confG$.
    By contradiction, suppose that the \nrule{[Rdm]} action burns 
    all the supply of the minted token, \ie it burns
    $\valV = \supply[\confG]{\tokM{\tokT[0]}{\tokT[1]}}$ units.
    The rule premise requires $\valV > 0$, and it implies:
    \[
      \ammRi[0] 
      \; = \;
      \ammR[0] - \valV \frac{\ammR[0]}{\supply[\confG]{\tokM{\tokT[0]}{\tokT[1]}}}
      \; = \;
      0
      \qquad\qquad
      \ammRi[1] 
      \; = \;
      \ammR[1] - \valV \frac{\ammR[1]}{\supply[\confG]{\tokM{\tokT[0]}{\tokT[1]}}} 
      \; = \;
      0
    \]
    Therefore, we would have $\ammRi[0] = \ammRi[1] = 0$ --- contradiction. \qedhere 

  \end{itemize}
\end{proofof}

\begin{proofof}{Lemma}{lem:supply:const}
  By cases on the rule used in the transition 
  $\confG \xrightarrow{\txT} \confGi$.
  It is straightforward to check that, in all the rules,
  the changes applied to atomic tokens cancel out.
  Further, the \nrule{[Swap]} rule does not affect the supply
  of minted tokens.
\end{proofof}

\begin{proofof}{Lemma}{lem:dep-rdm:const}
  Let $\confG \xrightarrow{\txT} \confGi$,
  where $\amm{\ammR[0]:\tokT[0]}{\ammR[1]:\tokT[1]} \in \confG$
  and $\amm{\ammRi[0]:\tokT[0]}{\ammRi[1]:\tokTi[1]} \in \confGi$.
  If \mbox{$\txT = \actAmmDeposit{\pmvA}{\valV[0]}{\tokT[0]}{\valV[1]}{\tokT[1]}$},
  then by the \nrule{[Dep]} rule it must be
  $\ammRi[i] = \ammR[i] + \valV[i]$ for $i \in \setenum{0,1}$.
  Furthermore, by the premises of \nrule{[Dep]}, we obtain: 
  \[
  \ammR[1] \valV[0]
  = \ammR[1] \valV \cdot \RX{0}{\confG}{\tokT[0]}{\tokT[1]}
  = \valV \cdot \frac{\ammR[0] \ammR[1]}{\supply[\confG]{\tokM{\tokT[0]}{\tokT[1]}}}
  = \ammR[0] \valV \cdot \RX{1}{\confG}{\tokT[0]}{\tokT[1]}
  = \ammR[0] \valV[1]
  \]
  Therefore:
  \begin{equation}
    \label{eq:amm:dep-preserves-ratio}
    \frac{\ammR[1]+\valV[1]}{\ammR[0]+\valV[0]} 
    \; = \; 
    \frac{(\frac{\ammR[0] \valV[1]}{\valV[0]})+\valV[1]}{\ammR[0]+\valV[0]} 
    \; = \; 
    \frac{\ammR[0]\valV[1] +\valV[0]\valV[1]}{(\ammR[0]+\valV[0])\valV[0]} 
    \; = \; 
    \frac{(\ammR[0]+\valV[0])\valV[1]}{(\ammR[0]+\valV[0])\valV[0]} 
    \; = \; 
    \frac{\valV[1]}{\valV[0]}
    \; = \; 
    \frac{\ammR[1]}{\ammR[0]}
  \end{equation}
  If $\txT = \actAmmRedeem{\pmvA}{\valV:\tokM{\tokT[0]}{\tokT[1]}}$,
  then by rule~\nrule{[Rdm]} it must be, for $i \in \setenum{0,1}$:
  \[
    \ammRi[i] 
    \; = \;
    \ammR[i] - \valV[i]
    \; = \;
    \ammR[i] - \valV \RX{i}{\confG}{\tokT[0]}{\tokT[1]}
    \; = \;
    \ammR[i] - \valV \frac{\ammR[i]}{\supply[\confG]{\tokM{\tokT[0]}{\tokT[1]}}}
  \]
  Therefore, since $\supply[\confG]{\tokM{\tokT[0]}{\tokT[1]}} = \supply[\confGi]{\tokM{\tokT[0]}{\tokT[1]}} + \valV$:
  \begin{equation}
    \label{eq:amm:rdm-preserves-ratio}
    \frac{\ammR[1] - \valV[1]}{\ammR[0] - \valV[0]}
    \; = \;
    \frac
        {\ammR[1] - \valV \frac{\ammR[1]}{\supply[\confG]{\tokM{\tokT[0]}{\tokT[1]}}}}
        {\ammR[0] - \valV \frac{\ammR[0]}{\supply[\confG]{\tokM{\tokT[0]}{\tokT[1]}}}}
    \; = \;
    \frac
        {\ammR[1] (\supply[\confGi]{\tokM{\tokT[0]}{\tokT[1]}} + \valV) - \valV\ammR[1]}
        {\ammR[0] (\supply[\confGi]{\tokM{\tokT[0]}{\tokT[1]}} + \valV) - \valV\ammR[0]}
     \; = \;
     \frac{\ammR[1]}{\ammR[0]}
  \end{equation}
  Summing up, \eqref{eq:amm:dep-preserves-ratio} and \eqref{eq:amm:rdm-preserves-ratio}
  give item~\ref{lem:dep-rdm:const:ratio}.
  
  \medskip\noindent
  For item~\ref{lem:dep-rdm:const:rx},
  if $\txT =$ $\actAmmDeposit{\pmvA}{\valV[0]}{\tokT[0]}{\valV[1]}{\tokT[1]}$,
  then by the \nrule{[Dep]} rule it must be
  $\ammRi[i] = \ammR[i] + \valV[i]$ for $i \in \setenum{0,1}$,
  and $\supply[\confGi]{\tokM{\tokT[0]}{\tokT[1]}} = \supply[\confG]{\tokM{\tokT[0]}{\tokT[1]}} + \frac{\valV[i]}{\ammR[i]} \supply[\confG]{\tokM{\tokT[0]}{\tokT[1]}}$.
  Therefore:
  \[
    \RX{i}{\confGi}{\tokT[0]}{\tokT[1]}
    =
    \frac{\ammR[i] + \valV[i]}{\supply[\confGi]{\tokM{\tokT[0]}{\tokT[1]}}}
    =
    \frac{\ammR[i] + \valV[i]}{\supply[\confG]{\tokM{\tokT[0]}{\tokT[1]}} ( 1 + \frac{\valV[i]}{\ammR[i]} )}
    =
    \frac{(\ammR[i] + \valV[i]) \ammR[i]}{\supply[\confG]{\tokM{\tokT[0]}{\tokT[1]}} ( \ammR[i] + \valV[i] )}
    =
    \RX{i}{\confG}{\tokT[0]}{\tokT[1]}
  \]
  Otherwise, if $\txT = \actAmmRedeem{\pmvA}{\valV:\tokM{\tokT[0]}{\tokT[1]}}$,
  then by rule~\nrule{[Rdm]} it must be, for $i \in \setenum{0,1}$:
  \[
    \ammRi[i] 
    \; = \;
    \ammR[i] - \valV[i]
    \; = \;
    \ammR[i] - \valV \RX{i}{\confG}{\tokT[0]}{\tokT[1]}
    \; = \;
    \ammR[i] - \valV \frac{\ammR[i]}{\supply[\confG]{\tokM{\tokT[0]}{\tokT[1]}}}
  \]
  Therefore:
  \begin{align*}
    \RX{i}{\confGi}{\tokT[0]}{\tokT[1]}
    & =
    \frac{\ammR[i] - \valV[i]}{\supply[\confG]{\tokM{\tokT[0]}{\tokT[1]}} - \valV}
    =
    \frac
    {\ammR[i] - \valV \frac{\ammR[i]}{\supply[\confG]{\tokM{\tokT[0]}{\tokT[1]}}}}
    {\supply[\confG]{\tokM{\tokT[0]}{\tokT[1]}} - \valV}
    =
    \frac
    {\ammR[i] \supply[\confG]{\tokM{\tokT[0]}{\tokT[1]}}- \valV \ammR[i]}
    {(\supply[\confG]{\tokM{\tokT[0]}{\tokT[1]}} - \valV) \supply[\confG]{\tokM{\tokT[0]}{\tokT[1]}}}
    \\
    & = \frac
      {\ammR[i]}
      {\supply[\confG]{\tokM{\tokT[0]}{\tokT[1]}}}
      = \RX{i}{\confG}{\tokT[0]}{\tokT[1]} \tag*{\qedhere} 
  \end{align*}

  \medskip\noindent
  For item~\ref{lem:dep-rdm:const:price},
  if \mbox{$\txT = \actAmmDeposit{\pmvA}{\valV[0]}{\tokT[0]}{\valV[1]}{\tokT[1]}$},
  we have that:
  \begin{align*}
    \exchO[\confGi]{\tokM{\tokT[0]}{\tokT[1]}}
    & =
    \dfrac
        {\ammRi[0] \cdot \exchO{\tokT[0]} + \ammRi[1] \cdot \exchO{\tokT[1]}}
        {\supply[\confGi]{\tokM{\tokT[0]}{\tokT[1]}}}
    && \text{by~\Cref{def:price}}
    \\
    & =
    \dfrac
        {(1 + \frac{\valV[0]}{\ammR[0]}) \cdot \ammR[0] \cdot \exchO{\tokT[0]} +
         (1 + \frac{\valV[1]}{\ammR[1]}) \cdot \ammR[1] \cdot \exchO{\tokT[1]}}
        {\supply[\confG]{\tokM{\tokT[0]}{\tokT[1]}} + \frac{\valV[i]}{\ammR[i]} \supply[\confG]{\tokM{\tokT[0]}{\tokT[1]}}}
    \\
    & =
    \dfrac
        {(1 + \frac{\valV[i]}{\ammR[i]}) \cdot \ammR[0] \cdot \exchO{\tokT[0]} +
         (1 + \frac{\valV[i]}{\ammR[i]}) \cdot \ammR[1] \cdot \exchO{\tokT[1]}}
        {\big( 1 + \frac{\valV[i]}{\ammR[i]} \big) \cdot \supply[\confG]{\tokM{\tokT[0]}{\tokT[1]}}}
        && \text{since $\tfrac{\valV[0]}{\ammR[0]} = \tfrac{\valV[1]}{\ammR[1]}$}
    \\
    & = \exchO[\confG]{\tokM{\tokT[0]}{\tokT[1]}}
    && \text{by~\Cref{def:price}}    
  \end{align*}
  The proof for the case
  \mbox{$\txT = \actAmmRedeem{\pmvA}{\valV:\tokM{\tokT[0]}{\tokT[1]}}$}
  is similar.
\end{proofof}

\begin{proofof}{Lemma}{lem:amm:W-preservation}
  Let $\confG \xrightarrow{\txT} \confGi$.
  We first prove~\cref{lem:amm:W-preservation:b}.
  Depending on the rule used to fire the transition,
  we have the following cases:
  \begin{itemize}


  \item \nrule{[Dep0]}.
    Let \mbox{$\txT = \actAmmDeposit{\pmvB}{\valV[0]}{\tokT[0]}{\valV[1]}{\tokT[1]}$}.
    We have that:
    \begin{align*}
      \confG
      & = \walB{\tokBal} \mid \confG[0]
      \\
      \confGi
      & = \walB{\tokBal - \valV[0]:\tokT[0] - \valV[1]:\tokT[1] + \valV[0]:\tokM{\tokT[0]}{\tokT[1]}} \mid \amm{\valV[0]:\tokT[0]}{\valV[1]:\tokT[1]} \mid \confG[0]
    \end{align*}
    If $\pmvB \neq \pmvA$, then $\pmvA$'s net worth is unaffected.
    Otherwise, if $\pmvB = \pmvA$, then:
    \begin{align*}
      W_{\pmvA}(\confGi)
      & = W_{\pmvA}(\confG)
        - \valV[0] \exchO{\tokT[0]}
        - \valV[1] \exchO{\tokT[1]}
        + \valV[0] \exchO[\confG]{\tokM{\tokT[0]}{\tokT[1]}}
      \\
      & = W_{\pmvA}(\confG)
        - \valV[0] \exchO{\tokT[0]}
        - \valV[1] \exchO{\tokT[1]}
        + \valV[0] \frac{\valV[0] \exchO{\tokT[0]} + \valV[1] \exchO{\tokT[1]}}{\valV[0]}
        && \text{by \Cref{def:price}}        
      \\
      & = W_{\pmvA}(\confG)
    \end{align*}

  \item \nrule{[Dep]}.
    Let $\txT = \actAmmDeposit{\pmvB}{\valV[0]}{\tokT[0]}{\valV[1]}{\tokT[1]}$.
    We have that:
    \begin{align*}
      \confG
      & = \walB{\tokBal} \mid \amm{\ammR[0]:\tokT[0]}{\ammR[1]:\tokT[1]} \mid \confG[0]
      \\
      \confGi
      & = \walB{\tokBal - \valV[0]:\tokT[0] - \valV[1]:\tokT[1] + \valV:\tokM{\tokT[0]}{\tokT[1]}} \mid
        \amm{\ammR[0]+\valV[0]:\tokT[0]}{\ammR[1]+\valV[1]:\tokT[1]} \mid \confG[0]
    \end{align*}
    where:
    \[
      \valV = 
      \frac{\valV[0] \cdot \supply[\confG]{\tokM{\tokT[0]}{\tokT[1]}}}{\ammR[0]}
    \]
    
    \noindent
    If $\pmvB \neq \pmvA$, then $\pmvA$'s net worth is unaffected
    (note that the value of minted tokens in $\pmvA$'s wallet is
    preserved by deposits, by Lemma~\ref{lem:dep-rdm:const}\ref{lem:dep-rdm:const:price}).    
    Otherwise, if $\pmvB = \pmvA$, then:
    \begin{align*}
      W_{\pmvA}(\confGi)
      & = W_{\pmvA}(\confG)
        - \valV[0] \exchO{\tokT[0]}
        - \valV[1] \exchO{\tokT[1]}
        + \valV \exchO[\confG]{\tokM{\tokT[0]}{\tokT[1]}}
      \\
      & = W_{\pmvA}(\confG)
        - \valV[0] \exchO{\tokT[0]}
        - \valV[1] \exchO{\tokT[1]}
        + \valV \frac{\ammR[0] \exchO{\tokT[0]} + \ammR[1] \exchO{\tokT[1]}}{\supply[\confG]{\tokM{\tokT[0]}{\tokT[1]}}}
        && \text{by \Cref{def:price}}
      \\
      & = W_{\pmvA}(\confG)
        - \valV[0] \exchO{\tokT[0]}
        - \valV[1] \exchO{\tokT[1]}
        + \frac{\valV[0]}{\ammR[0]} \Big( \ammR[0] \exchO{\tokT[0]} + \ammR[1] \exchO{\tokT[1]} \Big)
      \\
      & = W_{\pmvA}(\confG)
        - \valV[1] \exchO{\tokT[1]}
        + \frac{\valV[0]}{\ammR[0]} \ammR[1] \exchO{\tokT[1]}
      \\
      & = W_{\pmvA}(\confG)
      && \text{since $\ammR[1]\valV[0] = \ammR[0]\valV[1]$}
    \end{align*}

  \item \nrule{[Swap]}.
    This case cannot happen, since we are assuming 
    $\txType{\txT} \neq \ammSwapOp$.

  \item \nrule{[Rdm]}.
    Let $\txT = \actAmmRedeem{\pmvB}{\valV:\tokM{\tokT[0]}{\tokT[1]}}$.
    We have that:
    \begin{align*}
      \confG
      & = \walB{\tokBal} \mid \amm{\ammR[0]:\tokT[0]}{\ammR[1]:\tokT[1]} \mid \confG[0]
      \\
      \confGi
      & = \walB{\tokBal + \valV[0]:\tokT[0] + \valV[1]:\tokT[1] - \valV:\tokM{\tokT[0]}{\tokT[1]}} \mid
        \amm{\ammR[0]-\valV[0]:\tokT[0]}{\ammR[1]-\valV[1]:\tokT[1]} \mid \confG[0]
    \end{align*}
    where:
    \[
      \valV[0] = \frac{\valV \cdot \ammR[0]}{s}
      \qquad
      \valV[1] = \frac{\valV \cdot \ammR[1]}{s}
      \qquad
      s = \supply[\confG]{\tokM{\tokT[0]}{\tokT[1]}}
    \]
    If $\pmvB \neq \pmvA$, then $\pmvA$'s net worth is unaffected
    (note that the value of minted tokens in $\pmvA$'s wallet is
    preserved by redeems, by Lemma~\ref{lem:dep-rdm:const}\ref{lem:dep-rdm:const:price}).    
    Otherwise, if $\pmvB = \pmvA$, then:
    \begin{align*}
      W_{\pmvA}(\confGi)
      & = W_{\pmvA}(\confG)
        + \valV[0] \exchO{\tokT[0]}
        + \valV[1] \exchO{\tokT[1]}
        - \valV \exchO[\confG]{\tokM{\tokT[0]}{\tokT[1]}}
      \\
      & = W_{\pmvA}(\confG)
        + \valV[0] \exchO{\tokT[0]}
        + \valV[1] \exchO{\tokT[1]}
        - \frac{\valV \cdot \ammR[0]}{s} \exchO{\tokT[0]}
        - \frac{\valV \cdot \ammR[1]}{s} \exchO{\tokT[1]}
      \\
      & = W_{\pmvA}(\confG)
        + \frac{\valV \cdot \ammR[0]}{s} \exchO{\tokT[0]}
        + \frac{\valV \cdot \ammR[1]}{s} \exchO{\tokT[1]}
        - \frac{\valV \cdot \ammR[0]}{s} \exchO{\tokT[0]}
        - \frac{\valV \cdot \ammR[1]}{s} \exchO{\tokT[1]}
      \\
      & = W_{\pmvA}(\confG)
    \end{align*}
  \end{itemize}

  \noindent
  We now prove~\cref{lem:amm:W-preservation:a},
  \ie that the \emph{global} net worth is preserved by \emph{any} transactions.
  First, we recall from~\autoref{sec:econom-defs} 
  the definition of global net worth.
  Let:
  \[
    \confG \; = \; \wal{\pmvA[1]}{\tokBal[1]} \mid \cdots \mid \wal{\pmvA[n]}{\tokBal[n]}
    \mid
    \amm{\ammR[1]:\tokT[1]}{\ammRi[1]:\tokTi[1]}
    \mid \cdots \mid
    \amm{\ammR[k]:\tokT[k]}{\ammRi[k]:\tokTi[k]}
  \]
  Then, the global net worth of $\confG$ is:
  \begin{equation*}
    W(\confG)
    \; = \;
    \sum_{i = 1}^n W_{{\pmvA[i]}}(\confG)
  \end{equation*}
  We have the following cases:
  \begin{itemize}


  \item \nrule{[Dep0]}, \nrule{[Dep]}, \nrule{[Rdm]}.
    These rules affect the token reserves in AMMs,
    which do not contribute to the global net worth,
    and the balances of users, which we know to be preserved.
    Therefore, the global net worth is preserved.

  \item \nrule{[Swap]}.
    Let $\actAmmSwapExact{\pmvA}{}{\valV}{\tokT[0]}{\tokT[1]}$
    be the fired transaction.
    We have that:
    \begin{align*}
      \confG
      & =
        \walA{\tokBal}
        \; \mid \;
        \amm{\ammR[0]:\tokT[0]}{\ammR[1]:\tokT[1]}
        \; \mid \;
        \confG[0]
      \\
      \confGi
      & =
        \walA{\tokBal - \valV:\tokT[0] + \valVi:\tokT[1]}
        \; \mid \;
        \amm{\ammR[0]+\valV:\tokT[0]}{\ammR[1]-\valVi:\tokT[1]}
        \; \mid \;
        \confG[0]
    \end{align*}
    The global net worth in $\confGi$ can be computed in terms of
    the global net worth in $\confG$,
    by removing the value of the $\valV:\tokT[0]$ paid by $\pmvA$ to the AMM,
    adding the value of the $\valV[1]:\tokT[1]$ obtained by $\pmvA$
    through the swap,
    and then adding the difference between the value of the minted tokens
    in $\confGi$ and in $\confG$, \ie:
    \[
    \supply[\confGi]{\tokM{\tokT[0]}{\tokT[1]}} 
    \exchO[\confGi]{\tokM{\tokT[0]}{\tokT[1]}} -
    \supply[\confG]{\tokM{\tokT[0]}{\tokT[1]}}     
    \exchO[\confG]{\tokM{\tokT[0]}{\tokT[1]}}    
    \]
    By Lemma~\ref{lem:supply:const}, we have that
    $\supply[\confGi]{\tokM{\tokT[0]}{\tokT[1]}} = \supply[\confG]{\tokM{\tokT[0]}{\tokT[1]}}$.
    Therefore:
    \begin{align*}
      W(\confGi)
      & = W(\confG)
        - \valV \exchO{\tokT[0]}
        + \valVi \exchO{\tokT[1]}
        + \supply[\confG]{\tokM{\tokT[0]}{\tokT[1]}} 
        \big(
        \exchO[\confGi]{\tokM{\tokT[0]}{\tokT[1]}} - \exchO[\confG]{\tokM{\tokT[0]}{\tokT[1]}}
        \big)
      \\
      & = W(\confG)
        - \valV \exchO{\tokT[0]}
        + \valVi \exchO{\tokT[1]}
      \\
      & \quad
        + \supply[\confG]{\tokM{\tokT[0]}{\tokT[1]}} \cdot
        \Big( 
        \dfrac{
        \ammR[0] \exchO{\tokT[0]}
        + \ammR[1] \exchO{\tokT[1]}
        + \valV \exchO{\tokT[0]}
        - \valVi \exchO{\tokT[1]}
        }
        {\supply[\confG]{\tokM{\tokT[0]}{\tokT[1]}}}
        -
        \dfrac{
        \ammR[0] \exchO{\tokT[0]}
        + \ammR[1] \exchO{\tokT[1]}
        }
        {\supply[\confG]{\tokM{\tokT[0]}{\tokT[1]}}}
        \Big)
      \\
      & = W(\confG) \tag*{\qedhere} 
    \end{align*}
  \end{itemize}
\end{proofof}

\begin{proofof}{Lemma}{lem:amm:W-TokUInit}
  Direct consequence of Lemma~\ref{lem:amm:W-preservation}\ref{lem:amm:W-preservation:b}
  and of the hypothesis that $\pmvA$ does not hold minted tokens in $\confGi$.
\end{proofof}

\begin{proofof}{Lemma}{lem:amm:liquidity}
  Let $\confG^0 = \amm{\ammR[0]:\tokT[0]}{\ammR[1]:\tokT[1]} \mid \confD^0$
  be a reachable state.
  We define below a procedure to construct a sequence of transitions:
  \[
    \confG^0 \xrightarrow{\txT[1]} \cdots \xrightarrow{\txT[n]} \confG^n
    \qquad
    \text{where } \;
    \confG^n = \amm{\ammR[0]^i:\tokT[0]}{\ammR[1]^i:\tokT[1]} \mid \confD^n
  \]
  By Lemma~\ref{lem:non-depletion}, we have that 
  $\ammR[0]^i>0$, $\ammR[1]^i>0$, 
  and $\supply[\confG^i]{\tokM{\tokT[0]}{\tokT[1]}} > 0$ for all $i$.
  At step $i$:
  \begin{enumerate}
  \item Let $x = \ammR[0]^i - \ammRi[0]$
    be the amount of $\tokT[0]$ that users must redeem from the AMM,
    and let:
    \[
      \valV = \frac{x}{\ammR[0]^i} \supply[\confG^i]{\tokM{\tokT[0]}{\tokT[1]}}
    \]

  \item if there exists some $\walA{\tokBal} \in \confG^i$ such that
    $\tokBal(\tokM{\tokT[0]}{\tokT[1]}) \geq \valV$, then
    $\pmvA$ can fire $\actAmmRedeem{\pmvA}{\valV:\tokM{\tokT[0]}{\tokT[1]}}$,
    obtaining, for some $\ammRi[1] \leq \ammR[1]$:
    \begin{align*}
      \amm{\ammR[0]^i:\tokT[0]}{\ammR[1]^i:\tokT[1]} \mid \confD^i
      \xrightarrow{} \confGi
      & = \bigamm{\ammR[0]^i - \valV \frac{\ammR[0]^i}{\supply[\confG^i]{\tokM{\tokT[0]}{\tokT[1]}}}}{\ammRi[1]:\tokT[1]} \mid \cdots
      \\
      & = \amm{\ammRi[0]:\tokT[0]}{\ammRi[1]:\tokT[1]} \mid \cdots
    \end{align*}

    \item otherwise, pick an $\walA{\tokBal} \in \confG^i$ such that
    $\tokBal(\tokM{\tokT[0]}{\tokT[1]}) = \valVi \geq 0$, fire
    $\actAmmRedeem{\pmvA}{\valVi:\tokM{\tokT[0]}{\tokT[1]}}$.
  \end{enumerate}
  Note that the procedure always terminates: 
  since $\supply[\confG^i]{\tokM{\tokT[0]}{\tokT[1]}} > 0$ for all $i$,
  either step $(2)$ or $(3)$ can be performed;
  further, the number of performed transactions is bounded 
  by the number of users, which is finite.
\end{proofof}

\begin{proofof}{Lemma}{lem:tx-concurrent}
  Assume that
  $\confG \xrightarrow{\txT[0]} \confG[0] \xrightarrow{\txT[1]} \confG[01]$.
  We have the following exhaustive cases
  on the type of the transactions $\txT[0]$ and $\txT[1]$:
  \begin{enumerate}

  \item \label{item:tx-concurrent:dep}
    $\txT[0] = \actAmmDeposit{\pmvA[0]}{\valV[0]}{\tokT[0]}{\valVi[0]}{\tokTi[0]}$.
    \begin{enumerate}

    \item \label{item:tx-concurrent:dep-dep}
      \mbox{$\txT[1] = \actAmmDeposit{\pmvA[1]}{\valV[1]}{\tokT[1]}{\valVi[1]}{\tokT[1]}$}.
      Both transactions are $\ammDepositOp$,
      so we are in case~\ref{lem:tx-concurrent:1} of the statement.
      If $\setenum{\tokT[0],\tokTi[0]} \neq \setenum{\tokT[1],\tokTi[1]}$,
      then the thesis is straightforward, since $\txT[0],\txT[1]$
      operate on different AMMs.
      Otherwise,
      let:
      \[
      \begin{array}{llll}
        a_0 = 1 + \tfrac{\valV[0]}{\ammR[0]}
        &
        m_0 = \tfrac{\valV[0]}{\ammR[0]} \supply[\confG]{\tokM{\tokT[0]}{\tokT[1]}}
        &
        a_{01} = 1 + \tfrac{\valV[1]}{a_0\ammR[0]}
        &
        m_{01} = \tfrac{\valV[1]}{a_0\ammR[0]} \supply[{\confG[0]}]{\tokM{\tokT[0]}{\tokT[1]}}
        \\[10pt]        
        a_1 = 1 + \tfrac{\valV[1]}{\ammR[0]}
        &
        m_1 = \tfrac{\valV[1]}{\ammR[0]} \supply[\confG]{\tokM{\tokT[0]}{\tokT[1]}}
        &
        a_{10} = 1 + \tfrac{\valV[0]}{a_1\ammR[0]}
        &
        m_{10} = \tfrac{\valV[0]}{a_1\ammR[0]} \supply[{\confG[1]}]{\tokM{\tokT[0]}{\tokT[1]}}        
      \end{array}
      \]
      We have that:
      \begin{align*}
        & \walA{\tokBal} \mid
        \amm{\ammR[0]:\tokT[0]}{\ammR[1]:\tokT[1]} \mid \confD
        \\
        \xrightarrow{\txT[0]} \;
        & \walA{\tokBal - \valV[0]:\tokT[0] - \valVi[0]:\tokT[1] + m_0:\tokM{\tokT[0]}{\tokT[1]}} \mid      
        \amm{a_0 \ammR[0]:\tokT[0]}{a_0 \ammR[1]:\tokT[1]} \mid \confD
        \\
        \xrightarrow{\txT[1]} \;
        & \walA{\tokBal - (\valV[0]+\valV[1]):\tokT[0] - (\valVi[0]+\valVi[1]):\tokT[1] + (m_0 + m_{01}):\tokM{\tokT[0]}{\tokT[1]}} \mid
        \\
        & 
        \amm{a_{01} a_0 \ammR[0]:\tokT[0]}{a_{01} a_0 \ammR[1]:\tokT[1]} \mid \confD
      \end{align*}

      \noindent
      Inverting the two transactions, we obtain:
      \begin{align*}
        & \walA{\tokBal} \mid
        \amm{\ammR[0]:\tokT[0]}{\ammR[1]:\tokT[1]} \mid \confD
        \\
        \xrightarrow{\txT[1]} \;
        & \walA{\tokBal - \valV[1]:\tokT[0] - \valVi[1]:\tokT[1] + m_1:\tokM{\tokT[0]}{\tokT[1]}} \mid      
        \amm{a_1 \ammR[0]:\tokT[0]}{a_1 \ammR[1]:\tokT[1]} \mid \confD
        \\
        \xrightarrow{\txT[0]} \;
        & \walA{\tokBal[1] - (\valV[0]+\valV[1]):\tokT[0] - (\valVi[0]+\valVi[1]):\tokT[1] + (m_1 + m_{10}):\tokM{\tokT[0]}{\tokT[1]}} \mid
        \\
        &
        \amm{a_{10} a_1 \ammR[0]:\tokT[0]}{a_{10} a_1 \ammR[1]:\tokT[1]} \mid \confD
      \end{align*}

      \noindent
      We have that $a_{01} a_0 = a_{10} a_1$, since:
      \begin{align*}
        a_{10} a_1
        & =
        \big (1 + \tfrac{\valV[0]}{a_1\ammR[0]} \big)
        a_1
        =
        \frac{a_1 \ammR[0] + \valV[0]}{\ammR[0]}
        =
        \frac{\big( 1 + \tfrac{\valV[1]}{\ammR[0]} \big) \ammR[0] + \valV[0]}{\ammR[0]}
        =
        \frac{\ammR[0] + \valV[0] + \valV[1]}{\ammR[0]}
        \\
        a_{01} a_0
        & = 
        \big (1 + \tfrac{\valV[1]}{a_0\ammR[0]} \big)
        a_0
        =
        \frac{a_0 \ammR[0] + \valV[1]}{\ammR[0]}
        =
        \frac{\big( 1 + \tfrac{\valV[0]}{\ammR[0]} \big) \ammR[0] + \valV[1]}{\ammR[0]}
        =
        \frac{\ammR[0] + \valV[0] + \valV[1]}{\ammR[0]}        
      \end{align*}
      Furthermore, we have that $m_0 + m_{01} = m_1 + m_{10}$, since:
      \begin{align*}
        m_{10} + m_1
        & =
        \frac{\valV[0]\valV[1] + a_1 \ammR[0]\valV[1] + \ammR[0]\valV[0]}{a_1 \ammR[0]^2}
        \supply[{\confG}]{\tokM{\tokT[0]}{\tokT[1]}}
        \\
        & =
        \frac{\valV[0]\valV[1] + \valV[1](\ammR[0]+\valV[1]) + \ammR[0]\valV[0]}{(\ammR[0] + \valV[1]) \ammR[0]}
        \supply[{\confG}]{\tokM{\tokT[0]}{\tokT[1]}}        
        \\
        & =
        \frac{(\valV[0]+\valV[1]) (\ammR[0]+\valV[1])}{(\ammR[0] + \valV[1]) \ammR[0]}
        \supply[{\confG}]{\tokM{\tokT[0]}{\tokT[1]}}
        \; = \;
        \frac{\valV[0]+\valV[1]}{\ammR[0]}
        \supply[{\confG}]{\tokM{\tokT[0]}{\tokT[1]}}
        \\
        m_{01} + m_0
        & =
        \frac{\valV[0]\valV[1] + a_0 \ammR[0]\valV[0] + \ammR[0]\valV[1]}{a_0 \ammR[0]^2}
        \supply[{\confG}]{\tokM{\tokT[0]}{\tokT[1]}}
        \\
        & =
        \frac{\valV[0]\valV[1] + \valV[0](\ammR[0]+\valV[0]) + \ammR[0]\valV[1]}{(\ammR[0] + \valV[0]) \ammR[0]}
        \supply[{\confG}]{\tokM{\tokT[0]}{\tokT[1]}}        
        \\
        & =
        \frac{(\valV[0]+\valV[1]) (\ammR[0]+\valV[0])}{(\ammR[0] + \valV[0]) \ammR[0]}
        \supply[{\confG}]{\tokM{\tokT[0]}{\tokT[1]}}
        \; = \;
        \frac{\valV[0]+\valV[1]}{\ammR[0]}
        \supply[{\confG}]{\tokM{\tokT[0]}{\tokT[1]}}        
      \end{align*}

      Summing up, we have shown that $\confG[01] = \confG[10]$.
      
    \item \label{item:tx-concurrent:dep-swap}
      \mbox{$\txT[1] = \actAmmSwapExact{\pmvA[1]}{}{\valV[1]}{\tokT[1]}{\tokTi[1]}$}.      
      Then, we are in case~\ref{lem:tx-concurrent:1} of the statement,
      with $\txTok{\txT[0]}$ disjoint from $\txTok{\txT[1]}$.
      The thesis is straightforward by analysis of the rules.

    \item \label{item:tx-concurrent:dep-rdm}
      \mbox{$\txT[1] = \actAmmRedeem{\pmvA[1]}{\valV[1]:\tokM{\tokT[1]}{\tokTi[1]}}$}.
      There are two subcases.
      If we are in case~\ref{lem:tx-concurrent:1},
      then $\txT[0],\txT[1]$ operate on different AMMs,
      and so the thesis is straightforward.
      Otherwise, if we are in case~\ref{lem:tx-concurrent:2} of the statement,
      by hypothesis we know that $\txT[1] \txT[0]$ is enabled in $\confG$,
      leading to a state $\confG[10]$.
      If $\setenum{\tokT[0],\tokTi[0]} \neq \setenum{\tokT[1],\tokTi[1]}$,
      then the thesis is straightforward.
      Otherwise, the proof is done by computing the states
      $\confG[01]$ and $\confG[10]$ and showing they are equal,
      similarly to what we have done in case \eqref{item:tx-concurrent:dep-dep}.
      
    \end{enumerate}

  \item \label{item:tx-concurrent:rdm}
    $\txT[0] = \actAmmRedeem{\pmvA[0]}{\valV[0]:\tokM{\tokT[0]}{\tokTi[0]}}$.
    \begin{enumerate}

    \item \label{item:tx-concurrent:rdm-dep}
      $\txT[1] = \actAmmDeposit{\pmvA[1]}{\valV[1]}{\tokT[1]}{\valVi[1]}{\tokTi[1]}$.
      Symmetric to case~\eqref{item:tx-concurrent:dep-rdm}.

    \item \label{item:tx-concurrent:rdm-swap}
      $\txT[1] = \actAmmSwapExact{\pmvA[1]}{}{\valV[1]}{\tokT[1]}{\tokTi[1]}$.
      Then, we are in case~\ref{lem:tx-concurrent:1} of the statement,
      where $\setenum{\tokT[1],\tokTi[1]}$ and $\setenum{\tokT[0],\tokTi[0]}$
      are disjoint.
      Then, the thesis is straightforward.

    \item \label{item:tx-concurrent:rdm-rdm}
      $\txT[1] = \actAmmRedeem{\pmvA[1]}{\valV[1]:\tokM{\tokT[1]}{\tokTi[1]}}$.
      Then, we are in case~\ref{lem:tx-concurrent:1} of the statement.
      If $\txTok{\txT[0]}$ is disjoint from $\txTok{\txT[1]}$,
      then the thesis is straightforward.
      Otherwise, note that the tokens paid by the AMM in response of $\txT[0]$ and $\txT[1]$
      only depend on the ratio between the amounts
      of $\tokT[0]$ and $\tokTi[0]$ initially held by the AMM,
      which are constrained to preserve the ratio.

    \end{enumerate}

  \item \label{item:tx-concurrent:swap}
    $\txT[0] = \actAmmSwapExact{\pmvA[0]}{}{\valV[0]}{\tokT[0]}{\tokTi[0]}$.
    The only case not covered by the previous items is when
    $\txT[1] = \actAmmSwapExact{\pmvA[1]}{}{\valV[1]}{\tokT[1]}{\tokTi[1]}$.
    Then, we are in case~\ref{lem:tx-concurrent:1} of the statement,    
    where $\setenum{\tokT[1],\tokTi[1]}$ and $\setenum{\tokT[0],\tokTi[0]}$
    are disjoint.
    The thesis is straightforward. \qedhere 

  \end{enumerate}
\end{proofof}

\begin{proofof}{Theorem}{thm:additivity}
  For item~\ref{thm:additivity:dep}, 
  there are two cases, depending on whether $\txT[0]$ is fired through
  rule \nrule{[Dep0]} or \nrule{[Dep]}.
  If $\txT[0]$ is fired through rule~\nrule{[Dep]}, 
  let $\confG = \amm{\ammR[0]:\tokT[0]}{\ammR[1]:\tokT[1]} \mid \confD$.
  We have that:
  \begin{align}
    \label{eq:amm:additivity:dep:0}
    \confG[0] & = \amm{\ammR[0]+\valV[0]:\tokT[0]}{\ammR[1]+\valV[1]:\tokT[1]} \mid \confD[0]
    && \ammR[1] \valV[0] = \ammR[0] \valV[1]
    \\
    \label{eq:amm:additivity:dep:1}
    \confG[1] & = \amm{(\ammR[0]+\valV[0])+\valVi[0]:\tokT[0]}{(\ammR[1]+\valV[1])+\valVi[1]:\tokT[1]} \mid \confD[1]
    && (\ammR[1]+\valV[1]) \valVi[0] = (\ammR[0]+\valV[0])\valVi[1]
  \end{align}
  We must just check that the premises for firing
  $\actAmmDeposit{\pmvA}{\valV[0]+\valVi[0]}{\tokT[0]}{\valVi[1]+\valVi[1]}{\tokT[1]}$
  are satisfied:
  \begin{align*}
    \ammR[1] (\valV[0] + \valVi[0])
    & = \ammR[1]\valV[0] + \ammR[1]\valVi[0]
    \\
    & = \ammR[0]\valV[1] + \ammR[1]\valVi[0]
    && \text{by~\eqref{eq:amm:additivity:dep:0}}
    \\
    & = \ammR[0]\valV[1] + \ammR[1]\Big(\frac{\ammR[0]+\valV[0]}{\ammR[1]+\valV[1]}\Big)\valVi[1]
    && \text{by~\eqref{eq:amm:additivity:dep:1}}
    \\
    & = \ammR[0]\valV[1] + \ammR[1]\frac{\ammR[0]}{\ammR[1]}\valVi[1]
    && \text{by~\eqref{eq:amm:dep-preserves-ratio}}
    \\
    & =  \ammR[0] (\valV[1] + \valVi[1])
  \end{align*}
  The case where $\txT[0]$ is fired through rule~\nrule{[Dep0]} is similar:
  \begin{align*}
    \confG[0] & = \amm{\valV[0]:\tokT[0]}{\valV[1]:\tokT[1]} \mid \confD[0]
    && 
    \\
    \confG[1] & = \amm{\valV[0]+\valVi[0]:\tokT[0]}{\valV[1]+\valVi[1]:\tokT[1]} \mid \confD[1]
    && \valV[1] \valVi[0] = \valV[0]\valVi[1]
  \end{align*}
  The premises of \nrule{[Dep0]} when firing
  $\actAmmDeposit{\pmvA}{\valV[0]+\valVi[0]}{\tokT[0]}{\valVi[1]+\valVi[1]}{\tokT[1]}$
  are trivially satisfied, hence the thesis follows.

  \medskip
  For item~\ref{thm:additivity:rdm}, 
  let $\confG = \amm{\ammR[0]:\tokT[0]}{\ammR[1]:\tokT[1]} \mid \confD$
  let $\tokT = \tokM{\tokT[0]}{\tokT[1]}$, and
  let $s = \supply[\confG]{\tokT}$.
  By rule~\nrule{[Rdm]}, we have that:
  \begin{align}
    \label{eq:amm:additivity:rdm:0}
    \confG[0] & = \amm{\ammR[0]-\valV[0]:\tokT[0]}{\ammR[1]-\valV[1]:\tokT[1]} \mid \confD[0]
    && \valV[i] = \valV \cdot \frac{\ammR[i]}{s}
    \\
    \label{eq:amm:additivity:rdm:1}
    \confG[1] & = \amm{(\ammR[0]-\valV[0])-\valVi[0]:\tokT[0]}{(\ammR[1]-\valV[1])-\valVi[1]:\tokT[1]} \mid \confD[1]
    && \valVi[i] = \valVi \cdot \frac{\ammR[i]-\valV[i]}{s-\valV}
  \end{align}
  Therefore, for $i \in \setenum{0,1}$, we have that:
  \begin{align*}
    \ammR[i] - \valV[i] - \valVi[i]
    & = \ammR[i] - \valV \cdot \frac{\ammR[i]}{s} - \valVi \cdot \frac{\ammR[i]-\valV \cdot \frac{\ammR[i]}{s}}{s-\valV}
    && \text{by~\eqref{eq:amm:additivity:rdm:0}, \eqref{eq:amm:additivity:rdm:1}}
    \\
    & = \ammR[i] - \valV \cdot \frac{\ammR[i] (s-\valV)}{s (s-\valV)} - \valVi \cdot \frac{s \ammR[i]-\valV \cdot \ammR[i]}{s (s-\valV)}
    \\
    & = \ammR[i] - \frac{\valV \ammR[i] (s-\valV) + \valVi (s \ammR[i]-\valV  \ammR[i])}{s (s-\valV)}
    \\
    & = \ammR[i] - \frac{\valV \ammR[i] (s-\valV) + \valVi \ammR[i] (s -\valV)}{s (s-\valV)}
    \\
    & = \ammR[i] - (\valV + \valVi) \cdot \frac{\ammR[i]}{s}
  \end{align*}
  from which the thesis follows.
\end{proofof}

\begin{proofof}{Theorem}{thm:reversibility}
  By cases on the rule used to deduce $\confG \xrightarrow{\txT} \confGi$.
  The premise that $\supply[\confG]{\tokT} = 0$ implies $\supply[\confGi]{\tokT} = 0$
  excludes the case \nrule{[Dep0]}, so we have two cases:
  \begin{itemize}

  \item \nrule{[Dep]}.
  We have that
  $\actAmmDeposit{\pmvA}{\valV[0]}{\tokT[0]}{\valV[1]}{\tokT[1]}$,
  $\confG = \walA{\tokBal} \mid \amm{\ammR[0]:\tokT[0]}{\ammR[1]:\tokT[1]} \mid \confD$,
  and:
  \begin{align*}
    \confGi 
    & =
    \walA{\tokBal - \valV[0]:\tokT[0] - \valV[1]:\tokT[1] + \valV:\tokM{\tokT[0]}{\tokT[1]}} 
    \mid 
    \amm{\ammR[0]+\valV[0]:\tokT[0]}{\ammR[1]+\valV[1]:\tokT[1]} 
    \mid 
    \confD
    \\
    & =
    \walA{\tokBali} 
    \mid 
    \amm{\ammRi[0]:\tokT[0]}{\ammRi[1]:\tokT[1]} 
    \mid 
    \confD
  \end{align*}
  where $\valV = \frac{\valV[i]}{\ammR[i]} \cdot s$, with
  $s = \supply[\confG]{\tokM{\tokT[0]}{\tokT[1]}}$.
  Let $\txT^{-1} = \actAmmRedeem{\pmvA}{\valV:\tokM{\tokT[0]}{\tokT[1]}}$.
  We have that:
  \begin{align*}
    \confGi
    & \xrightarrow{\txT^{-1}}
    \walA{\tokBali + \valVi[0]:\tokT[0] + \valVi[1]:\tokT[1] - \valV:\tokM{\tokT[0]}{\tokT[1]}} 
    \mid 
    \amm{\ammRi[0]-\valVi[0]:\tokT[0]}{\ammRi[1]-\valVi[1]:\tokT[1]} 
    \mid 
    \confD
    \; = \;
    \confGii
  \end{align*}
  where, for $i \in \setenum{0,1}$ and $s' = \supply[\confG]{\tokM{\tokT[0]}{\tokT[1]}} = s + \valV$:
  \begin{align*}
    \valVi[i] 
    & = \valV \cdot \frac{\ammRi[i]}{s'}
      = \valV \cdot \frac{\ammR[i]+\valV[i]}{s+\valV}
      = \Big( \frac{\valV[i]}{\ammR[i]} \cdot s \Big) \cdot \frac{\ammR[i]+\valV[i]}{s+\big( \frac{\valV[i]}{\ammR[i]} \cdot s \big)}
      = \frac{\valV[i] s (\ammR[i] + \valV[i])}{\ammR[i] s + \valV[i] s}
      = \valV[i]
  \end{align*}
  Since $\valV[i] = \valVi[i]$ for $i \in \setenum{0,1}$, we conclude that
  $\confGii = \confG$.

\item \nrule{[Rdm]}.
  We have that
  $\txT = \actAmmRedeem{\pmvA}{\valV:\tokM{\tokT[0]}{\tokT[1]}}$,
  $\confG = \walA{\tokBal} \mid \amm{\ammR[0]:\tokT[0]}{\ammR[1]:\tokT[1]} \mid \confD$,
  and:
  \begin{align*}
    \confGi 
    & =
    \walA{\tokBal + \valV[0]:\tokT[0] + \valV[1]:\tokT[1] - \valV:\tokM{\tokT[0]}{\tokT[1]}} 
    \mid 
    \amm{\ammR[0]-\valV[0]:\tokT[0]}{\ammR[1]-\valV[1]:\tokT[1]} 
    \mid 
    \confD
    \\
    & =
    \walA{\tokBali} 
    \mid 
    \amm{\ammRi[0]:\tokT[0]}{\ammRi[1]:\tokT[1]} 
    \mid 
    \confD
  \end{align*}
  where $\valV[i] = \valV \cdot \frac{\ammR[i]}{s}$, 
  for $i \in \setenum{0,1}$ and $s = \supply[\confG]{\tokM{\tokT[0]}{\tokT[1]}}$.
  Let $\txT^{-1} = \actAmmDeposit{\pmvA}{\valV[0]}{\tokT[0]}{\valV[1]}{\tokT[1]}$.
  We have that:
  \begin{align*}
    \confGi
    & \xrightarrow{\txT^{-1}}
    \walA{\tokBali - \valV[0]:\tokT[0] - \valV[1]:\tokT[1] + \valVi:\tokM{\tokT[0]}{\tokT[1]}} 
    \mid 
    \amm{\ammRi[0]+\valV[0]:\tokT[0]}{\ammRi[1]+\valV[1]:\tokT[1]} 
    \mid 
    \confD
    \; = \;
    \confGii
  \end{align*}
  where $\valVi = \frac{\valV[i]}{\ammRi[i]} \cdot s'$, 
  with $s' = \supply[\confGi]{\tokM{\tokT[0]}{\tokT[1]}} = s - \valV$.
  We have that:
  \begin{align*}
    \valVi 
    & = \frac{\valV[i]}{\ammRi[i]} \cdot s'
      = \frac{\valV \cdot \frac{\ammR[i]}{s}}{\ammR[i]-\valV \cdot \frac{\ammR[i]}{s}} \cdot (s - \valV)
      = \frac{\valV \cdot \ammR[i]}{s\ammR[i]-\valV\ammR[i]} \cdot (s - \valV)
      = \frac{\valV}{s-\valV} \cdot (s - \valV)
      = \valV
  \end{align*}
  Since $\valVi = \valV$, we conclude that $\confGii = \confG$.
  \qedhere
  \end{itemize}
\end{proofof}

%
%

\section{Proofs for Section~\ref{sec:swap-rate}}

\begin{proofof}{Lemma}{lem:sr-output-bound}
  The condition $\supply[\confG]{\tokM{\tokT[0]}{\tokT[1]}} > 0$
  ensures that $\confG$ contains an AMM for the pair $\tokT[0]$, $\tokT[1]$.
  The premise $\tokBal(\tokT[0]) \geq x$ ensures that $\pmvA$
  has enough units of the input token $\tokT[0]$.
  Output-boundedness implies the premise
  $x \cdot \SX{x,\ammR[0],\ammR[1]} < \ammR[1]$
  of \nrule{[Swap]}.
\end{proofof}

\begin{proofof}{Lemma}{lem:swap-gain-monotonicity}
  Straightforward by Definition~\ref{def:sr-monotonicity}
  and Lemma~\ref{lem:swap:gain}.
\end{proofof}

\begin{proofof}{Theorem}{thm:sr-additivity}
  Let $\confG = \amm{\ammR[0]:\tokT[0]}{\ammR[1]:\tokT[1]} \mid \confD$.
  We have that:
  \begin{align*}
    \confG[0] & = \amm{\ammR[0]+x_0:\tokT[0]}{\ammR[1]-y_0:\tokT[1]} \mid \confD[0]
    && y_0 = x_0 \cdot \SX{x_0,\ammR[0],\ammR[1]}
    \\
    \confG[1] & = \amm{\ammR[0]+x_0+x_1:\tokT[0]}{\ammR[1]-y_0-y_1:\tokT[1]} \mid \confD[1]
    && y_1 = x_1 \cdot \SX{x_1,\ammR[0]+x_0,\ammR[1]-y_0}
  \end{align*}
  Since $\SX{}$ is additive, we have that:
  \[
    \SX{x_0+x_1,\ammR[0],\ammR[1]}
    \; = \;
    \frac{y_0+y_1}{x_0+x_1}
  \]
  Therefore, rule~\nrule{[Swap]} gives the thesis:
  \[
    \confG 
    \xrightarrow{\actAmmSwapExact{\pmvA}{}{x_0+x_1}{\tokT[0]}{\tokT[1]}} 
    \amm{\ammR[0]+x_0+x_1:\tokT[0]}{\ammR[1]-(y_0+y_1):\tokT[1]} \mid \confD[1]
    \tag*{\qedhere}
  \]
\end{proofof}

\begin{proofof}{Lemma}{lem:swap-gain:additivity}
  Since $\SX{}$ is output-bounded, then by Lemma~\ref{lem:sr-output-bound}, 
  $\txT(x_0)$ and $\txT(x_0+x_1)$ are enabled in $\confG$, and
  $\txT(x_1)$ is enabled in $\confGi$.
  Let:
  \[
    \valSXa = \SX{x_0,\ammR[0],\ammR[1]}
    \qquad
    \valSXb = \SX{x_1,\ammR[0]+x_0,\ammR[1]-\valSXa x_0}
  \]
  By additivity of $\SX{}$ (Definition~\ref{def:sr-additivity}), we have that:
  \begin{equation}
    \label{eq:swap-gain-additivity}
    \valSXc =
    \SX{x_0+x_1,\ammR[0],\ammR[1]} = 
    \frac{\valSXa x_0 + \valSXb x_1}{x_0+x_1}
  \end{equation}
  Therefore:
  \begin{align*}
    & \hspace{-12pt}
      \gain[{\confG}]{\pmvA}{\txT(x_0+x_1)}
      -
      \gain[{\confG}]{\pmvA}{\txT(x_0)}
    \\
    & = 
      \valSXc (x_0+x_1) \exchO{\tokT[1]} - (x_0+x_1) \exchO{\tokT[0]}
      - \valSXa x_0 \exchO{\tokT[1]} + x_0 \exchO{\tokT[0]}
    && \text{(Lemma~\ref{lem:swap:gain})}
    \\
    & = 
      \big( (\valSXc (x_0 + x_1) - \valSXa x_0 \big) \exchO{\tokT[1]}
      - x_1 \exchO{\tokT[0]}
    \\
    & = 
      \big(\valSXa x_0 + \valSXb x_1 - \valSXa x_0) \big) \exchO{\tokT[1]}
      - x_1 \exchO{\tokT[0]}
    && \text{(Equation~\ref{eq:swap-gain-additivity})}
    \\
    & = 
      \valSXb x_1 \exchO{\tokT[1]} - x_1 \exchO{\tokT[0]}
    \\
    & = \gain[{\confGi}]{\pmvA}{\txT(x_1)}
    && \text{(Lemma~\ref{lem:swap:gain})} \tag*{\qedhere} 
  \end{align*}
\end{proofof}

\begin{proofof}{Theorem}{thm:sr-reversibility}
  Let 
  \(
  \confG = \amm{\ammR[0]:\tokT[0]}{\ammR[1]:\tokT[1]} \mid \confD
  \), and
  let $y = x \cdot \SX{x,\ammR[0],\ammR[1]}$.
  By the \nrule{[Swap]} rule, there exists $\confDi$ such that:
  \[
    \confGi 
    \; = \;
    \amm{\ammR[0]+x:\tokT[0]}{\ammR[1]-y:\tokT[1]} \mid \confDi
  \]
  Let $\txT^{-1} = \actAmmSwapExact{\pmvA}{}{y}{\tokT[1]}{\tokT[0]}$, and 
  let $x' = y \cdot \SX{y,\ammR[1]-y,\ammR[0]+x}$.
  For some $\confDii$, we have:
  \[
    \confGi 
    \; \xrightarrow{\txT^{-1}} \;
    \amm{\ammR[0]+x-x':\tokT[0]}{\ammR[1]-y+y:\tokT[1]}
    \mid \confDii
  \]
  By reversibility of the swap rate, we have that:
  \[
    \frac{y}{x} = \SX{x,\ammR[0],\ammR[1]}
    \implies
    \SX{y,\ammR[1]-y,\ammR[0]+x} = \frac{x}{y}
  \]
  from which we obtain that:
  \[
    x' 
    \; = \;
    y \cdot \SX{y,\ammR[1]-y,\ammR[0]+x}
    \; = \; 
    y \cdot \frac{x}{y}    
    \; = \; 
    x
  \]
  from which we obtain the thesis.
\end{proofof}

\begin{proofof}{Lemma}{lem:swap-gain:reversibility}
  Straightforward from the definition of gain and 
  from Theorem~\ref{thm:sr-reversibility}.
\end{proofof}

\begin{proofof}{Lemma}{lem:sr-homogeneity:dep-rdm}
  Let $\amm{\ammR[0]:\tokT[0]}{\ammR[1]:\tokT[1]} \in \confG$,
  $\amm{\ammRi[0]:\tokT[0]}{\ammRi[1]:\tokT[1]} \in \confGi$, 
  and let $a = \nicefrac{\ammRi[0]}{\ammR[0]}$.
  We have that:
  \begin{align*}
    \X[\confG]{\tokT[0],\tokT[1]}
    & = \lim_{x \rightarrow 0} \SX{x,\ammR[0],\ammR[1]}
    && \text{by~\Cref{eq:exchange-rate:internal}}
    \\
    & = \lim_{x \rightarrow 0} \SX{a x,a \ammR[0],a \ammR[1]}
    && \text{since $\SX{}$ is homogeneous}
    \\
    & = \lim_{x \rightarrow 0} \SX{a x,\ammRi[0],\ammRi[1]}
    && \text{by Lemma~\ref{lem:dep-rdm:const}\ref{lem:dep-rdm:const:ratio}}
    \\
    & = \X[\confGi]{\tokT[0],\tokT[1]}
    && \text{by~\Cref{eq:exchange-rate:internal}} \tag*{\qedhere} 
  \end{align*}
\end{proofof}

\begin{proofof}{Lemma}{lem:sr-reduced-slippage}
  For item~\ref{lem:sr-reduced-slippage:dep},
  let $\txT = \actAmmDeposit{\pmvA}{\valV[0]}{\tokT[0]}{\valV[1]}{\tokT[1]}$.
  By rule~\nrule{[Dep]}, 
  $\ammRi[i] = \ammR[i] + \valV[i]$ for $i \in \setenum{0,1}$,
  with $\ammR[0]\valV[1] = \ammR[1]\valV[0]$.
  By Lemma~\ref{lem:dep-rdm:const}\ref{lem:dep-rdm:const:ratio}, 
  \(
    \nicefrac{\ammR[0]+\valV[0]}{\ammR[1]+\valV[1]} = \nicefrac{\ammR[0]}{\ammR[1]}
  \).
  Then:
  \[
    \ammR[0]+\valV[0]
    = \frac{\ammR[1]+\valV[1]}{\ammR[1]} \ammR[0]
    = a \, \ammR[0]
    \qquad
    \ammR[1]+\valV[1]
    = \frac{\ammR[1]+\valV[1]}{\ammR[1]} \ammR[1]
    = a \, \ammR[1]
    \tag*{where $a = \frac{\ammR[1]+\valV[1]}{\ammR[1]}$}
  \]
  Therefore:
  \begin{align*}
    \SX{x,\ammRi[0],\ammRi[1]}
    & = \SX{x,a \ammR[0],a \ammR[1]}
    \\
    & = \SX{\tfrac{x}{a},\ammR[0],\ammR[1]}    
    && \text{(homogeneity)}
    \\
    & > \SX{x,\ammR[0],\ammR[1]}    
    && \text{(strict monotonicity, $a > 1 \implies \tfrac{x}{a}<x$)}
  \end{align*}
  The thesis
  \(
  \SL[\confG]{x,\tokT[0],\tokT[1]}
  >
  \SL[\confGi]{x,\tokT[0],\tokT[1]}
  \)
  follows from this inequality and Lemma~\ref{lem:sr-homogeneity:dep-rdm}.
  
  \noindent
  For item~\ref{lem:sr-reduced-slippage:rdm},
  let $\txT = \actAmmRedeem{\pmvA}{\valV:\tokM{\tokT[0]}{\tokT[1]}}$.
  By rule~\nrule{[Rdm]}, for $i \in \setenum{0,1}$:
  \[
    \ammRi[i] 
    \; = \;
    \ammR[i] - \valV[i]
    \; = \;
    \ammR[i] - \valV \frac{\ammR[i]}{\supply[\confG]{\tokM{\tokT[0]}{\tokT[1]}}}
    \; = \;
    a \, \ammR[i]
    \tag*{where $a = 1 - \frac{\valV}{\supply[\confG]{\tokM{\tokT[0]}{\tokT[1]}}}$}
  \]
  Therefore:
  \begin{align*}
    \SX{x,\ammRi[0],\ammRi[1]}
    & = \SX{x,a \ammR[0],a \ammR[1]}
    \\
    & = \SX{\tfrac{x}{a},\ammR[0],\ammR[1]}    
    && \text{(homogeneity)}
    \\
    & < \SX{x,\ammR[0],\ammR[1]}    
    && \text{(strict monotonicity, $a < 1 \implies \tfrac{x}{a}<x$)}
  \end{align*}

  \noindent
  The thesis
  \(
  \SL[\confG]{x,\tokT[0],\tokT[1]}
  <
  \SL[\confGi]{x,\tokT[0],\tokT[1]}
  \)
  follows from this inequality and Lemma~\ref{lem:sr-homogeneity:dep-rdm}.  
\end{proofof}

\begin{proofof}{Theorem}{lem:swap-rate:const-prod}
  For output-boundedness, 
  let $x > 0$ and $\ammR[0], \ammR[1] > 0$.
  We have that:
  \[
    \SX{x,\ammR[0],\ammR[1]}
    = \frac{\ammR[1]}{\ammR[0]+x} 
    < \frac{\ammR[1]}{x}
  \]

  \noindent
  For monotonicity,
  Let $x' \leq x$, $\ammRi[0] \leq \ammR[0]$ and $\ammR[1] \leq \ammRi[1]$.
  We have that:
  \[
    \SX{x',\ammRi[0],\ammRi[1]}
    \; = \;
    \frac{\ammRi[1]}{\ammRi[0] + x'}
    \; \geq \;
    \frac{\ammR[1]}{\ammR[0] + x}
    \; = \;
    \SX{x,\ammR[0],\ammR[1]}
  \]
  The proof for strict monotonicity is similar.

  \noindent
  For additivity,
  by Definition~\ref{def:const-prod} we have that:
  \begin{align*}
    \valSXa 
    & = \SX{x,\ammR[0],\ammR[1]} = \frac{\ammR[1]}{\ammR[0]+x}
    \\
    \valSXb 
    & = \SX{y,\ammR[0]+x,\ammR[1]-\valSXa x}
      = \frac{\ammR[1]-\valSXa x}{\ammR[0]+x+y}
      = \frac{\ammR[0]\ammR[1]}{(\ammR[0]+x)(\ammR[0]+x+y)}
  \end{align*}
  Therefore:
  \begin{align*}
    \frac{\valSXa x + \valSXb y}{x+y}
    & = \frac{1}{x+y} \Big(
      \frac{\ammR[1]x}{\ammR[0]+x} + \frac{\ammR[0]\ammR[1]y}{(\ammR[0]+x)(\ammR[0]+x+y)}
      \Big)
    \\
    & = \frac{1}{x+y}
      \frac{\ammR[0]\ammR[1]x + \ammR[1]x^2 + \ammR[1] x y + \ammR[0]\ammR[1]y}{(\ammR[0]+x)(\ammR[0]+x+y)}
    \\
    & = \frac{\ammR[1] (\ammR[0] + x) (x + y)}{(x+y) (\ammR[0]+x)(\ammR[0]+x+y)}
    \\
    & = \frac{\ammR[1]}{\ammR[0]+x+y}  
    \\
    & = \SX{x+y,\ammR[0],\ammR[1]}
  \end{align*}

  \noindent
  For reversibility,
  let $\alpha = \SX{x,\ammR[0],\ammR[1]}$.
  By Definition~\ref{def:const-prod}, we have that:
  \[
    \SX{\alpha x,\ammR[1]-\alpha x ,\ammR[0]+x}
    \; = \;
    \frac{\ammR[0]+x}{(\ammR[1]-\alpha x) + \alpha x}
    \; = \;
    \frac{\ammR[0]+x}{\ammR[1]}
    \; = \;
    \Big( \frac{\ammR[1]}{\ammR[0]+x} \Big)^{-1}
    = \frac{1}{\alpha}
  \]
  For homogeneity, we have that:
  \[
  \SX{a x, a \ammR[0],a \ammR[1]}
  \; = \;
  \frac{a \ammR[1]}{a \ammR[0] + a x}
  \; = \;
  \frac{\ammR[1]}{\ammR[0] + x}
  \; = \;
  \SX{x,\ammR[0],\ammR[1]}
  \]
  The computations of the internal exchange rate and of the slippage
  are straightforward.
\end{proofof}

\begin{proofof}{Theorem}{thm:const-mean}
  Output-boundedness, monotonicity and homogeneity are straightforward.
  For additivity,
  by Definition~\ref{def:const-mean} we have that:
  \begin{align*}
    \valSXa 
    & = \SX{x,\ammR[0],\ammR[1]}
    =
    \frac{\ammR[1]}{x} \bigg( 1 - \Big(\frac{\ammR[0]}{\ammR[0]+x} \Big)^{\frac{w_0}{w_1}} \bigg)
    \\
    \valSXb 
    & = \SX{y,\ammR[0]+x,\ammR[1]-\valSXa x}
    = 
    \frac{\ammR[1] -\valSXa x}{y} \bigg( 1 - \Big(\frac{\ammR[0]+x}{\ammR[0]+x+y} \Big)^{\frac{w_0}{w_1}} \bigg)
  \end{align*}
  Therefore:
  \begin{align*}
    \frac{\valSXa x + \valSXb y}{x+y}
    & = \frac{1}{x+y} \Bigg(
    \valSXa x 
    +
    (\ammR[1] -\valSXa x) \bigg( 1 - \Big(\frac{\ammR[0]+x}{\ammR[0]+x+y} \Big)^{\frac{w_0}{w_1}} \bigg)
    \Bigg)
    \\
    & =
    \frac{1}{x+y} \Bigg(
    \ammR[1] - \ammR[1] \Big(\frac{\ammR[0]+x}{\ammR[0]+x+y} \Big)^{\frac{w_0}{w_1}}
    +
    \ammR[1] \bigg( 1 - \Big(\frac{\ammR[0]}{\ammR[0]+x} \Big)^{\frac{w_0}{w_1}} \bigg)
    \Big(\frac{\ammR[0]+x}{\ammR[0]+x+y} \Big)^{\frac{w_0}{w_1}}
    \Bigg)
    \\
    & = \frac{1}{x+y}
    \bigg(
    \ammR[1]
    -
    \ammR[1]
    \Big(\frac{\ammR[0]}{\ammR[0]+x} \Big)^{\frac{w_0}{w_1}}
    \Big(\frac{\ammR[0]+x}{\ammR[0]+x+y} \Big)^{\frac{w_0}{w_1}}
    \bigg)
    \\
    & = \frac{\ammR[1]}{x+y}
    \bigg(
    1
    -
    \Big(\frac{\ammR[0]}{\ammR[0]+x+y} \Big)^{\frac{w_0}{w_1}}
    \bigg)
    \\
    & = \SX{x+y,\ammR[0],\ammR[1]} \tag*{\qedhere} 
  \end{align*}
\end{proofof}

%
%

\section{Proofs for Section~\ref{sec:arbitrage}}

\begin{proofof}{Lemma}{lem:SX-X}
  Assume that $\SX{x,\ammR[0],\ammR[1]} \geq \X{\tokT[0],\tokT[1]}$.
  Let $\alpha(z) = \SX{z,\ammR[1],\ammR[0]}$.
  We have that:
  \begin{align*}
    \SX{y,\ammR[1],\ammR[0]} 
    & < \lim_{z \rightarrow 0}
    \SX{z,\ammR[1],\ammR[0]} 
    && \text{(strict monotonicity)}
    \\
    & = \lim_{z \rightarrow 0}
    \frac{1}{\SX{\alpha(z) \cdot z,\ammR[0] - \alpha(z) \cdot z,\ammR[1] + z}}
    && \text{(reversibility)}
    \\
    & < \frac{1}{\SX{x,\ammR[0],\ammR[1]}}
    && \text{(strict monotonicity)}
    \\
    & \leq \frac{1}{\X{\tokT[0],\tokT[1]}}
    && \text{(hypothesis)}
    \\
    & = \X{\tokT[1],\tokT[0]}
    && \text{(def.\ of $\X{}$)} \tag*{\qedhere} 
  \end{align*}
  where in the second application of strict monotonicity, we have exploited the
  (asymptotic) inequalities $\alpha(z) \cdot z < x$
  (where $\lim_{z \rightarrow 0} \alpha(z) \cdot z = 0$ follows from the existence
  of the internal exchange rate),
  $\ammR[0] - \alpha(z) \cdot z < \ammR[0]$, and
  $\ammR[1] + z > \ammR[1]$.
\end{proofof}

\begin{proofof}{Lemma}{lem:swap:gain-0-xor-1}
  Let $y > 0$.
  Assume that $\gain[{\confG}]{\pmvA}{\txT[d](x)} > 0$.
  Then, $\txT[d](x)$ is enabled in $\confG$, and so
  by Lemma~\ref{lem:swap:gain-SX-X}, we have that
  $\SX{x,\ammR[d],\ammR[1-d]} > \X{\tokT[d],\tokT[1-d]}$.
  Then, by Lemma~\ref{lem:SX-X} it follows that
  $\SX{y,\ammR[1-d],\ammR[d]} < \X{\tokT[1-d],\tokT[d]}$.
  Since $\tokBal{\tokT[1-d]} \geq y$ and $\SX{}$ is output-bounded, then 
  Lemma~\ref{lem:sr-output-bound} implies that 
  $\txT[1-d](y)$ is enabled in $\confG$.
  By using again Lemma~\ref{lem:swap:gain-SX-X}, concluding that
  $\gain[{\confG}]{\pmvA}{\txT[1-d](y)} < 0$.
\end{proofof}

\begin{proofof}{Theorem}{thm:arbitrage}
  Let $x_0$ and $\confGi$ be as in the hypotheses, \ie:
  \[
  \confG
  \xrightarrow{\txT(x_0)}
  \confGi = \walA{\tokBali} \mid \amm{\ammR[0]+x_0:\tokT[0]}{\ammR[1]-\valSXa x_0:\tokT[1]} \mid \confD
  \qquad \text{where }
  \begin{array}{l}
    \valSXa = \SX{x_0,\ammR[0],\ammR[1]}
    \\
    \X[\confGi]{\tokT[0],\tokT[1]} = \X{\tokT[0],\tokT[1]}
  \end{array}
  \]
  We have two cases, depending on whether $x > x_0$ or $x < x_0$.
  \begin{itemize}
    
  \item If $x > x_0$, let $x_1>0$ be such that $x = x_0 + x_1$.
    Since $\SX{}$ is output-bounded and additive, 
    then by Lemma~\ref{lem:swap-gain:additivity}:
    \begin{equation}
      \label{eq:arbitrage:max:x-gt-x0}
      \gain[{\confG}]{\pmvA}{\txT(x)}
      \; = \;
      \gain[{\confG}]{\pmvA}{\txT(x_0)}
      +
      \gain[{\confGi}]{\pmvA}{\txT(x_1)}
    \end{equation}
    We have that:
    \begin{align*}
      \SX{x_1,\ammR[0]+x_0,\ammR[1]-\valSXa x_0} 
      & < \lim_{z \rightarrow 0} \SX{z,\ammR[0]+x_0,\ammR[1]-\valSXa x_0}  
      && \text{(strict monotonicity)}
      \\
      & = \X[\confGi]{\tokT[0],\tokT[1]}
      && \text{def.\ $\X[\confGi]{}$}
      \\
      & = \X{\tokT[0],\tokT[1]}
      && \text{(hypothesis)}
    \end{align*}
    Then, by Lemma~\ref{lem:swap:gain-SX-X} we obtain
    $\gain[{\confGi}]{\pmvA}{\txT(x_1)} < 0$.
    By \Cref{eq:arbitrage:max:x-gt-x0},
    we conclude that  
    $\gain[{\confG}]{\pmvA}{\txT(x)} < \gain[{\confG}]{\pmvA}{\txT(x_0)}$.

  \item If $x < x_0$, let $x_1>0$ be such that $x_0 = x + x_1$.
    Since $\SX{}$ is output-bounded, then by Lemma~\ref{lem:sr-output-bound}, 
    $\txT(x_0)$ and $\txT(x)$ are enabled in $\confG$,
    and $\txT(x_1)$ is enabled in the state $\confG[1]$ 
    reached after performing $\txT(x_1)$, \ie:
    \[
      \confG
      \xrightarrow{\txT(x)} 
      \confG[1]
      \xrightarrow{\txT(x_1)} 
      \confGi
    \]
    Since $\SX{}$ is output-bounded and additive, 
    then by Lemma~\ref{lem:swap-gain:additivity}:
    \[
      \gain[{\confG}]{\pmvA}{\txT(x_0)}
      \; = \;
      \gain[{\confG}]{\pmvA}{\txT(x)}
      +
      \gain[{\confG[1]}]{\pmvA}{\txT(x_1)}
    \]
    Since $\SX{}$ is reversible, then by Theorem~\ref{thm:sr-reversibility},
    $\txT(x_1)$ has an inverse, which has the form
    $\txT^{-1}(x_1) = \actAmmSwapExact{\pmvA}{}{y_1}{\tokT[1]}{\tokT[0]}$
    for some $y_1 > 0$.
    Then, by Lemma~\ref{lem:swap-gain:reversibility},
    $\gain[{\confG[1]}]{\pmvA}{\txT(x_1)} = - \gain[{\confGi}]{\pmvA}{\txT^{-1}(y_1)}$, 
    therefore:
    \begin{equation}
      \label{eq:arbitrage:max:x-lt-x0}
      \gain[{\confG}]{\pmvA}{\txT(x_0)}
      \; = \;
      \gain[{\confG}]{\pmvA}{\txT(x)}
      -
      \gain[{\confGi}]{\pmvA}{\txT^{-1}(y_1)}
    \end{equation}
    We have that:
    \begin{align*}
      \SX{y_1,\ammR[1]-\valSXa x_0,\ammR[0]+x_0} 
      & < \lim_{z \rightarrow 0} \SX{z,\ammR[1]-\valSXa x_0,\ammR[0]+x_0}  
      && \text{(strict monotonicity)}
      \\
      & = \X[\confGi]{\tokT[1],\tokT[0]}
      && \text{def.\ $\X[\confGi]{}$}      
      \\
      & = \frac{1}{\X[\confGi]{\tokT[0],\tokT[1]}}
      && \text{(Equation~\eqref{eq:sr-reversibility:internal-exchange-rate})}
      \\
      & = \frac{1}{\X{\tokT[0],\tokT[1]}}
      && \text{(hypothesis)}
      \\
      & = \X{\tokT[1],\tokT[0]}
      && \text{(def.~$\X{}$)}
    \end{align*}
    Then, by Lemma~\ref{lem:swap:gain-SX-X} we obtain
    $\gain[{\confGi}]{\pmvA}{\txT^{-1}(y_1)} < 0$.
    By Equation~\eqref{eq:arbitrage:max:x-lt-x0},
    we conclude that  
    $\gain[{\confG}]{\pmvA}{\txT(x)} < \gain[{\confG}]{\pmvA}{\txT(x_0)}$.
  \end{itemize}

  For uniqueness, by contradiction assume that there exists $x_1 \neq x_0$
  satisfying~\Cref{eq:arbitrage:max:x0}.
  Then, it should be
  $\gain[{\confG}]{\pmvA}{\txT(x_1)} > \gain[{\confG}]{\pmvA}{\txT(x_0)}$
  --- contradiction.
\end{proofof}

\begin{proofof}{Lemma}{lem:arbitrage:const-prod}
  Let $\confG \xrightarrow{\txT} \confGi = \walA{\tokBali} \mid \amm{\ammR[0]+x_0:\tokT[0]}{\ammR[1]-x_0 \cdot \SX{x_0,\ammR[0],\ammR[1]}:\tokT[1]}$.
  We have that:
  \begin{align*}
    \X[\confGi]{\tokT[0],\tokT[1]}
    & = 
    \frac{\ammR[1] - x_0 \cdot \SX{x_0,\ammR[0],\ammR[1]}}{\ammR[0]+x_0}
    && \text{by Theorem~\ref{lem:swap-rate:const-prod}}
    \\
    & =
      \frac
      {\ammR[1] - x_0 \cdot \frac{\ammR[1]}{\ammR[0]+x_0}}
      {\ammR[0]+x_0}
    && \text{by Definition~\ref{def:const-prod}}
    \\
    & = \frac{\ammR[0]\ammR[1]}{(\ammR[0] + x_0)^2}
    \\
    & = \frac{\ammR[0]\ammR[1]}{\frac{\exchO{\tokT[1]}}{\exchO{\tokT[0]}} \ammR[0] \ammR[1]}
    && \text{by~\Cref{eq:arbitrage:const-prod}}
    \\
    & = \X{\tokT[0],\tokT[1]}
    && \text{by~\Cref{eq:exchange-rate}}
  \end{align*}
  The thesis follows from Theorem~\ref{thm:arbitrage}.
\end{proofof}

\begin{proofof}{Theorem}{thm:swap-after-dep}
  Let:
  \[
  \confG = \walA{\tokBal} \mid \amm{\ammR[0]:\tokT[0]}{\ammR[1]:\tokT[1]} \mid \confD
  \; \xrightarrow{\txT[\ammDepositOp]} \;
  \confGi = \walA{\tokBali} \mid \amm{\ammRi[0]:\tokT[0]}{\ammRi[1]:\tokT[1]} \mid \confDi
  \]
  The hypothesis
  $\txWal{\txT[\ammSwapOp]} = \pmvA \neq \txWal{\txT[\ammRedeemOp]}$
  means that the user who performs the deposit is \emph{not} $\pmvA$,
  hence the deposit does not affect the number of minted tokens
  in $\pmvA$'s wallet.
  Then:
  \begin{align*}
    & \hspace{-12pt} 
      \gain[\confG]{\pmvA}{\txT[\ammDepositOp]\txT[\ammSwapOp]}
    \\
    & = \gain[\confGi]{\pmvA}{\txT[\ammSwapOp]}
    \\
    & = x \cdot \big(
      \SX{x,\ammRi[0],\ammRi[1]} \, \exchO{\tokT[1]}
      -
      \exchO{\tokT[0]}
      \big)
      \cdot
      \Big(
      1 - \frac{\tokBali\tokM{\tokT[0]}{\tokT[1]}}{\supply[\confGi]{\tokM{\tokT[0]}{\tokT[1]}}}
      \Big)
    && \text{(Lemma~\ref{lem:swap:gain})}
    \\
    & > x \cdot \big(
      \SX{x,\ammR[0],\ammR[1]} \, \exchO{\tokT[1]}
      -
      \exchO{\tokT[0]}
      \big)
      \cdot
      \Big(
      1 - \frac{\tokBali\tokM{\tokT[0]}{\tokT[1]}}{\supply[\confGi]{\tokM{\tokT[0]}{\tokT[1]}}}
      \Big)
    && \text{(Lemma~\ref{lem:sr-reduced-slippage}\ref{lem:sr-reduced-slippage:dep})}
    \\
    & = x \cdot \big(
      \SX{x,\ammR[0],\ammR[1]} \, \exchO{\tokT[1]}
      -
      \exchO{\tokT[0]}
      \big)
      \cdot
      \Big(
      1 - \frac{\tokBal\tokM{\tokT[0]}{\tokT[1]}}{\supply[\confGi]{\tokM{\tokT[0]}{\tokT[1]}}}
      \Big)
    && \text{$(\tokBali\tokM{\tokT[0]}{\tokT[1]} = \tokBal\tokM{\tokT[0]}{\tokT[1]})$}
    \\
    & > x \cdot \big(
      \SX{x,\ammR[0],\ammR[1]} \, \exchO{\tokT[1]}
      -
      \exchO{\tokT[0]}
      \big)
      \cdot
      \Big(
      1 - \frac{\tokBal\tokM{\tokT[0]}{\tokT[1]}}{\supply[\confG]{\tokM{\tokT[0]}{\tokT[1]}}}
      \Big)
    && \text{$(\supply[\confGi]{\tokM{\tokT[0]}{\tokT[1]}} > \supply[\confG]{\tokM{\tokT[0]}{\tokT[1]}})$}
    \\    
    & = \gain[\confG]{\pmvA}{\txT[\ammSwapOp]} \tag*{\qedhere} 
  \end{align*}
\end{proofof}

\begin{proofof}{Theorem}{thm:dep-arbitrage}
  Let $\confG$ and $\confG[d]$ be as in the statement.
  By rule~\nrule{[Dep]}, 
  $\ammRi[i] = \ammR[i] + \valV[i]$ for $i \in \setenum{0,1}$.
  By Lemma~\ref{lem:dep-rdm:const}\ref{lem:dep-rdm:const:ratio}, we have that
  $\nicefrac{\ammR[0]+\valV[0]}{\ammR[1]+\valV[1]} = \nicefrac{\ammR[0]}{\ammR[1]}$.
  Then:
  \[
    \ammRi[0]
    = \ammR[0]+\valV[0]
    = \frac{\ammR[1]+\valV[1]}{\ammR[1]} \ammR[0]
    = a \, \ammR[0]
    \qquad
    \ammRi[1]
    = \ammR[1]+\valV[1]
    = \frac{\ammR[1]+\valV[1]}{\ammR[1]} \ammR[1]
    = a \, \ammR[1]
    \qquad
    \text{where $a = \frac{\ammR[1]+\valV[1]}{\ammR[1]}$}
  \]
  For item~\eqref{thm:dep-arbitrage:swap}, 
  assume that $\bcB = \actAmmSwapExact{\pmvA}{}{x}{\tokT[0]}{\tokT[1]}$
  is a solution to the arbitrage game in $\confG$.
  By Theorem~\ref{thm:arbitrage}, it must be:
  \begin{equation}
    \label{eq:dep-arbitrage}
    \X[{\confG[s]}]{\tokT[0],\tokT[1]}
    =
    \X{\tokT[0],\tokT[1]}
    \qquad
    \text{ where }
    \confG \xrightarrow{\bcB} \confG[s]
  \end{equation}
  Let $x' = a x$,  
  let $\txTi = \actAmmSwapExact{\pmvA}{}{x'}{\tokT[0]}{\tokT[1]}$,
  and let $\confG[d] \xrightarrow{\txTi} \confG[ds]$.
  We have that:
  \begin{align*}
    & \hspace{-12pt}
    \X[{\confG[ds]}]{\tokT[0],\tokT[1]}
    \\
    & = \lim_{z \rightarrow 0}
    \SX{z,\ammRi[0]+x',\ammRi[1] - x' \cdot \SX{x',\ammRi[0],\ammRi[1]}}
    \\
    & = \lim_{z \rightarrow 0}
    \SX{z,a \ammR[0]+a x,a \ammR[1] - a x \cdot \SX{a x,a \ammR[0],a \ammR[1]}}
    \\
    & = \lim_{z \rightarrow 0}
    \SX{z,a \ammR[0]+a x,a \ammR[1] - a x \cdot \SX{x,\ammR[0],\ammR[1]}}
    && \text{(homogeneity)}
    \\
    & = \lim_{z \rightarrow 0}
    \SX{z,\ammR[0]+x,\ammR[1] - x \cdot \SX{x,\ammR[0],\ammR[1]}}
    && \text{(homogeneity)}
    \\
    & = \X[{\confG[s]}]{\tokT[0],\tokT[1]}
    && \text{(def. $\X[{\confG[s]}]{}$)}
    \\
    & = \X{\tokT[0],\tokT[1]}
    && \text{(Equation~\eqref{eq:dep-arbitrage})}
  \end{align*}
  Therefore, Theorem~\ref{thm:arbitrage} implies that
  $\txTi$ is a solution to the arbitrage game in $\confG[d]$.
  We compute the gain of $\txTi$ in $\confG[d]$ as follows:
  \begin{align*}
    \gain[{\confG[d]}]{\pmvA}{\txTi}
    & =
    x' \cdot \big(
    \SX{x',\ammRi[0],\ammRi[1]} \, \exchO{\tokT[1]}
    -
    \exchO{\tokT[0]}
    \big)
    \\
    & =
    a x \cdot \big(
    \SX{a x,a \ammR[0],a \ammR[1]} \, \exchO{\tokT[1]}
    -
    \exchO{\tokT[0]}
    \big)
    \\
    & = a x \cdot \big(
    \SX{x,\ammR[0],\ammR[1]} \, \exchO{\tokT[1]}
    -
    \exchO{\tokT[0]}
    \big)
    && \text{(homogeneity)}
    \\
    & = a \gain[\confG]{\pmvA}{\txT}
  \end{align*}

  \noindent
  For item~\eqref{thm:dep-arbitrage:emptyseq}, 
  assume that $\emptyseq$ is a solution to the arbitrage game in $\confG$.
  By contradiction, assume that
  $\bcB[d] = \actAmmSwapExact{\pmvA}{}{x'}{\tokT[0]}{\tokT[1]}$
  is a solution in $\confG[d]$.
  By Theorem~\ref{thm:arbitrage}, it must be:
  \begin{equation}
    \label{eq:dep-arbitrage:emptyseq}
    \X[{\confG[ds]}]{\tokT[0],\tokT[1]}
    =
    \X{\tokT[0],\tokT[1]}
  \end{equation}
  The chain of equations above shows that
  $\X[{\confG[ds]}]{\tokT[0],\tokT[1]} = \X[{\confG[s]}]{\tokT[0],\tokT[1]}$.
  By Equation~\eqref{eq:dep-arbitrage:emptyseq}, this implies that
  $\X[{\confG[ds]}]{\tokT[0],\tokT[1]} = \X{\tokT[0],\tokT[1]}$.
  Hence, by Theorem~\ref{thm:arbitrage}, $\emptyseq$ 
  cannot be a solution to the arbitrage game in $\confG$
  --- contradiction.
\end{proofof}

\begin{proofof}{Theorem}{thm:swap-after-rdm}
  Let:
  \[
  \confG = \walA{\tokBal} \mid \amm{\ammR[0]:\tokT[0]}{\ammR[1]:\tokT[1]} \mid \confD
  \; \xrightarrow{\txT[\ammRedeemOp]} \;
  \confGi = \walA{\tokBali} \mid \amm{\ammRi[0]:\tokT[0]}{\ammRi[1]:\tokT[1]} \mid \confDi
  \]
  The hypothesis
  $\txWal{\txT[\ammSwapOp]} = \pmvA \neq \txWal{\txT[\ammRedeemOp]}$
  means that the user who performs the redeem is \emph{not} $\pmvA$,
  hence the redeem does not affect the number of minted tokens
  in $\pmvA$'s wallet.  
  Then:
  \begin{align*}
    & \hspace{-12pt} 
      \gain[\confG]{\pmvA}{\txT[\ammRedeemOp]\txT[\ammSwapOp]}
    \\
    & = \gain[\confGi]{\pmvA}{\txT[\ammSwapOp]}
    \\
    & = x \cdot \big(
      \SX{x,\ammRi[0],\ammRi[1]} \, \exchO{\tokT[1]}
      -
      \exchO{\tokT[0]}
      \big)
      \cdot
      \Big(
      1 - \frac{\tokBali\tokM{\tokT[0]}{\tokT[1]}}{\supply[\confGi]{\tokM{\tokT[0]}{\tokT[1]}}}
      \Big)
    && \text{(Lemma~\ref{lem:swap:gain})}
    \\
    & < x \cdot \big(
      \SX{x,\ammR[0],\ammR[1]} \, \exchO{\tokT[1]}
      -
      \exchO{\tokT[0]}
      \big)
      \cdot
      \Big(
      1 - \frac{\tokBali\tokM{\tokT[0]}{\tokT[1]}}{\supply[\confGi]{\tokM{\tokT[0]}{\tokT[1]}}}
      \Big)
    && \text{(Lemma~\ref{lem:sr-reduced-slippage}\ref{lem:sr-reduced-slippage:rdm})}
    \\
    & = x \cdot \big(
      \SX{x,\ammR[0],\ammR[1]} \, \exchO{\tokT[1]}
      -
      \exchO{\tokT[0]}
      \big)
      \cdot
      \Big(
      1 - \frac{\tokBal\tokM{\tokT[0]}{\tokT[1]}}{\supply[\confGi]{\tokM{\tokT[0]}{\tokT[1]}}}
      \Big)
    && \text{$(\tokBali\tokM{\tokT[0]}{\tokT[1]} = \tokBal\tokM{\tokT[0]}{\tokT[1]})$}
    \\
    & < x \cdot \big(
      \SX{x,\ammR[0],\ammR[1]} \, \exchO{\tokT[1]}
      -
      \exchO{\tokT[0]}
      \big)
      \cdot
      \Big(
      1 - \frac{\tokBal\tokM{\tokT[0]}{\tokT[1]}}{\supply[\confG]{\tokM{\tokT[0]}{\tokT[1]}}}
      \Big)
    && \text{$(\supply[\confGi]{\tokM{\tokT[0]}{\tokT[1]}} < \supply[\confG]{\tokM{\tokT[0]}{\tokT[1]}})$}
    \\    
    & = \gain[\confG]{\pmvA}{\txT[\ammSwapOp]} \tag*{\qedhere} 
  \end{align*}
\end{proofof}

\begin{proofof}{Theorem}{thm:rdm-arbitrage}
  Let $\confG$ and $\confG[d]$ be as in the statement.
  By rule~\nrule{[Rdm]}, it must be, for $i \in \setenum{0,1}$:
  \[
    \ammRi[i] 
    \; = \;
    \ammR[i] - \valV[i]
    \; = \;
    \ammR[i] - \valV \RX{i}{\confG}{\tokT[0]}{\tokT[1]}
    \; = \;
    \ammR[i] - \valV \frac{\ammR[i]}{\supply[\confG]{\tokM{\tokT[0]}{\tokT[1]}}}
    \; = \;
    a \ammR[i] 
    \qquad
    \text{where } a = 1 - \frac{\valV}{\supply[\confG]{\tokM{\tokT[0]}{\tokT[1]}}}
  \]
  The rest of the proof follows exactly
  that of Theorem~\ref{thm:dep-arbitrage}.
\end{proofof}

\end{document}